\documentclass[table,svgnames,dvipsnames]{article}

\usepackage{soul}
\usepackage{url}
\usepackage[hidelinks]{hyperref} 
\usepackage[small]{caption}
\usepackage{graphicx}
\usepackage{booktabs}
\usepackage{amsmath,multirow}
\usepackage{amsthm,thm-restate,thmtools}
\usepackage{amssymb}
\usepackage{enumerate}
\usepackage{xcolor,xfrac}
\usepackage[english]{babel}
\usepackage{charter}
\usepackage{authblk}
\usepackage{framed}
\usepackage[normalem]{ulem}  
\usepackage{tikz}
\usetikzlibrary{shapes.arrows}
\usetikzlibrary{shapes.geometric} 
\usepackage{amsfonts} 
\usepackage[T1]{fontenc}
\usepackage[utf8]{inputenc}
\usepackage[top=1 in,bottom=1in, left=1 in, right=1 in]{geometry}  
\usepackage{enumitem}
\usepackage{tabularx}
\hypersetup{
     colorlinks   = true,
     linkcolor    = red, 
     urlcolor     = teal, 
	 citecolor    = blue 
}
\usepackage{mathtools}
\usepackage{nicefrac} 
\usepackage{pifont}
\usepackage{bm}
\usepackage{dsfont}
\usepackage{algorithm}
\usepackage{graphicx}
\usepackage[noend]{algpseudocode}
\usepackage{subcaption}
\usepackage{gensymb}
\usepackage{pgfplots}
\usepackage{bbm}
\usepackage{comment}
\usepackage{xr}
\usepackage[sort,numbers]{natbib}

\usepackage{tkz-graph}

\newtheorem{innercustomthm}{Theorem}
\newenvironment{customthm}[1]
  {\renewcommand\theinnercustomthm{#1}\innercustomthm}
  {\endinnercustomthm}
  
\newtheorem{innercustomlemma}{Lemma}
\newenvironment{customlemma}[1]
  {\renewcommand\theinnercustomlemma{#1}\innercustomlemma}
  {\endinnercustomlemma}

  \newtheorem{innercustomcoro}{Corollary}
\newenvironment{customcoro}[1]
  {\renewcommand\theinnercustomcoro{#1}\innercustomcoro}
  {\endinnercustomcoro}
  
\newtheorem{theorem}{Theorem}
\newtheorem{lemma}{Lemma}
\newtheorem{corollary}{Corollary}

\newtheorem{remark}{Remark}
\newtheorem{assumption}{Assumption}
\newtheorem{observation}{Observation}

\DeclareMathOperator*{\gap}{Gap}
\DeclareMathOperator*{\out}{out}
\DeclareMathOperator*{\opt}{opt}
\DeclareMathOperator*{\alg}{ALG}
\DeclareMathOperator*{\size}{size}
\DeclareMathOperator*{\act}{act}
\DeclareMathOperator*{\tr}{trace}

\title{Offline Learning of Nash Stable Coalition Structures \\ with Possibly Overlapping Coalitions}

\date{\empty}

\author{Saar Cohen\\
Department of Computer Science, Bar Ilan University, Ramat Gan, Israel\\
Department of Computer Science, University of Oxford, Oxford, United Kingdom\\
saar30@gmail.com}

\begin{document}

\maketitle

\begin{abstract}
Coalition formation concerns strategic collaborations of {selfish} agents that form coalitions based on their preferences. It is often assumed that coalitions are disjoint and preferences are fully known, which may not hold in practice. In this paper, we thus present a new model of coalition formation with \textit{possibly overlapping} coalitions under \textit{partial information}, where selfish agents may be part of \textit{multiple} coalitions simultaneously and their full preferences are initially unknown. Instead, information regarding past interactions and associated utility feedbacks is stored in a fixed offline dataset, and we aim to efficiently infer the agents' preferences from this dataset. We analyze the impact of diverse dataset information constraints by studying two types of utility feedbacks that can be stored in the dataset: \textit{semi-bandit} (agent-level) and \textit{bandit} (coalition-level) utility feedbacks. For both feedback models, we identify assumptions under which the dataset covers sufficient information for an offline learning algorithm to infer preferences and use them to recover a partition that is (approximately) \textit{Nash stable}, in which no agent can improve her utility by unilaterally deviating. Our additional goal is devising algorithms with \textit{low sample complexity}, requiring only a small dataset to obtain a desired approximation to Nash stability. Under semi-bandit feedback, we provide a sample-efficient algorithm proven to obtain an approximately Nash stable partition under a {\textit{sufficient} and \textit{necessary}} assumption on the information covered by the dataset. However, under bandit feedback, we show that {only under a stricter assumption is sufficient for sample-efficient learning}. Still, in multiple cases, our algorithms' sample complexity bounds have \textit{\textbf{optimality}} guarantees up to logarithmic factors. {Finally, extensive experiments show that our algorithm converges to a low approximation level to Nash stability across diverse settings.}
\end{abstract}

\section{Introduction}
\label{sec:intro}
In a large consulting firm, managers must assign consultants to upcoming client projects across various domains (e.g., {finance, logistics}). A consultant may work on multiple projects, one per domain. The compatibility between consultants depends on the domain{, e.g.}, two employees may collaborate efficiently on logistics, but clash in finance due to differing approaches. However, the managers do not have full access to employees’ preferences over team compositions. Testing new team structures through trial projects is also infeasible due to the financial and organizational costs of reassignments, client-facing risks, and time constraints. Instead, managers rely on a fixed dataset of historical project outcomes and post-project peer evaluations. Based only on this limited, pre-collected information, they aim to assign consultants to teams that align with their true but unknown preferences, ensuring no one would prefer to switch groups unilaterally. Such scenarios and many other real-world cases exemplify \textit{coalition formation}, where \textit{agents} perform activities in \textit{coalitions} rather than on their own, and the challenge is forming coalitions that satisfy some desired criteria.


A popular model for studying coalition formation is that of \textit{hedonic games} \cite{dreze1980hedonic}, whose outcome is a set of disjoint coalitions (hereafter, \textit{partition}). The desirability of partitions is often evaluated in terms of various concepts of \textit{single-agent stability} based on agents' preferences \cite{aziz_savani_moulin_2016,bullinger2025stability}, where no agent benefits from leaving her current coalition to join another one on her own. A common such notion is \textit{Nash stability}, where no agent can increase her utility by unilaterally deviating (e.g., no consultant can improve by changing teams unilaterally). In the hedonic games literature, agents express preferences only for coalitions they are part of while disregarding inter-coalitional relationships, and these preferences are typically assumed to be fully known when computing stable partitions. Yet, both assumptions may not hold in many real-life scenarios, as in our consulting firm example.

In this paper, we thus introduce and study a \textit{new} framework of coalition formation with \textit{possibly overlapping} coalitions in \textit{partial information} settings, where selfish agents may {join} \textit{several} coalitions concurrently{, following either \textit{mixed} or \textit{pure} strategies, while} their full preferences are initially unknown. Instead, we only have partial knowledge about past interactions and corresponding utility feedbacks, which is stored in a fixed offline dataset. The dataset is collected in advance, without further interactions within the environment, reflecting real-life settings where active interactions may be costly and risky, like in our consulting firm example. Our goal is to maximally exploit the available dataset to efficiently infer agents' preferences. We exhibit our results for \textit{additively separable} and \textit{symmetric} preferences \cite{bogomolnaia2002stability}, where any pair of agents assigns the same cardinal utility to one another, and an agent's utility for her chosen coalitions is the \textit{sum} of her utilities from their members. However, agents' mutual utilities for one another may vary across different coalitions, reflecting real-world cases such as our consulting firm example, where a consultant may benefit more from another peer in one domain than in another. {Symmetric preferences are realistic in reciprocal interactions like friendships, as widely studied in friend-oriented hedonic games \cite{gairing2019computing,dimitrov2006simple}. Further, our results for \textit{mixed} strategies readily extend to \textit{asymmetric} preferences, where agents may value each other differently, but extending our analysis for \textit{pure} strategies may be generally infeasible (see Remark \ref{remark:asymmetric}).}

Many realistic scenarios differ in the type and granularity of utility information available in the dataset. We thus explore the effect of different information constraints by studying two forms of utility feedbacks that can be stored in the dataset. Specifically, we examine \textit{semi-bandit} utility feedbacks, where \textit{agent-level} utilities are available for each agent's interactions with other members of her chosen coalitions. We also analyze \textit{bandit} feedback settings with less granular utility information, where we can only observe \textit{coalition-level} utility feedbacks about each agent's overall utility from her chosen coalitions.

\textbf{Contributions.} In both semi-bandit and bandit feedback settings, we characterize assumptions under which the dataset covers sufficient information for an offline learning algorithm to deduce agents' preferences and use them to construct an outcome that is (approximately) Nash stable. Under those assumptions, we develop algorithms with \textit{low sample complexity}, requiring only a small dataset to reach a desired approximation to Nash stability. In semi-bandit feedback settings, we design a sample-efficient offline learning algorithm that attains an approximately Nash stable {outcome} under a minimal assumption on the information covered by the dataset. {Our assumption is \textit{sufficient} and \textit{necessary}, as we prove that no weaker assumption enables efficient learning of an approximately Nash stable outcome, regardless of the dataset size.} Intuitively, we require that the dataset reflects the partitions that may form through unilateral deviations. For instance, in our consulting firm example, if a unilateral deviation would result in a finance team of a certain size, then the dataset should include some past project outcomes with a finance team of that same size. However, under bandit feedback, we show that {only a stricter assumption is \textit{sufficient} for sample-efficient learning of an approximately Nash stable outcome}. Specifically, for any agent and any unilateral deviation of that agent from some Nash stable outcome, {the stricter assumption} requires that the dataset is at least as informative as a sufficiently large dataset generated according to the outcome of that deviation. In many cases, our algorithms' sample complexity bounds have \textit{\textbf{optimality}} guarantees up to logarithmic factors. {Finally, extensive experiments confirm that our algorithms consistently reaches a low approximation to Nash stability in a variety of settings.}

\section{Related Works}
Hedonic games have been presented by \citet{dreze1980hedonic}, and later expanded to the study of various notions of stability, fairness, and optimality (see, e.g., \cite{aziz_savani_moulin_2016,woeginger2013core}). Highly related to our work are \textit{additively separable hedonic games} (ASHGs) with symmetric preferences \cite{bogomolnaia2002stability}, where many works evaluate the system by means of Nash stability \cite{aziz2011stable,banerjee2001core,ballester2004np}. {Particularly, 
\citet{bogomolnaia2002stability} proved that Nash stable partitions may not exist in \textbf{general} ASHGs, while \citet{sung2010computational} showed that checking if an instance admits such a partition is NP-complete in the strong sense. In contrast,} ASHGs with {\textit{\textbf{symmetric}}} preferences admit a Nash stable partition due to a potential function argument \cite{bogomolnaia2002stability}, but computing such partitions is PLS-complete \cite{gairing2019computing}. However, while hedonic games do not allow agents to be part of several coalitions, our model captures realistic cases where agents can join \textit{multiple} coalitions, which thus may overlap. 

\citet{shehory1996formation,shehory1998methods} introduced the first model of overlapping coalition formation for handling task allocation, which was later followed by additional works on task-oriented applications (see, e.g., \cite{dang2006overlapping,lin2007multi,zhang2010searching,zhang2017task}). However, this line of research aims to find (approximately) optimal coalition structures, while we analyze the system by means of stability. Though the \textit{group} stability notion of the core has been examined in cooperative games with overlapping coalitions (see, e.g., \cite{chalkiadakis2010cooperative,zick2012stability,elkind2013computational,zick2019cooperative}), we explore Nash stability, which is based on \textit{single-agent} deviations. Further, all above works on overlapping coalition formation typically allow arbitrary monetary transfers, i.e., the payoff or resources of a coalition can be distributed arbitrarily among its members. Conversely, monetary transfers are \textit{unavailable} in our context, as we consider games with \textit{non-transferable} utilities.


{However}, the above works on hedonic games and cooperative games with overlapping coalitions unrealistically require that the agents' preferences are fully known when computing stable partitions. To tackle this issue, several studies have examined PAC learning in those domains \cite{sliwinski2017learning,igarashi2019forming,fioravanti2023pac,balcan2015learning}, using samples to learn agents' preferences and a \textit{core-stable} partition, where no subset of agents can improve their utility by regrouping into a new coalition. Yet, while they focus on the \textit{group} deviations, we study \textit{single-agent} deviations by analyzing the system in terms of Nash stability. Recently, \citet{cohen2024online,cohen2025online,cohen2025decentralized,cohen2023online,cohen2024onlinefriends,cohen2025egalitarianism} and \citet{cohen2026delayed} proposed \textit{online} and \textit{online learning} variants of coalition formation. In online learning of coalition structures, agents initially lack preference knowledge and form coalitions based on preferences learned through active interactions. While \citet{cohen2025online,cohen2025decentralized} analyze Nash stability in \textit{online} settings, we assess Nash stability in \textit{offline} scenarios where active interactions are costly or risky. Instead, agents’ preferences are inferred from a fixed, pre-collected dataset without further interactions.

Our work is also closely related to offline reinforcement learning (RL), aiming to learn an optimal policy from a dataset collected a priori without further interactions with the environment. A key challenge in offline RL is the insufficient coverage of the dataset \cite{wang2021statistical}, arising from the lack of continued exploration \cite{szepesvari2022algorithms}. 
To address this challenge, existing studies presented various assumptions on the sufficient coverage of the dataset in both single-agent settings \cite{szepesvari2005finite,rashidinejad2021bridging,yin2021towards,jin2021pessimism} and multi-agent settings \cite{sidford2020solving,zhang2023model,zhang2021finite}. Recently, minimal assumptions for offline zero- and general-sum games have been identified \cite{cui2022offline,cui2022provably,zhong22pessimistic}, together with algorithms for learning a Nash equilibrium. However, their sample complexity scales with the number of actions each agent can take, which may be \textit{exponential} in the number of agents in our setting. In contrast, our proposed assumptions allow for algorithms that remove this exponential dependency, while attaining sample complexity with \textit{\textbf{optimality}} guarantees (up to logarithmic factors).

\section{Preliminaries}


We consider a \textit{\textbf{possibly overlapping coalition formation}} (POCF) game $G=(\mathcal{N}, \{\mathcal{A}_i\}_{i\in \mathcal{N}})$, defined as follows. We are given a finite set $\mathcal{N} = \{1,\dots,n\}$ of $n$ selfish agents with \textit{unknown} preferences. Hereafter, we denote $[k] := \{1,\dots,k\}$ for $k \in \mathbb{N}$ and $[0] = \{0\}$. We focus on realistic scenarios where the number of coalitions may be constrained. For an integer $k\geq 1$, each agent can join one of $k$ \textit{candidate} coalitions. In our consulting firm example, this can be thought of as if each consultant picks which rooms to enter among $k$ rooms. Thus, in our setting of possibly overlapping coalitions, the \textit{\textbf{action space}} of each agent $i$ is denoted by $\mathcal{A}_i \subseteq \mathcal{P}([k]) \setminus \{\emptyset\}$, where $\mathcal{P}([k])$ is the power set of $[k]$ and each action $a_i \in \mathcal{A}_i$ is a subset of at most $k$ candidate coalitions that agent $i$ can join, where $a_i\neq\emptyset$ (i.e., agent $i$ always decides to join some coalition). Note that, if $|a_i|\leq 1$ for any agent $i$ and any action $a_i\in \mathcal{A}_i$, then each agent can participate in at most one coalition, in which case coalitions are \textit{disjoint}; otherwise, coalitions may \textit{overlap}.

Each agent $i$ can join certain coalitions among $k$ candidate ones following a \textit{\textbf{mixed strategy}} $\varphi_i \in \Delta(\mathcal{A}_i)$, where $\Delta(\mathcal{A}_i)$ is the probability simplex over $\mathcal{A}_i$, i.e., for any $a_i \in \mathcal{A}_i$, agent $i$ joins the candidate coalitions in $a_i$ with probability $\varphi_i(a_i) \in [0,1]$. Letting $\mathcal{A} = \prod_{i=1}^n \mathcal{A}_i$ be the \textit{joint action space}, each agent $i$ can then sample an action $a_i \in \mathcal{A}_i$ from $\varphi_i$ independently from other agents, forming a \textit{joint action} $\mathbf{a} = (a_i)_{i \in \mathcal{N}}$. Similarly, let $\bm{\varphi} = (\varphi_i)_{i \in \mathcal{N}}$ be the agents' \textit{joint mixed strategy}. The constructed joint action $\mathbf{a}$ induces a partition of the agents $\pi^{\mathbf{a}} = (C_\ell^{\mathbf{a}})_{\ell \in [k]}$ with possibly overlapping coalitions, where, for any $\ell \in [k]$, $C_\ell^{\mathbf{a}}$ is the set of agents joining the $\ell$th candidate coalition, i.e., $C_\ell^{\mathbf{a}} = \{i \in \mathcal{N}: \ell\in a_i \}$. For any $\ell \in [k]$, if no agent joins the $\ell$th candidate coalition (i.e., $\ell\notin a_i$ for any agent $i$), then this coalition is empty. Thereby, we denote by $|\pi^{\mathbf{a}}|$ the number of non-empty coalitions in $\pi^{\mathbf{a}}$. We also denote the coalitions in $\pi^{\mathbf{a}}$ containing agent $i$ as $\pi^{\mathbf{a}}(i)$. 

Afterwards, we can derive each agent's utility from her chosen strategy, determined by aggregating her utilities resulting from interactions with other agents. We focus on POCF games with \textit{additively separable} and \textit{symmetric}{\footnote{{As noted in Section \ref{sec:intro}, our results for \textit{mixed} strategies naturally extend to \textit{asymmetric} preferences, yet this is generally infeasible for \textit{pure} strategies (see Remark \ref{remark:asymmetric}).}}} valuations in each coalition, where the utility that any pair of agents derives from their interaction within that coalition is equal, indicating the intensity by which they prefer each other to another agent. However, their mutual utilities from each other may differ across different coalitions. This captures real-life scenarios, such as our consulting firm example, where a consultant may benefit more from another one in one domain than in another. An agent's utility from her chosen coalitions is then the sum of her utilities from other members of those coalitions. As common in the literature (see, e.g., \cite{flammini2021strategyproof}), we assume that agents' utilities are within $[-1,1]$. Recall that preferences are \textit{unknown} and even the agents themselves may not be aware of them. Thus, supposing a joint action $\mathbf{a}$ is chosen by all agents, the uncertainty about the mutual valuations resulting from their interactions within the $\ell$-th candidate coalition is captured by an \textit{unknown} and \textit{fixed} distribution $\mathcal{D}_{i,j}^\ell(\cdot|\mathbf{a})$ over $[-1,1]$ with mean $d_{i,j}^\ell$ for any pair of distinct agents $i,j$, which we aim to learn. The \textit{\textbf{mutual utility}} $v_{i,j}^\ell$ of agents $i,j$ that results from their interactions in the $\ell$-th candidate coalition according to a joint action $\mathbf{a}$ is then independently drawn from $\mathcal{D}_{i,j}^\ell(\cdot|\mathbf{a})$. We use the convention that $v_{i,i}^\ell= d_{i,i}^\ell =0$ for any agent $i$. Thus, for any joint action $\mathbf{a}$ sampled from a joint mixed strategy $\bm{\varphi}$, if agent $i$ decides to join some candidate coalitions in $a_i$, then her utility from the induced partition $\pi^{\mathbf{a}} = (C_\ell^{\mathbf{a}})_{\ell \in [k]}$ is $v_i(\mathbf{a}) = \sum_{\ell\in a_i}\sum_{i \neq j \in C_\ell^{\mathbf{a}}} v_{i,j}^\ell$, whose mean is $d_i(\mathbf{a}) = \sum_{\ell\in a_i}\sum_{i \neq j \in C_\ell^{\mathbf{a}}} d_{i,j}^\ell$. 
Agent $i$'s utility from her strategy $\varphi_i$ is thus defined as $V_i(\bm{\varphi}) := \mathbb{E}_{\mathbf{a} \sim \bm{\varphi}} [d_i(\mathbf{a})]$.

\subsection{Nash Stability in POCF Games}
\label{sec:nash}
Each agent joins a coalition with the goal of maximizing her own utility. We thus want to study stability und
er single agents' incentives to deviate between coalitions. When agents play \textit{\textbf{pure}} strategies, each agent $i$'s strategy is joining certain candidate coalitions by only picking some action $a_i \in \mathcal{A}_i$. Let $\mathbf{a}_{-i} = (a_j)_{j \neq i}$ be the joint strategy of all agents except for agent $i$. Agent $i$ can then \textit{\textbf{deviate}} by moving from her selected coalition to another one with index $a_i' \in \mathcal{A}_i$, which is a \textit{\textbf{Nash deviation}} if it improves her utility, i.e., $V_i(\mathbf{a}_{-i},a_i') > V_i(\mathbf{a}_{-i},a_i)$. We also study \textit{mixed} strategies, where we consider the other notion of \textit{mixed Nash deviations}. Consider the joint strategy $\bm{\varphi}$. Letting $\bm{\varphi}_{-i} = (\varphi_j)_{j \neq i}$ for any agent $i$, agent $i$ may perform a \textit{\textbf{(mixed single-agent) deviation}} from her strategy $\varphi_i$ to another strategy $\phi_i \in {\Delta(\mathcal{A}_i)}$, which is a \textit{\textbf{mixed Nash deviation}} only if it immediately makes her better off, i.e., $V_i(\bm{\varphi}_{-i},\phi_i) > V_i(\bm{\varphi})$. Hence, a (\textit{pure} or \textit{mixed}) joint strategy for which no Nash deviation is possible is said to be \textit{\textbf{Nash stable}} (NS){, or a \textit{Nash equilibrium}}. 

Our POCF class of coalition formation games generalizes the well-known class of \textit{symmetric additively separable hedonic games} (S-ASHGs) \cite{bogomolnaia2002stability}. Indeed, consider cases where the number of coalitions is unconstrained ($k=n$), and each agent can join exactly one of the $n$ \textit{candidate} coalitions (i.e., $\mathcal{A}_i=\{\{\ell\}\}_{\ell\in[n]}$ for any agent $i$). We thereby obtain S-ASHGs. For S-ASHGs under \textit{pure} strategies, the existence of a Nash stable outcome is guaranteed by potential function argument \cite{bogomolnaia2002stability}. Next, we show that our class of symmetric POCF games are also potential games, {establishing that they always admit at least one pure (and hence mixed) NS strategy. Pure NS strategies need not be unique (e.g., if all agents assign zero utility to one another, then any pure strategy is Nash stable).}
\begin{lemma}
    \label{lemma:potential game}
    A symmetric POCF game with unknown, {symmetric} preferences is a {\normalfont potential game}. {Therefore, any symmetric POCF game always admits at least one pure (and thus mixed) NS strategy}.
\end{lemma}
\begin{proof}
    (\textit{Sketch})
    Consider a joint mixed strategy $\bm{\varphi}$. Given a joint action $\mathbf{a} \sim \bm{\varphi}$, in Appendix A we prove that $\Phi(\mathbf{a}) = \frac{1}{2} \sum_{i \in \mathcal{N}} v_i(\mathbf{a})$ is a potential function for \textit{pure} strategies, as $\Phi(\mathbf{a}_{-i},a_i) - \Phi(\mathbf{a}_{-i},a_i') = v_i(\mathbf{a}_{-i},a_i) - v_i(\mathbf{a}_{-i},a_i')$ for any agent $i$ playing {some} action $a_i' \neq a_i$. Similarly, {slightly abusing notation}, $\Phi(\bm{\varphi}) = \frac{1}{2} \sum_{i \in \mathcal{N}} V_i(\bm{\varphi})$ is a potential for \textit{mixed} strategies as $\Phi(\bm{\varphi}_{-i}, \varphi_i) - \Phi(\bm{\varphi}_{-i}, \phi_i) = V_i(\bm{\varphi}_{-i}, \varphi_i) - V_i(\bm{\varphi}_{-i}, \phi_i)$ for any agent $i$ playing another strategy $\phi_i \in \Delta(\mathcal{A}_i)$.
\end{proof}

In S-ASHGs with \textit{pure} strategies, computing {NS} strategies is PLS-complete \cite{gairing2019computing}. {As S-ASHGs are a special case of our model, this PLS-completeness extends to symmetric POCF games by Lemma \ref{lemma:potential game}}. We thus also consider an \textit{approximate} notion of Nash stability. For a joint strategy $\bm{\varphi}$, agent $i$'s \textit{best response} to the other agents' strategies is a Nash deviation given by a strategy $\phi_i^{\star}$ satisfying $V_i^{\star}(\bm{\varphi}_{-i}) := V_i(\bm{\varphi}_{-i},\phi_i^{\star}) = \max_{\phi \in \Delta(\mathcal{A}_i)} V_i(\bm{\varphi}_{-i},\phi)$. We measure the worst agent's local gap between the expected utilities she gets from her best response and her current strategy by $\bm{\varphi}$'s \textbf{\textit{duality gap}}:
\begin{equation}
    \label{eq:gap}
         \gap(\bm{\varphi}):=\max_{i \in \mathcal{N}} [V_i^\star(\bm{\varphi}_{-i}) - V_i(\bm{\varphi})]
\end{equation}
Hence, for any $\varepsilon> 0$, the agents' joint strategy $\bm{\varphi}$ is \textit{\textbf{$\varepsilon$-approximate Nash stable}} ($\varepsilon$-NS) if no agent can improve her gain by more than $\varepsilon$, i.e., $\gap(\bm{\varphi}) \leq \varepsilon$. If $\gap(\bm{\varphi})=0$, then $\bm{\varphi}$ is {\textit{exactly} Nash stable}.

\begin{remark}
    \label{remark:uniform}
    {If $\{\ell\} \in A_i$ for any $\ell \in [k]$ and any agent $i$, consider the mixed strategy $\varphi_i$ where each agent treats all candidate coalitions as equally desirable, i.e., $\varphi_i(\{\ell\}) = \frac{1}{k}$ for any $\ell \in [k]$.} Clearly, this is an exact mixed NS strategy{, but} it ignores the agents' preferences entirely, which is unrealistic in practical scenarios as agents often act strategically based on their own preferences, not arbitrarily. {Instead, our framework aims to learn preference-aligned mixed NS strategies that reflect real-life behavior by accounting for agents' incentives.}
\end{remark}

\subsection{Offline Learning in Coalition Formation}
In the offline version of our model, we have access only to a fixed dataset $\mathcal{S}=\{(\mathbf{a}^m, \mathbf{v}^m)\}_{m=1}^M$ of $M$ samples independently{\footnote{{Independence is a standard assumption in offline learning (e.g., \cite{cui2022provably,cui2022offline}), where the goal is to learn from a fixed dataset rather than model its temporal generation. In our setting, this is also practical: for statistical analysis, data can be treated as i.i.d. samples from an exploration policy, even if it was generated sequentially. Standard subsampling methods (e.g., \cite{cui2022provably}) can alleviate temporal correlations to enforce independence.}}} drawn from a possibly unknown exploration policy $\rho\in \Delta(\mathcal{A})$, where $\mathbf{a}^m$ is a joint action sampled from $\rho$ and $\mathbf{v}^m$ comprises utility feedbacks resulting from the partition induced by $\mathbf{a}^m$. In particular, no further sampling from $\rho$ is allowed. We say that a joint action $\mathbf{a}$ is \textit{\textbf{covered}} by the exploration policy if $\rho(\mathbf{a})>0$. We study the following two utility feedback models, ordered by decreasing level of granularity:
\begin{enumerate}
    \item \textbf{Semi-bandit feedback,} capturing scenarios where we can observe \textit{agent-level} feedback about each agent's utilities from her interactions with any other member of her chosen coalitions, i.e., for any $m\in[M]$, each agent $i$ and any $\ell \in a_i^m$, we attain her realized utility $v_{i,j}^{\ell,m}$ for each agent $j \neq i$ in the $\ell$-th candidate coalition $C_\ell^{\mathbf{a}^m}$ formed by $\mathbf{a}^m$, yielding $\mathbf{v}^m=\{v_{i,j}^{\ell,m}\}_{i\in \mathcal{N}, \ell \in a_i^m, i \neq j \in C_\ell^{\mathbf{a}^m}}$. 
    
    \item \textbf{Bandit feedback,} referring to cases where we can only observe \textit{coalition-level} feedback about each agent's overall utility from her coalitions. Namely, for any $m\in[M]$ and each agent $i$, we only obtain agent $i$'s utility from the partition induced by $\mathbf{a}^m$ (i.e., $v_i(\mathbf{a}^m)$), with no information about the individual utility $v_{i,j}^{\ell,m}$ assigned to any agent $j \neq i$, i.e., $\mathbf{v}^m=\{v_i(\mathbf{a}^m)\}_{i\in \mathcal{N}}$. 

\end{enumerate}

Under each feedback model, our goal is identifying conditions under which the dataset covers sufficient information for an \textit{offline learning} algorithm to learn preferences and use them to recover an $\varepsilon$-NS joint strategy that aligns with the true preferences. We term such assumptions as \textit{(dataset) coverage assumptions}. Particularly, our goal is devising algorithms with \textit{\textbf{small duality gap}} and \textit{\textbf{low sample complexity}}, i.e., such algorithms find an $\varepsilon$-NS strategy for a small $\varepsilon> 0$ using a number of samples that is small in its dependency on the number of agents $n$ and $1/\varepsilon$.


\subsection{Surrogate Minimization in POCF Games}

In Algorithm \ref{alg:surrogate min}, we present a general algorithmic framework for learning NS strategies in offline symmetric POCF games that will be used throughout our work. {To simplify the discussion, we mainly focus on learning NS \textit{mixed} strategies. Later, in Section \ref{sec:Learning Approximate NS Pure Strategies} and Footnote \ref{footnote:Learning Approximate NS Pure Strategies bandit}, we explain how Algorithm \ref{alg:surrogate min} can be adapted for learning NS \textit{pure} strategies.} Algorithm \ref{alg:surrogate min} relies on a \textbf{\textit{utility estimator}} $\hat{v}_i: \mathcal{A}\rightarrow\mathbb{R}$ for each agent $i$, which estimates her mean utility from the partition induced by a sampled joint action $\mathbf{a} \in \mathcal{A}$ via $\hat{v}_i(\mathbf{a})$ (line \ref{state:estimate}). It also exploits an exploration bonus, introducing conservatism into the learning process. Formally, for any agent $i$ and some confidence level $\delta \in (0,1]$ capturing the degree of certainty, $b_i^\delta:\mathcal{A}\rightarrow\mathbb{R}$ is an \textbf{\textit{exploration bonus}} for the utility estimator $\hat{v}_i$ if the following holds with probability at least $1-\delta$ for any joint action $\mathbf{a}\in\mathcal{A}$:
\begin{equation}
    \label{eq:bonus}
    |d_i(\mathbf{a})-\hat{v}_i(\mathbf{a})| \leq b_i(\mathbf{a})
\end{equation}
The specific expression for $\hat{v}_i$ and $b_i^\delta$ may vary based on the utility feedback model, as we will discuss in the remainder of the paper. In any case, we can define the upper confidence bound (UCB) and the lower confidence bound (LCB) of agent $i$ by $\overline{v}_i^\delta(\mathbf{a})= \hat{v}_i(\mathbf{a})+b_i^\delta(\mathbf{a})$ and $\underline{v}_i^\delta(\mathbf{a})= \hat{v}_i(\mathbf{a})-b_i^\delta(\mathbf{a})$, respectively. They allow us to construct optimistic and pessimistic estimates of the expected utility obtained by each agent $i$ from a joint strategy $\bm{\varphi}$, given by (respectively):
\begin{equation}
    \label{eq:opt and pess estimates}
    \overline{V}_{i}^\delta(\bm{\varphi}) := \mathbb{E}_{\mathbf{a} \sim \bm{\varphi}} [\overline{v}_i^\delta(\mathbf{a})] \quad , \quad \underline{V}_{i}^\delta(\bm{\varphi}) := \mathbb{E}_{\mathbf{a} \sim \bm{\varphi}} [\underline{v}_i^\delta(\mathbf{a})]
\end{equation}
For a joint strategy $\bm{\varphi}$, agent $i$'s \textit{optimistic best response} to {others'} strategies is a Nash deviation given by a strategy $\phi_i^{\star}$ {that satisfies}: 
\begin{equation}
    \label{eq:opt best response}
    \overline{V}_i^{\star,\delta}(\bm{\varphi}_{-i}) := \overline{V}_i^\delta(\bm{\varphi}_{-i},\phi_i^{\star}) = \max_{\phi \in {\Delta(\mathcal{A}_i)}} \overline{V}_i^\delta(\bm{\varphi}_{-i},\phi)
\end{equation}
Algorithm \ref{alg:surrogate min} uses \eqref{eq:opt best response} to estimate the true duality gap in \eqref{eq:gap} via $\widehat{\gap}^{{\delta}}(\bm{\varphi}) := \max_{i\in \mathcal{N}} [\overline{V}_i^{\star,\delta}(\bm{\varphi}_{-i})-\underline{V}_{i}^\delta(\bm{\varphi})]$, for which it then computes an approximate minimizer (line \ref{state:surrogate}). Namely, {Algorithm \ref{alg:surrogate min} finds} a joint {mixed} strategy $\bm{\varphi}^{\out}$ that solves \eqref{eq:approx NS} up to $\epsilon_{\opt}$-optimality (e.g., {by standard coordinate-descent schemes, as detailed in Section \ref{sec:EMPIRICAL EVALUATIONS}}){\footnote{{Solving \eqref{eq:approx NS} exactly is generally infeasible, as computing NS outcomes is PLS-complete, even under full information \cite{gairing2019computing}.\label{footnote:pls}}}:}
\begin{equation}
        \widehat{\gap}^\delta(\bm{\varphi}^{\out})
        \leq \min_{\bm{\varphi}} \widehat{\gap}^\delta(\bm{\varphi}) + \epsilon_{\opt}
\end{equation}
The intuition behind considering the above estimates is that they serve as surrogates for the true duality gap. Formally:
\begin{algorithm}[t!]
    \caption{{Surrogate Minimization in POCF Games}}
    \label{alg:surrogate min}   
    \textbf{Input:} An offline dataset $\mathcal{S}=\{(\mathbf{a}^m, \mathbf{v}^m)\}_{m=1}^M$.
    \begin{algorithmic}[1] 
        \State{Construct $\hat{v}_i$ and $b_i^\delta$ for any agent $i$ based on $\mathcal{S}$.\label{state:estimate}}
        \State{Return an \textit{approximate} solution $\bm{\varphi}^{\out}$ to:
        \begin{equation}
            \label{eq:approx NS}
            \min_{\bm{\varphi}\in \prod_{i=1}^n \Delta(\mathcal{A}_i)}\text{ }\max_{i \in \mathcal{N}} \left[\overline{V}_i^{\star,\delta}(\bm{\varphi}_{-i}) - \underline{V}_{i}^\delta(\bm{\varphi})\right]
        \end{equation}
        where $\overline{V}_{i}(\bm{\varphi}), \underline{V}_{i}^\delta(\bm{\varphi}),\overline{V}_i^{\star,\delta}(\bm{\varphi}_{-i})$ are as in \eqref{eq:opt and pess estimates}-\eqref{eq:opt best response}.\label{state:surrogate}}
    \end{algorithmic}
\end{algorithm}
\begin{lemma}
    \label{lemma:surrogate}
    For any $\delta \in (0,1]$ and any joint strategy $\bm{\varphi}${:} $\gap(\bm{\varphi}) \leq \widehat{\gap}^\delta(\bm{\varphi})$ and $\gap(\varphi^{\out}) \leq \min_{\bm{\varphi}}\widehat{\gap}^\delta(\bm{\varphi}) + \epsilon_{\opt}$ with probability at least $1-\delta$, where $\varphi^{\out}$ is the joint strategy produced by Algorithm \ref{alg:surrogate min}.
\end{lemma}
\begin{proof}
    (\textit{Sketch})
    By \eqref{eq:bonus} and \eqref{eq:opt and pess estimates}, $\underline{V}_{i}^\delta(\bm{\varphi})\leq V_{i}(\bm{\varphi}) \leq \overline{V}_{i}^{{\delta}}(\bm{\varphi})$ holds with probability at least $1-\delta$. In Appendix B, we show that this easily implies the desired due to \eqref{eq:gap}.
\end{proof}

Next, we supply a general upper bound on the duality gap of the joint strategy $\varphi^{\out}$ produced by Algorithm \ref{alg:surrogate min}, which serves as a key tool in deriving our main results.
\begin{theorem}
    \label{thm:general duality gap}
    Let $\Gamma$ be the set of all {\normalfont deterministic} joint strategies, $b_i^\delta$ be an exploration bonus for the utility estimator $\hat{v}_i$ of any agent $i$ for any $\delta \in (0,1]$ and consider some NS (possibly mixed) joint strategy $\bm{\varphi}^\star$. Then, under each feedback model, the duality gap of the joint strategy $\varphi^{\out}$ produced by Algorithm \ref{alg:surrogate min} is upper bounded as follows with probability at least $1-\delta$ (which directly translates into Algorithm \ref{alg:surrogate min}'s approximation guarantees for Nash stability): 
    \begin{gather*}
        \gap(\varphi^{\out}) \leq 2\max_{i \in \mathcal{N}} \left[\max_{\bm{\varphi}'\in\Gamma} \mathbb{E}_{\mathbf{a}\sim (\bm{\varphi}_{-i}^\star,\varphi_i')} [b_i^\delta(\mathbf{a})]+ \mathbb{E}_{\mathbf{a}\sim \bm{\varphi}^\star} [b_i^\delta(\mathbf{a})]\right]+ \epsilon_{\opt}
    \end{gather*}
\end{theorem}
\begin{proof}
    (\textit{Sketch})
    In Appendix C, we first prove that \eqref{eq:bonus}-\eqref{eq:opt and pess estimates} imply that $V_i(\bm{\varphi})-\underline{V}_{i}^\delta(\bm{\varphi})$ and $\overline{V}_{i}^\delta(\bm{\varphi})-V_i(\bm{\varphi})$ are at most $2\mathbb{E}_{\mathbf{a} \sim \bm{\varphi}} [b_i^\delta(\mathbf{a})]$. As $\overline{V}_i^\delta(\bm{\varphi}_{-i},\phi_i)$ and $\underline{V}_i^\delta(\bm{\varphi})$ are both linear in each entry of $\bm{\varphi}$ and $\phi_i \in \Delta(\mathcal{A}_i)$ for any agent $i$, the maximum of $\max_{i \in \mathcal{N}} \left[\overline{V}_i^\delta(\bm{\varphi}_{-i},\phi_i) - \underline{V}_i^\delta(\bm{\varphi})\right]$ over all joint strategies $\bm{\phi} \in \prod_{i=1}^n \Delta(\mathcal{A}_i)$ is obtained at a vertex, i.e., a \textit{pure} joint strategy. Combining the above with \eqref{eq:gap}, Lemma \ref{lemma:surrogate} and the fact that $\bm{\varphi}^\star$ is an NS (possibly mixed) joint strategy: $\gap(\varphi^{\out}) \leq \gap(\bm{\varphi}^\star) +  2\max_{i \in \mathcal{N}} [\max_{\bm{\varphi}'\in\Gamma} \mathbb{E}_{\mathbf{a} \sim (\bm{\varphi}_{-i}^\star, \varphi'_i)} [b_i^\delta(\mathbf{a})]$ $+ \mathbb{E}_{\mathbf{a} \sim \bm{\varphi}^\star} [b_i^\delta(\mathbf{a})]] + \epsilon_{\opt} $. We then obtain the desired from $\gap(\bm{\varphi}^\star)=0$.
\end{proof}

\begin{remark}[{\textbf{Asymmetric Preferences}}]
    \label{remark:asymmetric}
    {Though we focus on symmetric preferences mainly to reduce the valuation space and simplify exposition, our results for {\normalfont mixed} strategies easily extend to {\normalfont asymmetric} preferences (e.g., under semi-bandit feedback, this only doubles the number of utility estimates per agent pair). Yet, extending our analysis for {\normalfont pure} strategies to {\normalfont asymmetric} preferences is generally neither theoretically justified nor tractable. Namely, symmetric POCF games always admit a {\normalfont pure} NS outcome (Lemma \ref{lemma:potential game}), but {\normalfont asymmetric} ones may not, and deciding existence is strongly NP-complete \citep{sung2010computational}.} 
\end{remark}

\begin{remark}[{\textbf{Bandit Algorithms}}]
    {While bandit algorithms are often associated with online exploration-exploitation, in our {\normalfont \textbf{offline}} setting bandit-based estimators are not used for exploration. Instead, they are standard in offline bandit learning for exploiting the feedback structure to handle uncertainty under limited information (e.g., \cite{jin2021pessimism}).}
\end{remark}

\section{{Semi-Bandit Feedback}}
\label{sec:semi-bandit}
For semi-bandit feedback settings, we herein derive a {\textit{\textbf{necessary}} and \textit{\textbf{sufficient}}} dataset coverage assumption (Section \ref{sec:Minimal Dataset Coverage Assumption}), under which we devise an offline learning algorithm with low duality gap and a sample complexity bound with \textit{\textbf{optimality}} guarantees up to logarithmic factors (Section \ref{sec:Algorithm 1 under Semi-Bandit Feedback}).

\subsection{The Coalition Size Coverage Assumption}
\label{sec:Minimal Dataset Coverage Assumption}
Under semi-bandit feedback, we obtain each agent's individual utility from interacting with every member of her chosen coalitions. This motivates us to require that the dataset covers the kinds of partitions that could arise from agents unilaterally deviating. Specifically, for any $\ell \in [k]$, if the $\ell$-th candidate coalition formed by some \textit{unilateral deviation} from a single NS strategy $\bm{\varphi}^\star$ has size $m \in [k]$, then the dataset should contain at least one joint action in which the $\ell$-th candidate coalition also has size $m$. The key intuition is {described in Observation \ref{obs:deviations from pure}:}
\begin{observation}
    \label{obs:deviations from pure}
    {For any joint action $\mathbf{a} \in \mathcal{A}$ sampled from some joint strategy $\bm{\varphi}$ and coalition $C\in\pi^{\mathbf{a}}$, consider a unilateral deviation of some agent $i$. If $i \in C$, then agent $i$ may either remain in or leave $C$ after she deviates. Otherwise, if $i \notin C$, then agent $i$ may either join $C$ or stay out after deviating. Thus, each unilateral deviation affects coalition sizes by either increasing or decreasing them by $1$, or leaving them unchanged. This allows us to reason about candidate coalitions and their sizes by considering only such deviations.}
\end{observation}

Formally, for any joint strategy $\bm{\varphi}$, any $\ell \in [k]$ and any coalition size $\alpha\in[n]\cup\{0\}$, we define the \textit{\textbf{coalitional overall density}} of the $\ell$-th candidate coalition by $d_\ell^{\bm{\varphi}}(\alpha) = \sum_{\mathbf{a}\in\mathcal{A}: |C_\ell^{\mathbf{a}}|=\alpha} \bm{\varphi}(\mathbf{a})$, where $d_\ell^{\rho}(\alpha)$ is defined similarly for the exploration policy $\rho$. Our assumption is then formulated as:
\begin{assumption}[\textbf{Coalition Size Coverage}]
    \label{assump:unilateral deviation}
    There exists an NS joint strategy{\footnote{{As we focus on \textit{symmetric} POCF games, they always admit a pure NS strategy by Lemma \ref{lemma:potential game}, and thus the NS strategy $\bm{\varphi}^\star$ in Assumption \ref{assump:unilateral deviation} can be always chosen as \textit{pure}.}}} $\bm{\varphi}^\star$ such that, for any agent $i$, any $\ell \in [k]$ and any coalition size $\alpha\in[n]\cup\{0\}$, if there is a strategy $\phi_i \in \Delta(\mathcal{A}_i)$ with $d_\ell^{\bm{\varphi}_{-i}^\star,\phi_i}(\alpha)>0$, then $d_\ell^{\rho}(\alpha) > 0$.
\end{assumption}
\begin{remark}
    \label{remark:realistic}
    {There are realistic scenarios where Assumption \ref{assump:unilateral deviation} can be easily verified. For instance, it is {\normalfont always} satisfied by an exploration policy that samples joint actions uniformly at random, as such a policy covers all coalitions sizes. Datasets generated by this policy reflect early exploratory behavior, where agents lack prior knowledge.}
\end{remark}
We thus measure how well the dataset covers all coalition sizes through unilateral deviations from a joint strategy $\bm{\varphi}$ via the \textit{\textbf{coalition size coefficient}}, which is defined by:
\begin{equation}
    \label{eq:coalition coefficient}
    {c_{\size}^{\bm{\varphi}}} = \max_{i\in \mathcal{N}, \ell \in [k], \phi_i\in\Delta(\mathcal{A}_i), \alpha\in [n]\cup\{0\}:d_\ell^{\rho}(\alpha) > 0} \frac{d_\ell^{\bm{\varphi}_{-i},\phi_i}(\alpha)}{d_\ell^{\rho}(\alpha)}
\end{equation}

Next, we prove that {Assumption \ref{assump:unilateral deviation} is \textit{\textbf{necessary}}, i.e., no \textit{weaker} assumption} enables efficient learning of an approximate NS {mixed} strategy with a small duality gap, regardless of the dataset size. 
\begin{theorem}
    \label{thm:half gap}
    Let $\mathcal{G}$ be the class of all pairs $(G,\rho)$ consisting of a POCF game $G$ and an exploration policy $\rho$ satisfying Assumption \ref{assump:unilateral deviation}, except for at most one coalition size $\alpha\in[n]\cup\{0\}$. Then, for any algorithm $\alg$ with {\normalfont semi-bandit} feedback, there is $(G,\rho)\in \mathcal{G}$ such that any joint strategy $\bm{\varphi}$ produced by $\alg$ satisfies $\gap(\bm{\varphi})\geq\frac{1}{2}$ for the POCF game $G$, regardless of the dataset size.
\end{theorem}
\begin{proof}
    (\textit{Sketch})
    In Appendix D, we construct two symmetric POCF games $G_1$ and $G_2$ with $6$ agents whose utility distributions are deterministic, where the number of coalitions is at most $2$ (i.e., $k=2$), the action space of each agent is $\{\{1\},\{2\}\}$ and both games share the same exploration policy $\rho$. In the first game $G_1$, there are two types of pure NS joint strategies, which consist of all strategies where either only $2$ agents join the first candidate coalition or the grand coalition is formed within that coalition (i.e., all $6$ agents join the first candidate coalition). In the second game $G_2$, there is only one type of pure NS strategies, comprising all strategies where only $5$ agents join the first candidate coalition. For both games, we construct the same exploration policy $\rho$, which picks a joint action uniformly at random from the set of all joint actions where the first candidate coalition consists of exactly $2$, $4$ or $5$ agents, whereas all other joint actions are assigned zero probability under $\rho$. Then, we prove that the pairs $(G_1,\rho)$ and $(G_2,\rho)$ both belong to the class $\mathcal{G}$ stated in Theorem \ref{thm:half gap}. Afterwards, we show that any algorithm $\alg$ \textit{cannot} distinguish between $(G_1,\rho),(G_2,\rho)$ as both games appear behaviorally the same from the perspective of the data available under $\rho$, regardless of the size of the dataset obtained by $\alg$. Letting $q$ be the probability that a joint action $\mathbf{a}$ inducing a coalition of size $5$ is sampled from the joint strategy $\bm{\varphi}$ produced by $\alg$, we prove that $\gap(\bm{\varphi})\geq q$ for game $G_1$ and $\gap(\bm{\varphi})\geq 1-q$ for game $G_2$. As either $q\geq \frac{1}{2}$ or $1-q\geq \frac{1}{2}$, the desired follows.
\end{proof}

\subsection{Algorithm \ref{alg:surrogate min} under Semi-Bandit Feedback}
\label{sec:Algorithm 1 under Semi-Bandit Feedback}
Next, {we prove that Assumption \ref{assump:unilateral deviation} is \textit{\textbf{sufficient}} for learning approximate NS strategies. Specifically,} we derive the utility estimators and their exploration bonuses for which Algorithm \ref{alg:surrogate min} has low duality gap and sample complexity under semi-bandit feedback. Given an offline dataset $\mathcal{S}=\{(\mathbf{a}^m, \mathbf{v}^m)\}_{m=1}^M$, for any $m\in[M]$ and each agent $i$, recall that $\mathbf{v}^m$ contains agent $i$'s realized utility $v_{i,j}^{\ell,m}$ for each agent $i \neq j \in C_\ell^{\mathbf{a}^m}$, where $C_\ell^{\mathbf{a}^m}$ is the $\ell$-th candidate coalition formed by $\mathbf{a}^m$ for some $\ell \in a_i^m$. Hence, we use those feedbacks to estimate each agent's mean utility from each other agent $i \neq j \in \mathcal{N}$ by the empirical average of the utilities she received (if any) from agent $j$ across the entire dataset. Namely, let $N_{i,j}^\ell = \sum_{m=1}^M \mathds{1}\{\ell \in a_i^m\cap a_j^m\}$ be the number of samples any pair of agents $i,j$ joined the $\ell$-th candidate coalition, where $\mathds{1}\{\ell \in a_i^m\cap a_j^m\}$ equals $1$ if $\ell \in a_i^m\cap a_j^m$ and $0$ otherwise. Thus, agent $i$'s empirical mean utility from another agent $i \neq j \in \mathcal{N}$ within the $\ell$-th candidate coalition is $\hat{v}_i^\ell(j) = \frac{\sum_{m=1}^M v_{i,j}^{\ell,m} \mathds{1}\{\ell \in a_i^m\cap a_j^m\}}{N_{i,j}^\ell \vee 1}$. Afterwards, we estimate agent $i$'s utility from the partition $\pi^{\mathbf{a}} = (C_\ell^{\mathbf{a}})_{\ell \in [k]}$ induced by any possible joint action $\mathbf{a}$ via:
\begin{equation}
    \label{eq:empirical mean utility}
         \hat{v}_i(\mathbf{a}) = \sum_{\ell\in a_i} \sum_{i \neq j \in C_\ell^{\mathbf{a}}} \hat{v}_i^\ell(j) 
\end{equation}
To explore a joint action $\mathbf{a}$ more often if it is either promising or not explored enough, we construct agent $i$'s exploration bonus such that it decreases with the increase in $N_{i,j}^\ell$ {(here, $\delta \in (0,1]$ is a confidence level capturing the degree of certainty)}:
\begin{equation}
    \label{eq:bonus semi}
         b_i^\delta(\mathbf{a}) = \sum_{\ell\in a_i} \sum_{i \neq j \in C_\ell^{\mathbf{a}}} \sqrt{\frac{2  \log(4(n+1)k/\delta)}{N_{i,j}^\ell\vee 1}}
\end{equation}

\subsubsection{{Learning Approximate NS \textbf{Mixed} Strategies}}
\label{sec:Learning Approximate NS Mixed Strategies}
As we prove in Theorem \ref{thm:semi-bandit}, we carefully designed \eqref{eq:empirical mean utility} and \eqref{eq:bonus semi} based on Hoeffding's inequality to ensure that Algorithm \ref{alg:surrogate min} with the above utility estimators and exploration bonuses has a low duality gap{, establishing our algorithm's approximation to Nash stability under \textit{mixed} strategies}. 
\begin{theorem}
    \label{thm:semi-bandit}
    Under {\normalfont semi-bandit} feedback and Assumption \ref{assump:unilateral deviation}, for any $\delta \in (0,1]$, any dataset size $M \in \mathbb{N}$ and any NS joint strategy $\bm{\varphi}^\star$, the joint {mixed} strategy $\bm{\varphi}^{\out}$ formed by Algorithm \ref{alg:surrogate min} with utility estimators and exploration bonuses as in \eqref{eq:empirical mean utility}-\eqref{eq:bonus semi} satisfies $\gap(\varphi^{\out}) \leq \frac{f^\delta(n,k,\bm{\varphi}^\star)}{\sqrt{M}}+\epsilon_{\opt}$ with probability at least $1-\delta$, where:
    \begin{equation}
            \label{eq:func}
             f^\delta(n,k,\bm{\varphi}^\star)=8kn(n+1) {c_{\size}^{\bm{\varphi}^\star}}\log\left(\frac{4(n+1)k}{\delta}\right) {\sqrt{2(n-1)}} \left[\frac{n-1}{2} +\sqrt{\frac{n}{2}}\right]
    \end{equation}
    {and ${c_{\size}^{\bm{\varphi}^\star}}$ is the coalition size coefficient in \eqref{eq:coalition coefficient}.}
\end{theorem}
\begin{proof}
    (\textit{Sketch})
    In Appendix E, we first use Hoeffding’s inequality to prove that $b_i$ as in \eqref{eq:bonus semi} is an exploration bonus for the utility estimator $\hat{v}_i$ in \eqref{eq:empirical mean utility} for any agent $i$. After bounding $\mathbb{E}_{\mathbf{a}\sim (\bm{\varphi}_{-i}^\star,\varphi_i')} [b_i^\delta(\mathbf{a})]$ and $ \mathbb{E}_{\mathbf{a}\sim \bm{\varphi}^\star} [b_i^\delta(\mathbf{a})]$ for any strategy $\varphi_i \in \Delta(\mathcal{A}_i)$, we obtain the desired result by Theorem \ref{thm:general duality gap}.
\end{proof}

For certain values of $\varepsilon> 0$, we now conclude that our algorithm's sample complexity bound for finding an $\varepsilon$-NS strategy \textit{\textbf{optimally}} depends on $\varepsilon$ up to logarithmic factors.
\begin{corollary}
    \label{coro:optimal}
    Under {\normalfont semi-bandit} feedback and Assumption \ref{assump:unilateral deviation}, for any $\delta \in (0,1]$, any $\varepsilon > \epsilon_{\opt}$ with $\epsilon_{\opt} = o(\varepsilon)$ and any NS joint {{\normalfont mixed}} strategy $\bm{\varphi}^\star$, Algorithm \ref{alg:surrogate min} with utility estimators and exploration bonuses as in \eqref{eq:empirical mean utility}-\eqref{eq:bonus semi} has a sample complexity bound that {\normalfont \textbf{optimally}} depends on $\varepsilon$ (up to logarithmic factors): for a dataset of size $M \geq \frac{f^\delta(n,k,\bm{\varphi}^\star)^2}{{(\varepsilon -\epsilon_{\opt})^2}}$ (see \eqref{eq:func}), $\varphi^{\out}$ is $\varepsilon$-NS (i.e., $\gap(\varphi^{\out}) \leq \varepsilon$) with probability at least $1-\delta$.
\end{corollary}
\begin{proof}
    {(\textit{Sketch})}
    In Appendix E.2, the desired follows from Theorem \ref{thm:semi-bandit}{; our \textbf{\textit{optimality}} guarantees with respect to $\varepsilon$ is by \cite{hassani2020stochastic,bai20provable}.}
\end{proof}

\begin{remark}
    \label{remark:minimum value coefficient}
    {The scaling with $n$ and $k$ in Theorem \ref{thm:semi-bandit} and Corollary~\ref{coro:optimal} is {\normalfont inevitable}: under Assumption \ref{assump:unilateral deviation}, to cover all unilateral deviations of $n$ agents across $k$ candidate coalitions, the dataset must cover $\Theta(nk)$ distinct joint actions. Further, the dependence on $c_{\size}^{\bm{\varphi}^\star}$ is sensible, as its minimum value is at most $3$ (see Appendix E.1). Thus, an exploration policy satisfying the {\normalfont necessary} Assumption \ref{assump:unilateral deviation} with a small ${c_{\size}^{\bm{\varphi}^\star}}$ always exists, but it may be hard to find. Empirically, our algorithm still approaches a low approximation to NS for datasets generated by a uniformly random exploration policy (see Section \ref{sec:EMPIRICAL EVALUATIONS}).} 
\end{remark}

\subsubsection{{Learning Approximate NS \textbf{Pure} Strategies}}
\label{sec:Learning Approximate NS Pure Strategies}
{As we focus on \textit{symmetric} POCF games, they always admit a pure NS strategy by Lemma \ref{lemma:potential game}, which legitimates learning an approximate NS \textit{pure} strategy, opposed to \textit{asymmetric} preferences (see Remark \ref{remark:asymmetric}). Hence, we can adapt Algorithm \ref{alg:surrogate min} to produce a pure strategy. Indeed, let $\Gamma_i \subset \Delta(\mathcal{A}_i)$ be the set of agent $i$'s \textit{deterministic} strategies, i.e., each deterministic strategy $\varphi_i \in \Gamma_i$ corresponds to exactly one \textit{pure} strategy $a_i \in \mathcal{A}_i$, such that $\varphi_i(a_i)=1$ for any agent $i$ and $\varphi_i(a_i')=0$ for any other pure strategy $a_i \neq a_i' \in \mathcal{A}_i$. Thus, instead of approximately solving \eqref{eq:approx NS} over \textit{mixed} strategies, Algorithm \ref{alg:surrogate min} can be modified to find a joint \textit{deterministic} strategy $\bm{\varphi}^{\out}\in \Gamma := \prod_{i=1}^n \Gamma_i$ that approximately solves the minimization problem $\min_{\bm{\varphi}\in\Gamma} \max_{i\in \mathcal{N}} [\overline{V}_i^{\star,\delta}(\bm{\varphi}_{-i})-\underline{V}_{i}^\delta(\bm{\varphi})]$. See Algorithm 2 in Appendix E.3 for a pseudo-code of the resulting algorithm.}

{The approximation to Nash stability and the sample complexity of the modified Algorithm \ref{alg:surrogate min} for \textit{pure} strategies follow from arguments similar to Theorem \ref{thm:semi-bandit} and Corollary \ref{coro:optimal-semi}, and are thus deferred to Appendices E.3-E.4. Essentially, the factor $n+1$ in \eqref{eq:func} is replaced by the constant $3$, reflecting that deviations only change each coalition size by at most one for \textit{pure} strategies (Observation \ref{obs:deviations from pure}). Intuitively, $\mathbb{E}_{\mathbf{a}\sim (\bm{\varphi}_{-i}^\star,\varphi_i')} [b_i^\delta(\mathbf{a})]$ and $ \mathbb{E}_{\mathbf{a}\sim \bm{\varphi}^\star} [b_i^\delta(\mathbf{a})]$ can be bounded by a sum over at most $3$ coalition sizes for \textit{pure} strategies, but over all possible $n+1$ coalition sizes for \textit{mixed} strategies. In both cases, each summand is bounded by the same quantity, yielding the factors $n+1$ and $3$ for mixed and pure strategies, respectively.}

\section{{Bandit Feedback}}
\label{sec:bandit}
Under bandit feedback, we prove that Assumption \ref{assump:unilateral deviation} is insufficient. Intuitively, as single-agent utilities are unobservable in this setting, we \textit{cannot} estimate an agent’s mean utility from any other agent. This may prevent us from accurately estimating utilities obtained from unilateral deviations, as required by Assumption \ref{assump:unilateral deviation}. Formally:
\begin{theorem}
    \label{thm:1/8 gap}
    Let $\mathcal{G}'$ be the class of all pairs $(G,\rho)$ consisting of a POCF game $G$ and an exploration policy $\rho$ satisfying Assumption \ref{assump:unilateral deviation}. Then, for any algorithm $\alg$ with {\normalfont bandit} feedback, there is $(G,\rho)\in \mathcal{G}$ such that any joint strategy $\bm{\varphi}$ produced by $\alg$ satisfies $\gap(\bm{\varphi})\geq\frac{1}{20}$ for the POCF game $G$, regardless of the dataset size.
\end{theorem}
\begin{proof}
    (\textit{Sketch})
    In Appendix F, we build two symmetric POCF games $G_1$ and $G_2$ with $3$ agents whose utility distributions are deterministic, where the number of coalitions is at most $3$ (i.e., $k=3$), the action space of each agent is $\{\{1\}, \{2\}, \{3\}, \{1,2\}\}$, and both games share the same exploration policy $\rho$. Similarly to Theorem \ref{thm:half gap}, both games are built so that their pure NS joint strategies differ, while $\rho$ is designed such that $(G_1,\rho)$ and $(G_2,\rho)$ are both in the class $\mathcal{G}'$ defined in Theorem \ref{thm:1/8 gap}. We then show that any algorithm $\alg$ \textit{cannot} distinguish between $(G_1,\rho),(G_2,\rho)$, as both games are behaviorally indistinguishable under the data distribution $\rho$, regardless of the dataset size observed by $\alg$. By arguments similar to Theorem \ref{thm:half gap}, the desired follows.
\end{proof}

\subsection{Algorithm \ref{alg:surrogate min} under Bandit Feedback}
\label{sec:Reduction to Multi-Agent Offline Linear Bandits}
To circumvent the impossibility in Theorem \ref{thm:1/8 gap}, we next show how to use ridge regression for constructing utility estimators with corresponding exploration bonuses under bandit feedback.{\footnote{{Our results in this section extend to learning approximate NS \textbf{pure} strategies by arguments similar to Section \ref{sec:Learning Approximate NS Pure Strategies}, and are thus deferred to Appendix G.3.}\label{footnote:Learning Approximate NS Pure Strategies bandit}}} For this sake, we first illustrate that the mean utility of each agent $i$ from some joint action $\mathbf{a}$ can be written as an inner product between a vector capturing all single-agent utilities and a binary vector, indicating whether agent $i$ and any other agent $j$ join the same coalitions under $\mathbf{a}$. This allows us to derive an assumption on the information covered by the dataset, which {is sufficient for} sample-efficient learning of an approximate NS strategy with a sample complexity bound that admits \textit{\textbf{optimality}} guarantees up to logarithmic factors.

Formally, under bandit feedback, consider any agent $i$ and any $\ell \in[k]$. We represent agent $i$'s mean mutual utility from interacting with agent $j$ in the $\ell$-th candidate coalition via an $n$-dimensional vector $\bm{\theta}_i^\ell$, whose $j$-th coordinate is $[\bm{\theta}_i^\ell]_j=d_{i,j}^\ell$. We also define a vector-valued function $\mathbf{y}_{i}^\ell: \mathcal{A}\rightarrow\{0,1\}^n$, where $\mathbf{y}_{i}^\ell(\mathbf{a})$ is an $n$-dimensional binary vector whose $j$-th coordinate is $1$ if both agents $i$ and $j$ join the $\ell$-th candidate coalition under some joint action $\mathbf{a} \in \mathcal{A}$, i.e., $[\mathbf{y}_{i}^\ell(\mathbf{a})]_j = \mathds{1}\{\ell \in a_i\cap a_j\}$. We then concatenate agent $i$'s single-agent utilities across coalitions into an $nk$-dimensional vector $\bm{\theta}_i=[\bm{\theta}_i^\ell]_{\ell\in[k]}$, and similarly define $\mathbf{y}_{i}: \mathcal{A}\rightarrow\{0,1\}^{nk}$ via $\mathbf{y}_{i}(\mathbf{a})=[\mathbf{y}_{i}^\ell(\mathbf{a})]_{\ell\in[k]}$. Further, we concatenate all agents' utilities into an $n^2 k$-dimensional vector $\bm{\theta}=[\bm{\theta}_i]_{i\in \mathcal{N}}$. Finally, letting $\mathbf{0}_{\kappa}$ be the $\kappa$-dimensional all-zeros vector, we denote the concatenation $\mathbf{z}_{i}(\mathbf{a})= [\mathbf{0}_{k(i-1)} , \mathbf{y}_{i}(\mathbf{a}),\mathbf{0}_{k(n-i+1)}]$.

As $d_i(\mathbf{a}) = \sum_{\ell\in a_i}\sum_{i \neq j \in C_\ell^{\mathbf{a}}} d_{i,j}^\ell$ is agent $i$'s mean utility from a joint action $\mathbf{a} \in \mathcal{A}$, then it can be written as $d_i(\mathbf{a}) =\langle\mathbf{z}_{i}(\mathbf{a}), \bm{\theta} \rangle$. Therefore, under bandit feedback, we can construct a utility estimator $\hat{v}_i$ for agent $i$'s mean utility with an exploration bonus $b_i^\delta$ for a confidence level $\delta \in (0,1]$ using \textit{ridge regression} over a dataset $\mathcal{S}=\{(\mathbf{a}^m, \mathbf{v}^m)\}_{m=1}^M$ to estimate $\bm{\theta}$ as follows:
\begin{equation}
    \label{eq:ridge}
    \begin{aligned}
        \hat{v}_i(\mathbf{a})&=\langle\mathbf{z}_{i}(\mathbf{a}) ,\hat{\bm{\theta}} \rangle \\
        b_i^\delta(\mathbf{a})&= {\|\mathbf{z}_{i}(\mathbf{a})\|_{V^{-1}} \sqrt{\beta}= \sqrt{\mathbf{z}_{i}(\mathbf{a})^\top V^{-1} \mathbf{z}_{i}(\mathbf{a}) \beta}} \\
        \hat{\bm{\theta}} &= V^{-1} \sum_{m\in[M]} \sum_{i\in \mathcal{N}} \mathbf{z}_{i}(\mathbf{a}^{{m}}) v_{i}(\mathbf{a}^{{m}}) \\
        V&=I+\sum_{m\in[M]} \sum_{i\in \mathcal{N}} \mathbf{z}_{i}(\mathbf{a}^{{m}})\mathbf{z}_{i}(\mathbf{a}^{{m}})^\top \\
        \sqrt{\beta}&=2\sqrt{n^2k}+\sqrt{n^2k\log(1+M/n)+\iota}
    \end{aligned}
\end{equation}
where $\iota=2 \log(4(n+1)k/\delta)$. Next, we derive our assumption on the information contained in the dataset. Intuitively, for any agent $i$ and any unilateral deviation of agent $i$ from some NS joint strategy, it ensures that the dataset contains at least as much information as a sufficiently large dataset generated according to the joint strategy induced by that unilateral deviation, so that the associated utilities can be reliably estimated. Formally:

\begin{assumption}[\textbf{Action Coverage}]
    \label{assump:action}
    There exist a universal constant ${c_{\act}}>0$ and an NS joint strategy $\bm{\varphi}^\star$ such that, for any agent $i$ and any strategy $\phi_i \in \Delta(\mathcal{A}_i)$, it holds that:
    \begin{equation}
        \label{eq:action coverage}
        V\succeq I+M {c_{\act}} \mathbb{E}_{\mathbf{a}\sim (\bm{\varphi}_{-i}^\star,\phi_i)}[\mathbf{z}_{i}(\mathbf{a})\mathbf{z}_{i}(\mathbf{a})^\top]
    \end{equation}
\end{assumption}

Here, $V$ is the covariance matrix of the dataset as defined in \eqref{eq:ridge}, while the right-hand side of \eqref{eq:action coverage} is the covariance matrix of joint actions sampled over $M {c_{\act}}$ episodes from the strategy $(\bm{\varphi}_{-i}^\star,\phi_i)$ induced by a unilateral deviation $\phi_i \in \Delta(\mathcal{A}_i)$ from an NS joint strategy $\bm{\varphi}^\star$. Thus, to enable accurate utility estimation, Assumption \ref{assump:action} requires that the dataset is at least as informative as a dataset of $M {c_{\act}}$ samples independently drawn from $(\bm{\varphi}_{-i}^\star, \phi_i)$.

Under Assumption \ref{assump:action}, we are now ready to prove that Algorithm \ref{alg:surrogate min} with utility estimators and exploration bonuses as in \eqref{eq:ridge} has a low duality gap. In particular, Theorem \ref{thm:bandit} quantifies the approximation guarantees to Nash stability obtained by Algorithm \ref{alg:surrogate min} with utility estimators and exploration bonuses as in \eqref{eq:ridge}. Formally:
\begin{theorem}
    \label{thm:bandit}
    Under {\normalfont bandit} feedback and Assumption \ref{assump:action}, for any $\delta \in (0,1]$ and dataset size $M \in \mathbb{N}$, Algorithm \ref{alg:surrogate min} with utility estimators and exploration bonuses as in \eqref{eq:ridge} outputs a joint {{\normalfont mixed}} strategy $\bm{\varphi}^{\out}$ satisfying $\gap(\varphi^{\out}) \leq 4\sqrt{\frac{n^2 k\beta}{{c_{\act}}M}}+\epsilon_{\opt}$ {with probability at least $1-\delta$.}
\end{theorem}
\begin{proof}
    (\textit{Sketch})
    In Appendix G, we first prove that $b_i^{{\delta}}$ {in} \eqref{eq:ridge} is an exploration bonus for the utility estimator $\hat{v}_i$ in \eqref{eq:ridge} for any agent $i$. After bounding $\mathbb{E}_{\mathbf{a}\sim (\bm{\varphi}_{-i}^\star,\varphi_i')} [b_i^\delta(\mathbf{a})]$ and $\mathbb{E}_{\mathbf{a}\sim \bm{\varphi}^\star} [b_i^\delta(\mathbf{a})]$ by $\sqrt{\frac{n^2 k\beta}{{c_{\act}}M}}$ {for any $\varphi_i \in \Delta(\mathcal{A}_i)$}, we obtain the desired result by Theorem \ref{thm:general duality gap}.
\end{proof}

For specific values of $\varepsilon> 0$, we now infer that our algorithm's sample complexity bound for reaching an $\varepsilon$-NS strategy \textit{\textbf{optimally}} depends on $\varepsilon$ up to logarithmic factors. {The proof is by arguments similar to Corollary \ref{coro:optimal}, and thus deferred to Appendix G.2.}
\begin{corollary}
    \label{coro:optimal-semi}
    Under {\normalfont bandit} feedback and Assumption \ref{assump:action}, for any $\delta \in (0,1]$, any $\varepsilon > \epsilon_{\opt}$ with $\epsilon_{\opt} = o(\varepsilon)$ and any NS joint strategy $\bm{\varphi}^\star$, Algorithm \ref{alg:surrogate min} with utility estimators and exploration bonuses as in \eqref{eq:ridge} has a sample complexity bound that {\normalfont \textbf{optimally}} depends on $\varepsilon$ (up to logarithmic factors): for a dataset of size $M \geq \frac{16 n^2 k\beta}{{c_{\act}}(\varepsilon -\epsilon_{\opt})^2}$, $\varphi^{\out}$ is $\varepsilon$-NS (i.e., $\gap(\varphi^{\out}) \leq \varepsilon$) with probability at least $1-\delta$.
\end{corollary}

\begin{remark}
    {The scaling with $n$ and $k$ in Theorem \ref{thm:bandit} and Corollary~\ref{coro:optimal-semi} is {\normalfont unavoidable}: Assumption \ref{assump:action} requires that the dataset's covariance matrix is as informative as if we had sampled from any of the $\Theta(nk)$ possible unilateral deviations. Further, the dependence on $c_{\act}$ is reasonable: even in games where each agent $i$ can join any possible non-empty subset of candidate coalitions (i.e., $|\mathcal{A}_i|=2^k-1$), we prove in Appendix G.1 that there is always a sufficiently large dataset such that Assumption \ref{assump:action} holds for $c_{\act}=\frac{1}{2nk^4}$ with high probability.}
\end{remark}

\section{{Empirical Evaluations}}
\label{sec:EMPIRICAL EVALUATIONS}
{We evaluate Algorithm \ref{alg:surrogate min} through extensive experiments on several synthetic datasets \cite{code_implementation}. Due to space constraints, we herein focus on \textit{semi-bandit} feedback, with \textit{bandit} feedback results deferred to Appendix H.2. Our experiments evaluate how well our algorithm recovers approximate NS outcomes across various settings.} 

{{\textbf{\textit{Setup}.}}} {For each run, we generate a game with $n$ agents who share the same action set of size at least $3$, sampled uniformly at random from $\mathcal{P}([k])\setminus \emptyset$. Given a joint action $\mathbf{a}$, we sample the mutual utility $v_{i,j}^{\ell}$ of each pair of distinct agents $i,j$ in the $\ell$-th candidate coalition $C_\ell^{\mathbf{a}}$ using one of the following two utility generation models inspired by \citet{boehmer2025causes}: (1) \textbf{Size-Dependent Uniform:} We first draw $u_{i,j}^{\ell}$ uniformly at random from $[-1,1]$ and then set $v_{i,j}^{\ell} = \frac{|C_\ell^{\mathbf{a}}|}{n+1}\cdot u_{i,j}^{\ell}$; (2) \textbf{Size-Dependent Gaussian:} We first draw a mean $\mu_{i,j}$ uniformly at random from $[-1,1]$. Then, $u_{i,j}^{\ell}$ is drawn from the Gaussian distribution with mean $\mu_{i,j}$ and standard deviation $1-\mu_{i,j}$ if $\mu_{i,j}\geq 0$ and $|-1-\mu_{i,j}|$ if $\mu_{i,j}< 0$, ensuring that $v_{i,j}^{\ell} \in [-1,1]$. Finally, we set $v_{i,j}^{\ell} = \frac{|C_\ell^{\mathbf{a}}|}{n+1}\cdot u_{i,j}^{\ell}$. We also considered size-independent and mixed size effects variants, which showed similar trends to the size-dependent versions; their results are thus deferred to Appendix H.1.}

{For each configuration of parameters, we consider two exploration policies for constructing an offline dataset of size $M$: (1) The \textit{\textbf{uniformly random}} policy $\rho^{\text{rand}}$, where joint actions are sampled uniformly at random (i.e., $\rho^{\text{rand}}(\mathbf{a}) = 1/|\mathcal{A}|$ for any joint action $\mathbf{a} \in \mathcal{A}$). $\rho^{\text{rand}}$ covers all coalition sizes, thus satisfying Assumption \ref{assump:unilateral deviation}; (2) To validate the need of Assumption \ref{assump:unilateral deviation}, we also consider a policy $\rho^{\text{1Rand}}$ that does not necessarily satisfy it, which induces a uniformly random strategy $\rho_1^{\text{1Rand}}$ for agent $1$ (i.e., $\rho_1^{\text{1Rand}}(a_1)=\frac{1}{|\mathcal{A}_1|}$ for any $a_1 \in \mathcal{A}_1$); others always deterministically follow the second action that was inserted to their action set during its random generation.}

{Algorithm~\ref{alg:surrogate min} under semi-bandit feedback then uses the exploration bonuses in~\eqref{eq:bonus semi} with confidence level $\delta = 10^{-2}$. As exactly solving \eqref{eq:approx NS} in Algorithm \ref{alg:surrogate min} is intractable (Footnote \ref{footnote:pls}), we employ a practical \textit{coordinate-descent} scheme, widely used in large-scale optimization for its efficient convergence properties to stationary points (e.g., \cite{nesterov2012efficiency,wright2015coordinate}). At each round, we update each agent's mixed strategy via a convex combination of her current strategy and an optimistic best response to the estimated empirical utilities of the other agents. We stop once improvements fall below $10^{-3}$. This yields a smoothed best-response dynamics with gradual convergence to per-agent optimistic best responses. Here, all expectations are estimated by Monte Carlo sampling with $100$ samples per term.}

\begin{figure}[t!]
    \begin{tabular}{cc}
        \begin{subfigure}[b]{0.5\textwidth}
            \centering
            \includegraphics[width=\linewidth]{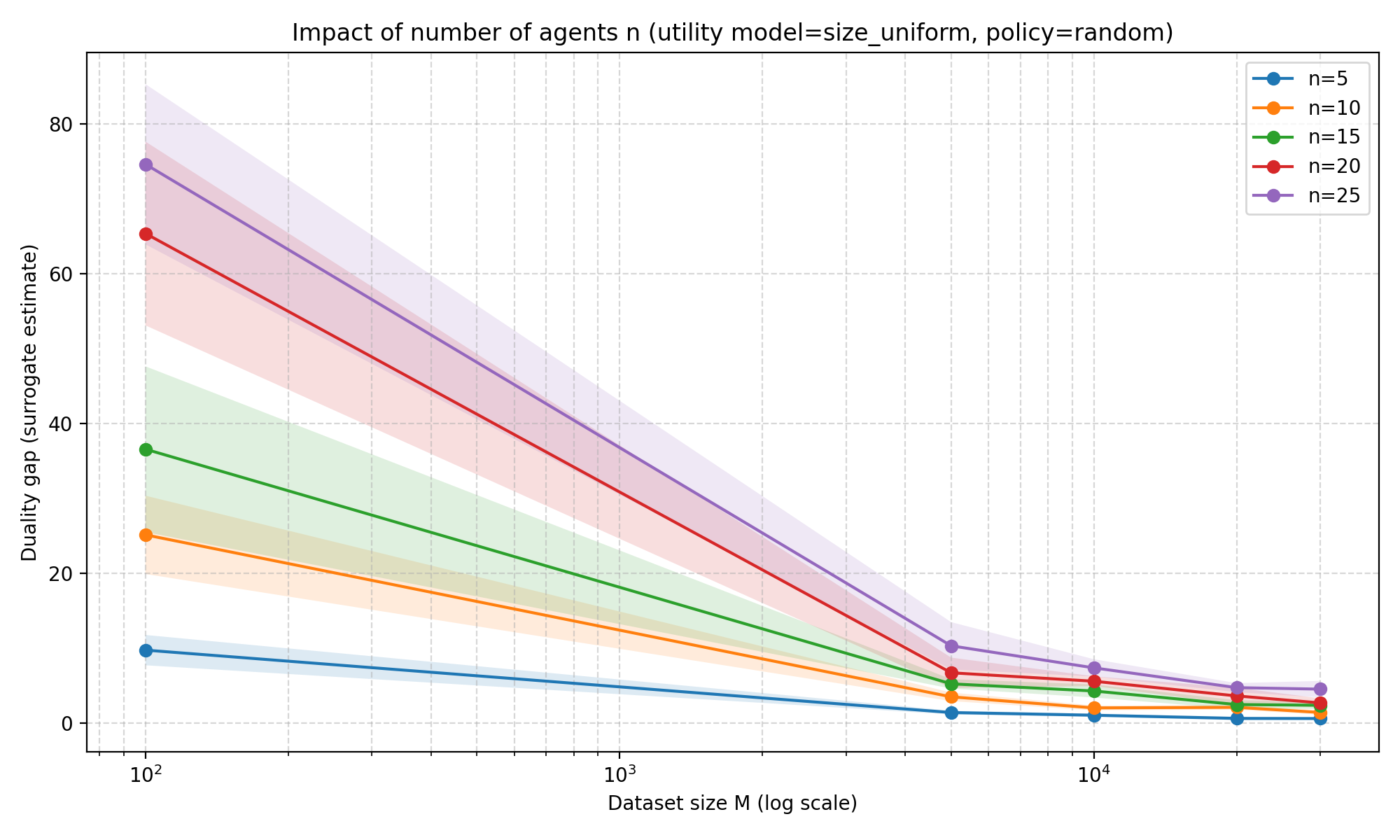}
        \end{subfigure}
        \begin{subfigure}[b]{0.5\textwidth}
            \centering
            \includegraphics[width=\linewidth]{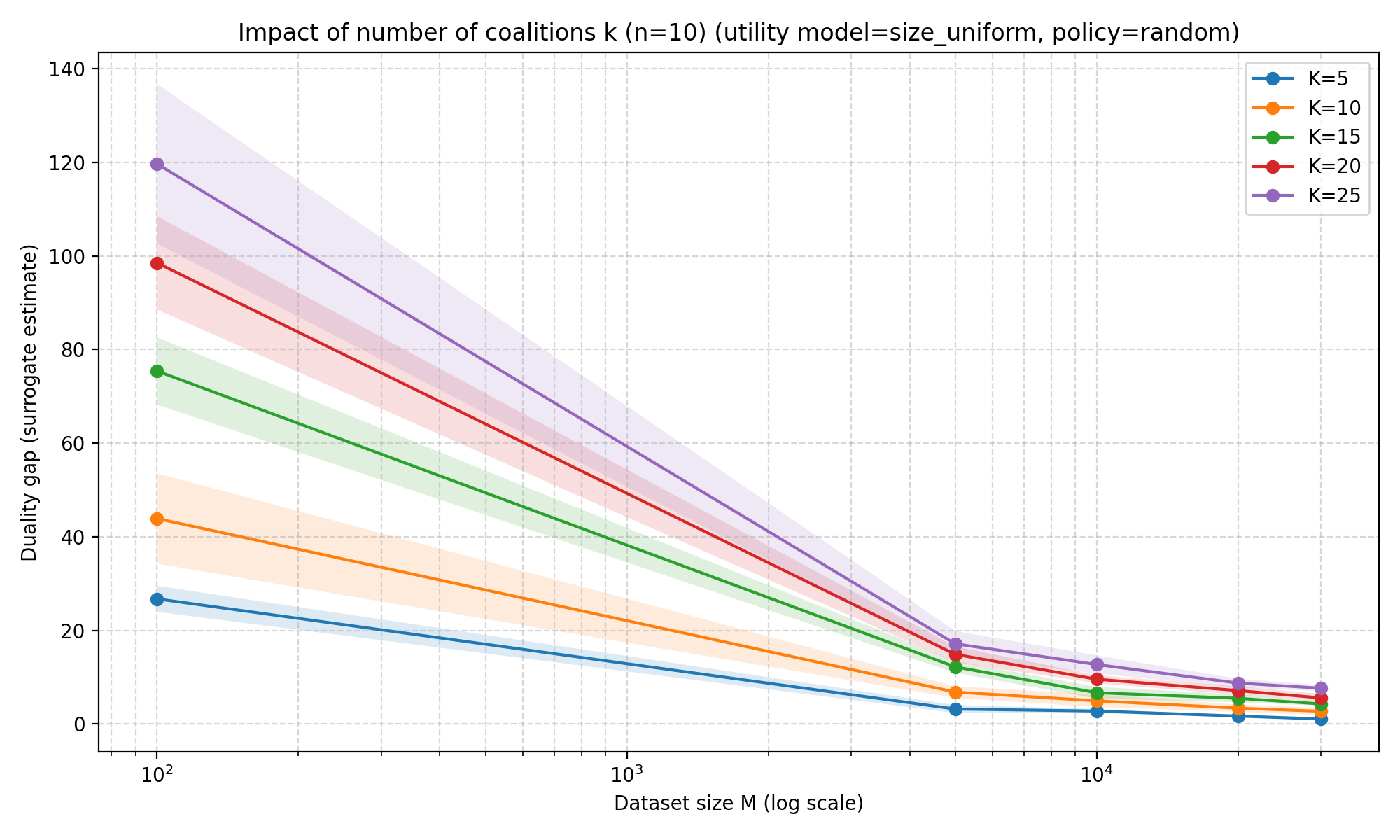}
        \end{subfigure}
        \\
        \begin{subfigure}[b]{0.5\textwidth}
            \centering
            \includegraphics[width=\linewidth]{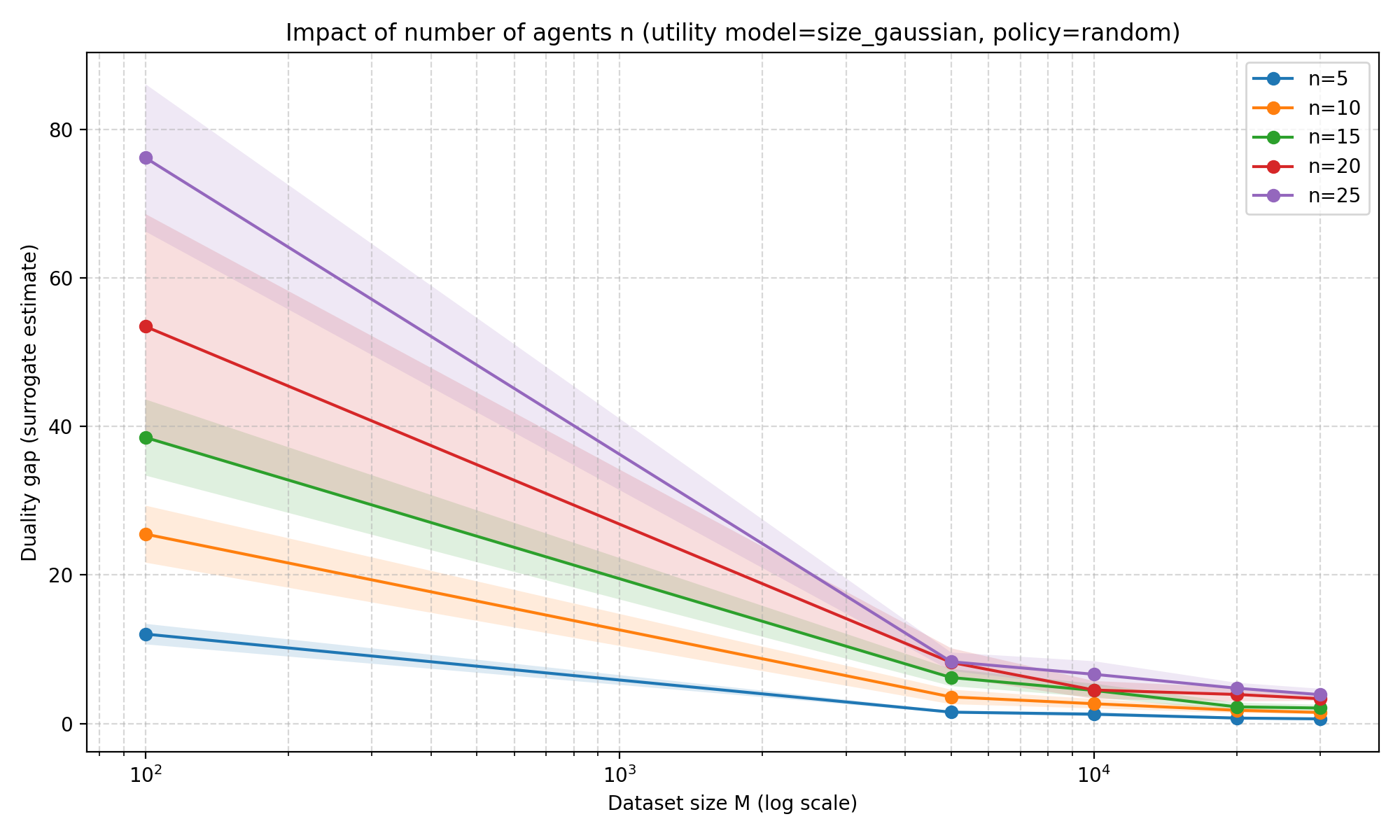}
        \end{subfigure}
        \begin{subfigure}[b]{0.5\textwidth}
            \centering
            \includegraphics[width=\linewidth]{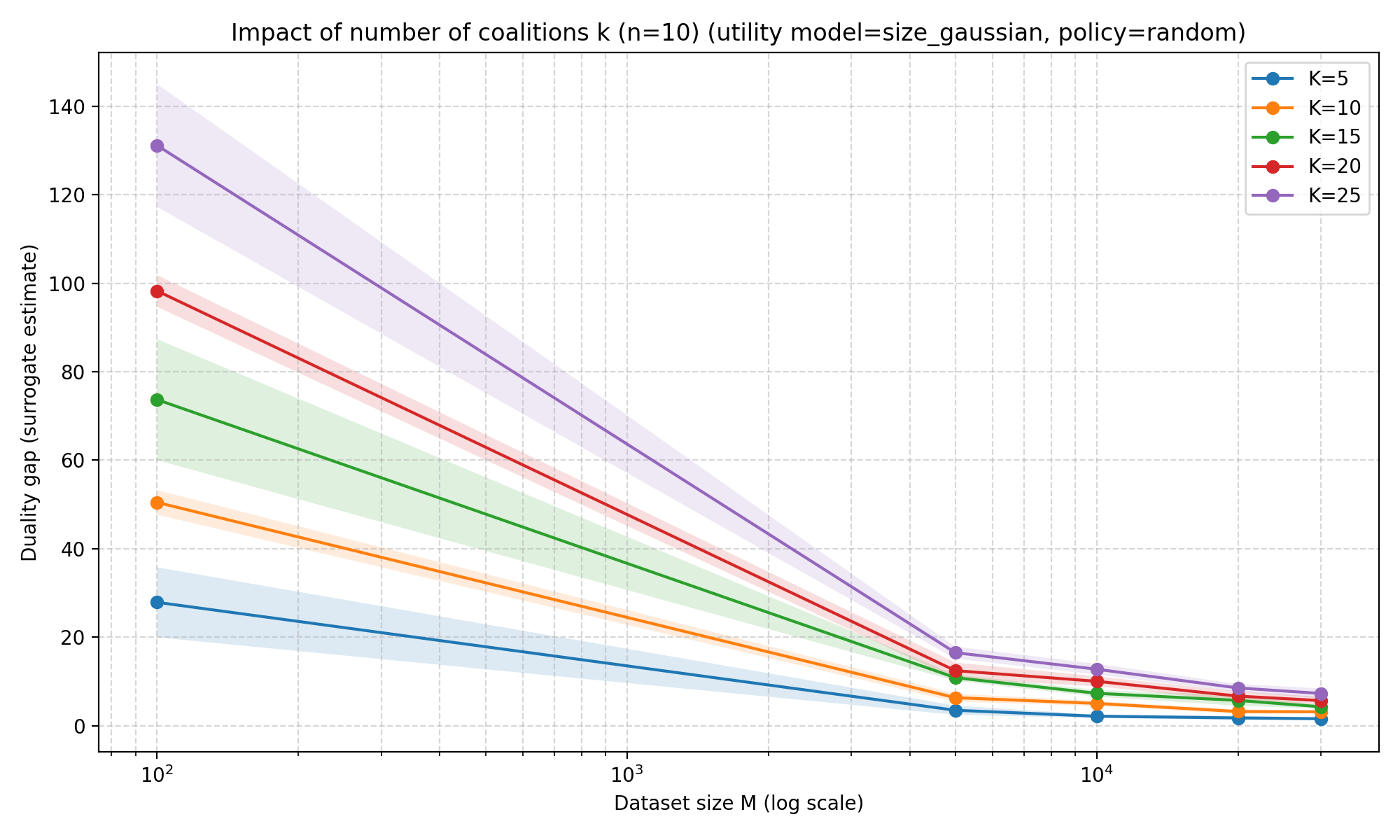}
        \end{subfigure}
        \\
        \begin{subfigure}[b]{0.5\textwidth}
            \centering
            \includegraphics[width=\linewidth]{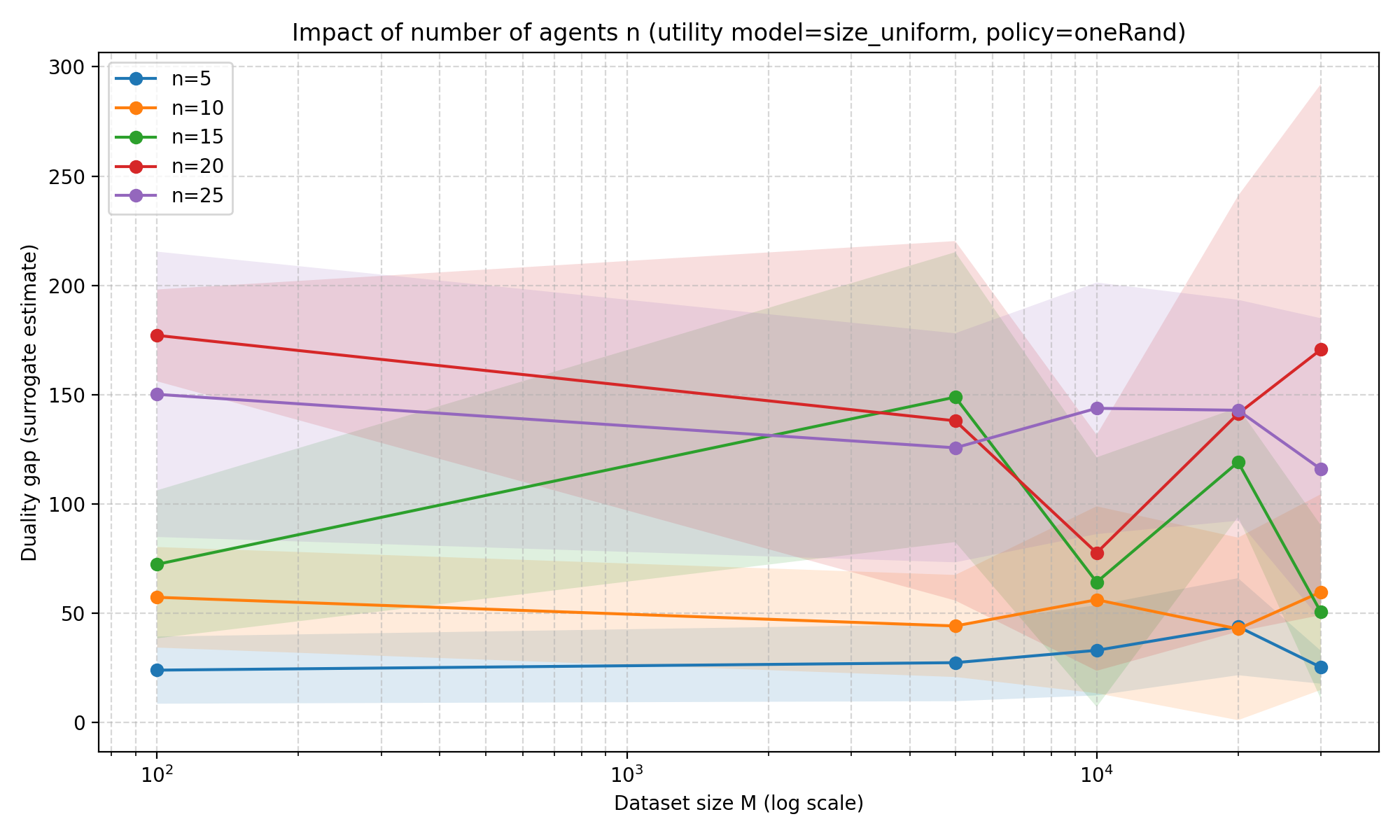}
        \end{subfigure}
        \begin{subfigure}[b]{0.5\textwidth}
            \centering
            \includegraphics[width=\linewidth]{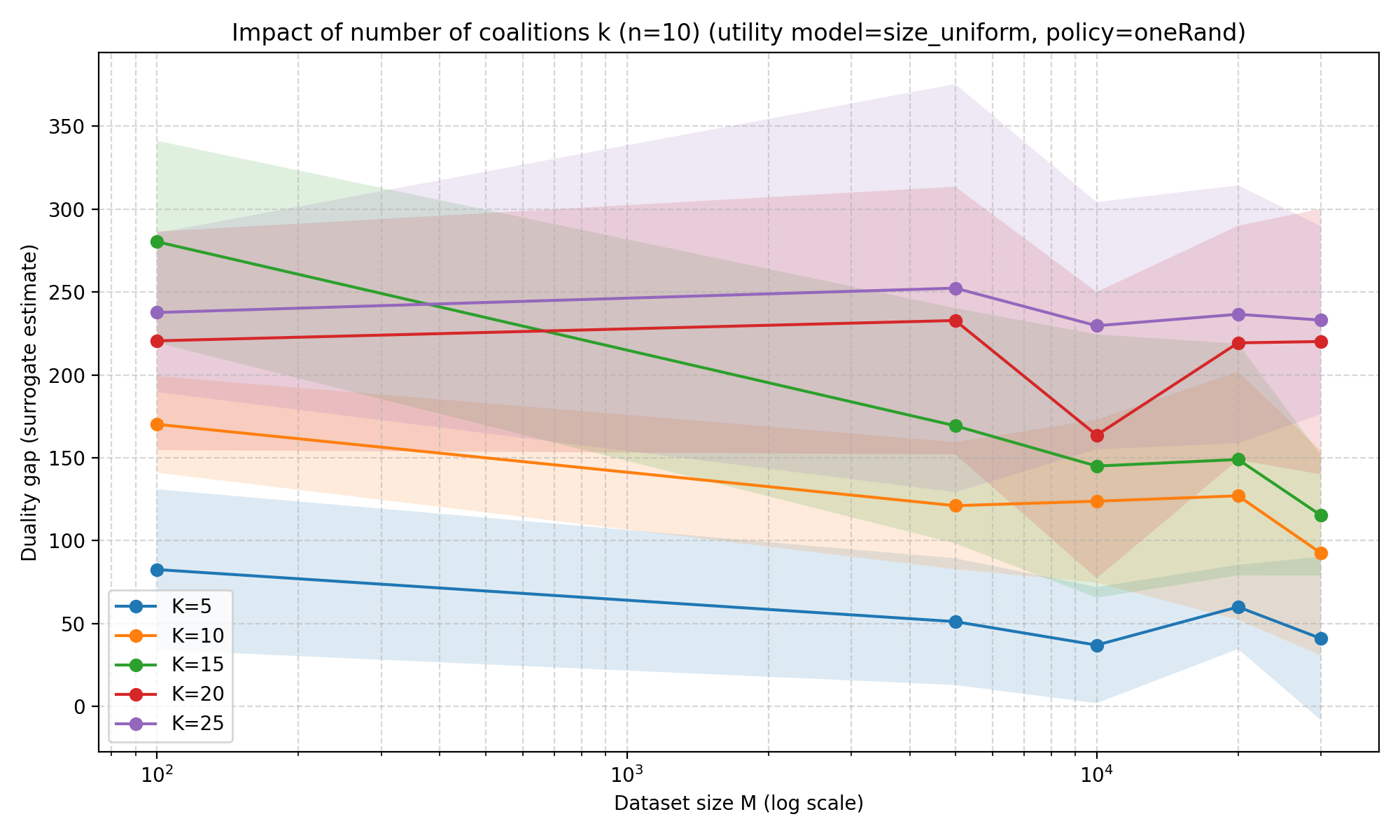}
        \end{subfigure}
    \end{tabular}
    \caption{{Mean approximate duality gap versus the size of datasets generated by $\rho^{\text{rand}}$ (top two rows) and $\rho^{\text{1Rand}}$ (last row) over $5$ runs with different seeds, for varying numbers of agents (left column) and candidate coalitions (right column). Shaded regions indicate standard deviations.}}
    \label{fig:experiments}
\end{figure} 

{\textbf{\textit{Results.}} Figure \ref{fig:experiments} reports the mean approximate duality gap $\widehat{\gap}^\delta(\bm{\varphi}^{\out})$ of the strategy $\bm{\varphi}^{\out}$ produced by Algorithm \ref{alg:surrogate min} versus the size of datasets generated by $\rho^{\text{rand}}$ (top two rows) and $\rho^{\text{1Rand}}$ (last row) over $5$ runs with different seeds. By Lemma \ref{lemma:surrogate}, $\widehat{\gap}^\delta(\bm{\varphi}^{\out})$ upper bounds the \textit{true} duality gap with high probability, thus quantifying how close the learned strategy is to Nash stability. The left column examines the effect of varying the number of agents $n \in \{5,10,15,20,25\}$ while fixing $k=5$, whereas the right column varies the number of candidate coalitions $k \in \{5,10,15,20,25\}$ while fixing $n=10$. In each experiment, we evaluate our algorithm on datasets of sizes $M \in \{10^2, 5\cdot 10^3, 10^4, 2\cdot 10^4, 3\cdot 10^4\}$. For $\rho^{\text{1Rand}}$, we report only size-independent uniform utilities; the Gaussian variant shows similar trends and is thus deferred to Appendix H.1.}

{Algorithm \ref{alg:surrogate min} consistently reaches a low approximation to Nash stability under $\rho^{\text{rand}}$, but fails to do so under $\rho^{\text{1Rand}}$. This supports the practical relevance of Assumption \ref{assump:unilateral deviation}: $\rho^{\text{rand}}$ satisfies it, allowing effective learning, whereas $\rho^{\text{1Rand}}$ may violate it, yielding poorer performance. For $\rho^{\text{rand}}$, the scaling in $n,k,M$ also matches Theorem \ref{thm:semi-bandit}, Corollary~\ref{coro:optimal} and Remark \ref{remark:minimum value coefficient}. For any fixed number of agents $n$, the gap decreases rapidly when the dataset size $M$ is small, but slowly for larger $M$, consistent with the $1/\sqrt{M}$ factor in Theorem \ref{thm:semi-bandit}. Further, obtaining a certain approximation requires larger datasets as $n$ increases, in line with Corollary~\ref{coro:optimal} and Remark~\ref{remark:minimum value coefficient}.}

\section{Conclusions and Future Work}
{We} presented a new model for studying coalition formation with \textit{possibly overlapping} coalitions under \textit{partial information}, where agents' preferences must be inferred from a fixed offline dataset. Under both semi-bandit and bandit feedback, we identified conditions under which the dataset covers sufficient information for an offline learning algorithm to infer preferences and use them to recover an approximately {NS} joint strategy. Under those conditions, we designed sample-efficient algorithms whose sample complexity bounds for {learning} an $\varepsilon$-approximate {NS} strategy \textit{\textbf{optimally}} depend on $\varepsilon$ up to logarithmic factors for certain values of $\varepsilon>0$.

Our research offers many promising directions for future works. {Immediate directions are exploring additional types of preferences, solution concepts and models of partial and/or noisy information. Finally, while we consider datasets with independent samples, common in offline learning (see, e.g., \cite{cui2022provably,cui2022offline}), studying \textit{correlated} samples remains an open challenge, even in single-agent offline reinforcement learning (see, e.g., \cite{cui2022provably,shi2023provably}).}



\section*{Acknowledgments}
The author acknowledges travel support from a Schmidt Sciences 2025 Senior Fellows award to Michael Wooldridge.

\appendix

\section{Proof of Lemma 1}
\label{supp:Lemma 1}

\begin{customlemma}{1}
    \label{supp:lemma:potential game}
    A symmetric POCF game with unknown and symmetric preferences is {\normalfont potential game}. Particularly, it contains both a pure and a mixed NS strategy.
\end{customlemma}
\begin{proof}
    Consider a joint mixed strategy $\bm{\varphi}$. Given a joint action $\mathbf{a} \sim \bm{\varphi}$, we will prove that $\Phi(\mathbf{a}) = \frac{1}{2}\sum_{i \in \mathcal{N}} v_i(\mathbf{a})$ is a potential function for \textit{pure} strategies by showing that $\Phi(\mathbf{a}_{-i},a_i) - \Phi(\mathbf{a}_{-i},a_i') = v_i(\mathbf{a}_{-i},a_i) - v_i(\mathbf{a}_{-i},a_i')$ for any agent $i$ playing another action $a_i' \in \mathcal{A}_i$. Indeed, consider an agent $i$ playing another action $a_i' \in \mathcal{A}_i$. We denote the partition induced by the joint action $\mathbf{a}$ as $\pi^{\mathbf{a}_{-i},a_i}=(C_1,\dots,C_k)$, while $\pi^{\mathbf{a}_{-i},a_i'}=(C_1',\dots,C_k')$ denotes the partition induced by after agent $i$ unilaterally deviates to the action $a_i'$. Note that:
    \begin{equation}
        \label{eq:diff potential}
        \Phi(\mathbf{a}_{-i},a_i) - \Phi(\mathbf{a}_{-i},a_i') = \frac{1}{2}\sum_{j \in \mathcal{N}} [v_j(\mathbf{a}_{-i},a_i) - v_j(\mathbf{a}_{-i},a_i')]
    \end{equation}
    Therefore, we need to analyze the utility gap $v_j(\mathbf{a}_{-i},a_i) - v_j(\mathbf{a}_{-i},a_i')$ incurred by each agent $i \neq j \in \mathcal{N}$ after agent $i$ unilaterally deviates. For any $\ell \in (a_i \cap a_i') \cup ([k]\setminus (a_i \cup a_i'))$, observe that the $\ell$th candidate coalition remains unchanged since agent $i$ remains in this coalition (i.e., $C_\ell=C_\ell'$), and thus the utility of any agent $j\neq i$ in that coalition is not affected after agent $i$ unilaterally deviates, that is, $v_j(\mathbf{a}_{-i},a_i) - v_j(\mathbf{a}_{-i},a_i')=0$ for any agent $i \neq j \in C_\ell$. Therefore, \eqref{eq:diff potential} decomposes as follows:
    \begin{subequations}
        \begin{align}
            \Phi(\mathbf{a}_{-i},a_i) - \Phi(\mathbf{a}_{-i},a_i') &= \frac{1}{2}[v_i(\mathbf{a}_{-i},a_i)-v_i(\mathbf{a}_{-i},a_i')] \label{eq:agent i} \\
            &+ \frac{1}{2}\underbrace{\sum_{i \neq j \in \cup_{\ell \in a_i \setminus a_i'} C_\ell} [v_j(\mathbf{a}_{-i},a_i)-v_j(\mathbf{a}_{-i},a_i')]}_{\text{Abandoned Coalitions}} \label{eq:Abandoned Coalitions} \\
            &+ \frac{1}{2}\underbrace{\sum_{i \neq j \in \cup_{\ell \in a_i'\setminus a_i} C_\ell'} [v_j(\mathbf{a}_{-i},a_i)-v_j(\mathbf{a}_{-i},a_i')]}_{\text{Welcoming Coalitions}} \label{eq:Welcoming Coalitions}
        \end{align}
    \end{subequations}
    Namely, the difference in potential decomposes into three terms: agent $i$'s utility gap (i.e., \eqref{eq:agent i}), the utility gap of each agent in a coalition abandoned by agent $i$ (i.e., \eqref{eq:Abandoned Coalitions}), and the utility gap of each agent in a coalition welcoming agent $i$ (i.e., \eqref{eq:Welcoming Coalitions}). Clearly, the term in \eqref{eq:Abandoned Coalitions} (resp. \eqref{eq:Welcoming Coalitions}) evaluates to zero if $a_i \setminus a_i' = \emptyset$ (resp. $a_i' \setminus a_i = \emptyset$). It thus remains to analyze the utility gap of any agent in the coalitions abandoned by agent $i$ and those welcoming her. Indeed, we consider each one separately:
    \begin{itemize}
        \item \underline{Abandoned Coalitions:} If $a_i \setminus a_i'\neq \emptyset$, then, for any $\ell \in a_i \setminus a_i'$, the coalition $C_\ell$ is abandoned by agent $i$, yielding that $C_\ell' = C_\ell \setminus\{i\}$. Now, consider an agent $i \neq j \in \cup_{\ell \in a_i \setminus a_i'} C_\ell$. Observe that any coalition $C\in \pi^{\mathbf{a}_{-i},a_i}(j)$ such that $i \notin C$ remains unchanged (i.e., $C\in \pi^{\mathbf{a}_{-i},a_i'}(j)$), and thus agent $j$'s utility from that coalition remains the same. Hence, agent $j$'s utility gap is the sum of differences between her utility from any coalition $C\in\pi^{\mathbf{a}_{-i},a_i}(j)$ abandoned by agent $i$ (i.e., $i\in C$) and her utility from that coalition after it is abandoned by agent $i$. That is:
        \begin{subequations}
            \begin{align}
                v_j(\mathbf{a}_{-i},a_i)-v_j(\mathbf{a}_{-i},a_i') &= \sum_{\ell \in a_j}\sum_{j \neq j' \in C_\ell} v_{j,j'}^\ell - \sum_{\ell \in a_j}\sum_{j \neq j' \in C_\ell'} v_{j,j'}^\ell \\
                &= \sum_{\ell \in a_j: i \in C_\ell}\sum_{j \neq j' \in C_\ell} v_{j,j'}^\ell - \sum_{\ell \in a_j: C_\ell' \cup\{i\}\in \pi^{\mathbf{a}}(j)}\sum_{j \neq j' \in C_\ell'} v_{j,j'}^\ell \\
                &= \sum_{\ell \in a_j:i\in C_\ell} v_{i,j}^\ell \\
                &\quad+ \sum_{\ell \in a_j:i\in C_\ell}\sum_{i,j \neq j' \in C_\ell}  v_{j,j'}^\ell - \sum_{\ell \in a_j: C_\ell' \cup\{i\}\in \pi^{\mathbf{a}}(j)}\sum_{j \neq j' \in C_\ell'} v_{j,j'}^\ell \\
                &= \sum_{\ell \in a_j:i\in C_\ell} v_{i,j}^\ell \\
                &\quad+ \sum_{\ell \in a_j:C_\ell'\cup\{i\}\in \pi^{\mathbf{a}}(j)}\sum_{i,j \neq j' \in C_\ell'}  v_{j,j'}^\ell - \sum_{\ell \in a_j: C_\ell' \cup\{i\}\in \pi^{\mathbf{a}}(j)}\sum_{j \neq j' \in C_\ell'} v_{j,j'}^\ell \label{eq:zero}\\
                &= \sum_{\ell \in a_j:i\in C_\ell} v_{i,j}^\ell 
            \end{align}
        \end{subequations}
        where \eqref{eq:zero} hold since $C_\ell=C_\ell'\cup\{i\}$ and we sum only over agent $j$'s utility from any agent $j'\in C_\ell$ with $j'\neq i,j$.

        \item \underline{Welcoming Coalitions:} If $a_i' \setminus a_i\neq \emptyset$, then, for any $\ell \in a_i' \setminus a_i$, the coalition $C_\ell$ welcomes agent $i$, yielding that $C_\ell' = C_\ell \cup \{i\}$. Now, consider an agent $i \neq j \in \cup_{\ell \in a_i' \setminus a_i} C_\ell$. Observe that any coalition $C'\in \pi^{\mathbf{a}_{-i},a_i'}(j)$ such that $i \notin C'$ remains unchanged (i.e., $C'\in \pi^{\mathbf{a}_{-i},a_i}(j)$), and thus agent $j$'s utility from that coalition remains the same. Hence, agent $j$'s utility gap is the sum of differences between her utility from any coalition $C'\in\pi^{\mathbf{a}_{-i},a_i'}(j)$ that welcomes agent $i$ (i.e., $i\in C$) and her utility from that coalition after it welcomes agent $i$. That is:
        \begin{subequations}
            \begin{align}
                v_j(\mathbf{a}_{-i},a_i)-v_j(\mathbf{a}_{-i},a_i') &= \sum_{\ell \in a_j}\sum_{j \neq j' \in C_\ell} v_{j,j'}^\ell - \sum_{\ell \in a_j}\sum_{j \neq j' \in C_\ell'} v_{j,j'}^\ell \\
                &= \sum_{\ell \in a_j: C_\ell \cup\{i\}\in \pi^{\mathbf{a}_{-i},a_i'}(j)}\sum_{j \neq j' \in C_\ell} v_{j,j'}^\ell - \sum_{\ell \in a_j: i \in C_\ell'}\sum_{j \neq j' \in C_\ell'} v_{j,j'}^\ell \\
                &= \sum_{\ell \in a_j: C_\ell \cup\{i\}\in \pi^{\mathbf{a}_{-i},a_i'}(j)}\sum_{j \neq j' \in C_\ell} v_{j,j'}^\ell - \sum_{\ell \in a_j: i \in C_\ell'}\sum_{i,j \neq j' \in C_\ell'} v_{j,j'}^\ell  \\
                &\quad- \sum_{\ell\in a_j:i\in C_\ell'} v_{i,j}^\ell\\
                &= \sum_{\ell \in a_j: i \in C_\ell'}\sum_{i,j \neq j' \in C_\ell'} v_{j,j'}^\ell - \sum_{\ell \in a_j: i \in C_\ell'}\sum_{i,j \neq j' \in C_\ell'} v_{j,j'}^\ell \label{eq:zero welcome} \\
                &\quad- \sum_{\ell\in a_j:i\in C_\ell'} v_{i,j}^\ell\\
                &= - \sum_{\ell\in a_j:i\in C_\ell'} v_{i,j}^\ell
            \end{align}
        \end{subequations}
        where \eqref{eq:zero welcome} hold since $C_\ell'=C_\ell\cup\{i\}$ and we sum only over agent $j$'s utility from any agent $j'\in C_\ell'$ with $j'\neq i,j$.
    \end{itemize}

    Next, note that agent $i$'s utility gap is the difference between her utility from her abandoned and welcoming coalitions, i.e.:
    \begin{equation}
        \label{eq:agent i gap}
        v_i(\mathbf{a}_{-i},a_i) - v_i(\mathbf{a}_{-i},a_i') =  \sum_{\ell \in a_i} \sum_{i \neq j\in C_\ell} v_{i,j}^\ell-\sum_{\ell \in a_i'} \sum_{i \neq j\in C_\ell'} v_{i,j}^\ell
    \end{equation}
    
    Similarly, by slight abuse of notation, $\Phi(\bm{\varphi}) = \frac{1}{2} \sum_{i \in \mathcal{N}} V_i(\bm{\varphi})$ is a potential function for \textit{mixed} strategies as $\Phi(\bm{\varphi}_{-i}, \varphi_i) - \Phi(\bm{\varphi}_{-i}, \phi_i) = V_i(\bm{\varphi}_{-i}, \varphi_i) - V_i(\bm{\varphi}_{-i}, \phi_i)$ for any agent $i$ playing another strategy $\phi_i \in \Delta(\mathcal{A}_i)$.
    
\end{proof}

\section{Proof of Lemma 2}
\label{supp:lemma 2}

\begin{customlemma}{2}
    \label{supp:lemma:surrogate}
    For any $\delta \in (0,1]$ and any joint strategy $\bm{\varphi}$, each of the following holds with probability at least $1-\delta$: 
     \begin{equation}
         \label{eq:part 1}
         \gap(\bm{\varphi}) \leq \widehat{\gap}^\delta(\bm{\varphi})
     \end{equation}
     \begin{equation}
         \label{eq:part 2}
         \gap(\varphi^{\out}) \leq \min_{\bm{\varphi}}\widehat{\gap}^\delta(\bm{\varphi}) + \epsilon_{\opt}
     \end{equation}
     where $\varphi^{\out}$ is the joint strategy produced by Algorithm 1.
\end{customlemma}
\begin{proof}
    We begin with proving \eqref{eq:part 1}. By equations (3) and (4), for any joint strategy $\bm{\varphi}$, note that $\underline{V}_{i}^\delta(\bm{\varphi})\leq V_{i}(\bm{\varphi}) \leq \overline{V}_{i}^\delta(\bm{\varphi})$ holds with probability at least $1-\delta$. Thus, by equation (1) in the full paper, the following is satisfied with probability at least $1-\delta$:
    \begin{equation}
        \label{eq:gap upper bound}
        \begin{aligned}
            \gap(\bm{\varphi}) &=\max_{i \in \mathcal{N}}\max_{\phi_i \in \Delta(\mathcal{A}_i)} [V_i(\bm{\varphi}_{-i},\phi_i) - V_i(\bm{\varphi})] \leq \max_{i \in \mathcal{N}}\max_{\phi_i \in \Delta(\mathcal{A}_i)} \left[\overline{V}_i^\delta(\bm{\varphi}_{-i},\phi_i) - \underline{V}_i^\delta(\bm{\varphi})\right] \\
            &= \max_{i \in \mathcal{N}}\left[\overline{V}_i^{\star,\delta}(\bm{\varphi}_{-i}) - \underline{V}_i^\delta(\bm{\varphi})\right] = \widehat{\gap}^\delta(\bm{\varphi})
        \end{aligned}
    \end{equation}
    
    We proceed to proving \eqref{eq:part 2}. Since the joint strategy $\bm{\varphi}^{\out}$ produced by Algorithm 1 satisfies that $\widehat{\gap}^\delta(\bm{\varphi}^{\out}) \leq \min_{\bm{\varphi}} \widehat{\gap}^\delta(\bm{\varphi}) + \epsilon_{\opt}$, then from \eqref{eq:gap upper bound} we achieve that the following holds with probability at least $1-\delta$:
    \begin{equation}
        \widehat{\gap}^\delta(\bm{\varphi}^{\out}) \leq \widehat{\gap}^\delta(\bm{\varphi}^{\out}) \leq \min_{\bm{\varphi}} \widehat{\gap}^\delta(\bm{\varphi}) + \epsilon_{\opt}
    \end{equation}
    as desired.
\end{proof}

\section{Proof of Theorem 1}

\begin{customthm}{1}
    \label{supp:thm:general duality gap}
    Let $\Gamma$ be the set of all {\normalfont pure} joint strategies and consider some NS (possibly mixed) joint strategy $\bm{\varphi}^\star$. Then, under each feedback model, letting $b_i^\delta$ be an exploration bonus for the utility estimator $\hat{v}_i$ of each agent $i$ for any $\delta \in (0,1]$, the duality gap of the joint strategy $\varphi^{\out}$ produced by Algorithm 1 is upper bounded as follows with probability at least $1-\delta$: 
    \begin{equation}
        \gap(\varphi^{\out}) \leq 2\max_{i \in \mathcal{N}} \left[\max_{\bm{\varphi}'\in\Gamma} \mathbb{E}_{\mathbf{a}\sim (\bm{\varphi}_{-i}^\star,\varphi_i')} [b_i^\delta(\mathbf{a})]+ \mathbb{E}_{\mathbf{a}\sim \bm{\varphi}^\star} [b_i^\delta(\mathbf{a})]\right]+ \epsilon_{\opt}
    \end{equation}
    which directly translates into Algorithm 1's approximation guarantees for Nash stability.
\end{customthm}
\begin{proof}
    For each agent $i$, since $V_i(\bm{\varphi}) := \mathbb{E}_{\mathbf{a} \sim \bm{\varphi}} [d_i(\mathbf{a})]$, $\overline{V}_{i}^\delta(\bm{\varphi}) := \mathbb{E}_{\mathbf{a} \sim \bm{\varphi}} [\overline{v}_i^\delta(\mathbf{a})] =\mathbb{E}_{\mathbf{a} \sim \bm{\varphi}} [\hat{v}_i(\mathbf{a})+b_i^\delta(\mathbf{a})]$ and $ \underline{V}_{i}^\delta(\bm{\varphi}) := \mathbb{E}_{\mathbf{a} \sim \bm{\varphi}} [\underline{v}_i^\delta(\mathbf{a})]=\mathbb{E}_{\mathbf{a} \sim \bm{\varphi}} [\hat{v}_i(\mathbf{a})-b_i^\delta(\mathbf{a})]$ for any joint strategy $\bm{\varphi}$, we obtain that:
    \begin{equation}
        \label{eq:diff opt pess}
        \begin{aligned}
            V_i(\bm{\varphi})-\underline{V}_{i}^\delta(\bm{\varphi}) = \mathbb{E}_{\mathbf{a} \sim \bm{\varphi}} [d_i(\mathbf{a})-\hat{v}_i(\mathbf{a})+b_i^\delta(\mathbf{a})] \leq 2\mathbb{E}_{\mathbf{a} \sim \bm{\varphi}} [b_i^\delta(\mathbf{a})] \\
            \overline{V}_{i}^\delta(\bm{\varphi})-V_i(\bm{\varphi}) = \mathbb{E}_{\mathbf{a} \sim \bm{\varphi}} [\hat{v}_i(\mathbf{a})-d_i(\mathbf{a})+b_i^\delta(\mathbf{a})] \leq 2\mathbb{E}_{\mathbf{a} \sim \bm{\varphi}} [b_i^\delta(\mathbf{a})] 
        \end{aligned}
    \end{equation}
    where the two last inequalities stem from the fact that $b_i^\delta$ is an exploration bonus for the utility estimator $\hat{v}_i$ of each agent $i$, meaning that, with probability at least $1-\delta$, it holds that $|d_i(\mathbf{a})-\hat{v}_i(\mathbf{a})| \leq b_i(\mathbf{a})$ for any joint action $\mathbf{a}\in\mathcal{A}$.

    Next, consider a joint strategy:
    \begin{equation}
        \label{eq:pure joint maxmax}
        \tilde{\bm{\varphi}} \in \max_{\bm{\phi} \in \prod_{i=1}^n \Delta(\mathcal{A}_i)} \max_{i \in \mathcal{N}} \left[\overline{V}_i^\delta(\bm{\varphi}_{-i},\phi_i) - \underline{V}_i^\delta(\bm{\varphi})\right]
    \end{equation}
    As $\overline{V}_i^\delta(\bm{\varphi}_{-i},\phi_i)$ and $\underline{V}_i^\delta(\bm{\varphi})$ are both linear in each entry of the joint strategy $\bm{\varphi}$ and a mixed strategy $\phi_i \in \Delta(\mathcal{A}_i)$ of each agent $i$, then the maximum of $\max_{i \in \mathcal{N}} \left[\overline{V}_i^\delta(\bm{\varphi}_{-i},\phi_i) - \underline{V}_i^\delta(\bm{\varphi})\right]$ over all joint strategies $\bm{\phi} \in \prod_{i=1}^n \Delta(\mathcal{A}_i)$ is obtained at its vertices, i.e., the maximum over $\prod_{i=1}^n \Delta(\mathcal{A}_i)$ corresponds to a \textit{deterministic} (i.e. pure) joint strategy. Accordingly, $\tilde{\bm{\varphi}}$ as selected in \eqref{eq:pure joint maxmax} is a \textit{deterministic} (i.e. pure) joint strategy.

    Now, by \eqref{eq:part 2} in Lemma \ref{supp:lemma:surrogate}, note that:
    \begin{subequations}
        \begin{align}
            \gap(\varphi^{\out}) &\leq \min_{\bm{\varphi}} \max_{i \in \mathcal{N}}\left[\overline{V}_i^{\star,\delta}(\bm{\varphi}_{-i}) - \underline{V}_i^\delta(\bm{\varphi})\right] + \epsilon_{\opt} \\
            &=\min_{\bm{\varphi}} \max_{i \in \mathcal{N}}\left[\overline{V}_i^{\delta}(\bm{\varphi}_{-i}, \tilde{\varphi}_i) - \underline{V}_i^\delta(\bm{\varphi})\right] + \epsilon_{\opt} \\
            &\leq \min_{\bm{\varphi}} \max_{i \in \mathcal{N}}\left[V_i^{\delta}(\bm{\varphi}_{-i}, \tilde{\varphi}_i) - V_i^\delta(\bm{\varphi}) +  2\mathbb{E}_{\mathbf{a} \sim (\bm{\varphi}_{-i}, \tilde{\varphi}_i)} [b_i^\delta(\mathbf{a})] + 2\mathbb{E}_{\mathbf{a} \sim \bm{\varphi}} [b_i^\delta(\mathbf{a})]\right] + \epsilon_{\opt} \label{eq:apply bonus bounds} \\
            &= \min_{\bm{\varphi}} \left[\max_{i \in \mathcal{N}}[V_i^{\delta}(\bm{\varphi}_{-i}, \tilde{\varphi}_i) - V_i^\delta(\bm{\varphi})] +  \max_{i \in \mathcal{N}}\left[2\mathbb{E}_{\mathbf{a} \sim (\bm{\varphi}_{-i}, \tilde{\varphi}_i)} [b_i^\delta(\mathbf{a})] + 2\mathbb{E}_{\mathbf{a} \sim \bm{\varphi}} [b_i^\delta(\mathbf{a})]\right]\right] + \epsilon_{\opt} \\
            &= \min_{\bm{\varphi}} \left[\gap(\bm{\varphi}) +  \max_{i \in \mathcal{N}} \left[2\max_{\bm{\varphi}'\in\Gamma} \mathbb{E}_{\mathbf{a} \sim (\bm{\varphi}_{-i}, \varphi'_i)} [b_i^\delta(\mathbf{a})] + 2\mathbb{E}_{\mathbf{a} \sim \bm{\varphi}} [b_i^\delta(\mathbf{a})]\right]\right] + \epsilon_{\opt} \label{eq:definitions} \\
            &\leq \gap(\bm{\varphi}^\star) + \max_{i \in \mathcal{N}} \left[2\max_{\bm{\varphi}'\in\Gamma} \mathbb{E}_{\mathbf{a} \sim (\bm{\varphi}_{-i}^\star, \varphi'_i)} [b_i^\delta(\mathbf{a})] + 2\mathbb{E}_{\mathbf{a} \sim \bm{\varphi}^\star} [b_i^\delta(\mathbf{a})]\right]+ \epsilon_{\opt} \label{eq:NS} \\
            &= 2\max_{i \in \mathcal{N}} \left[2\max_{\bm{\varphi}'\in\Gamma} \mathbb{E}_{\mathbf{a} \sim (\bm{\varphi}_{-i}^\star, \varphi'_i)} [b_i^\delta(\mathbf{a})] + 2\mathbb{E}_{\mathbf{a} \sim \bm{\varphi}^\star} [b_i^\delta(\mathbf{a})]\right]+ \epsilon_{\opt} \label{eq:last}
        \end{align}
    \end{subequations}
    where the inequality in \eqref{eq:apply bonus bounds} follows from \eqref{eq:diff opt pess}, \eqref{eq:definitions} is due to $\gap(\bm{\varphi}) = \max_{i \in \mathcal{N}}[V_i^{\delta}(\bm{\varphi}_{-i}, \tilde{\varphi}_i) - V_i^\delta(\bm{\varphi})]$ and the definition of $\tilde{\bm{\varphi}}$ in \eqref{eq:pure joint maxmax}, the inequality in \eqref{eq:NS} stems from the fact that $\bm{\varphi}^\star$ is a Nash-stable (possibly mixed) joint strategy. Finally, the equality in \eqref{eq:last} is due to $\gap(\bm{\varphi}^\star)=0$, as $\bm{\varphi}^\star$ is a Nash-stable (possibly mixed) joint strategy.
\end{proof}

\section{Proof of Theorem 2}

\begin{customthm}{2} 
    \label{supp:thm:half gap}
    Let $\mathcal{G}$ be the class of all pairs $(G,\rho)$ consisting of a POCF game $G$ and an exploration policy $\rho$ satisfying Assumption 1 in the full paper, except for at most one coalition size $\alpha\in[n]\cup\{0\}$. Then, under semi-bandit feedback, for any algorithm $\alg$, there is $(G,\rho)\in \mathcal{G}$ such that any joint strategy $\bm{\varphi}$ produced by $\alg$ satisfies $\gap(\bm{\varphi})\geq\frac{1}{2}$ for the POCF game $G$, regardless of the dataset size.
\end{customthm}
\begin{proof}
    We construct two symmetric POCF games $G_1$ and $G_2$, each with $6$ agents whose utility distributions are deterministic, where the number of coalitions is at most $2$ (i.e., $k=k$), the action space of each agent is $\{\{1\}, \{2\}\}$ and both games share the same exploration policy $\rho$. Formally, since any partition can contain at most two coalitions, any joint action $\mathbf{a}$ induces a partition of the agents $\pi^{\mathbf{a}} = (C_1^{\mathbf{a}},C_2^{\mathbf{a}})$, where $C_1^{\mathbf{a}}$ and $C_2^{\mathbf{a}}$ are the set of agents that joined the first and second candidate coalition, respectively. Therefore, we construct the two games and their common exploration policy as follows:
    
    \paragraph{Construction of the first game $G_1$.} In the first game, denoted as $G_1$, there are two types of pure NS strategies, which consist of all pure joint strategies where either only $2$ agents join the first candidate coalition or the grand coalition is formed within the first candidate coalition (i.e., all $6$ agents join the first candidate coalition). Specifically, for each pair of distinct agents $i,j$ and any joint action $\mathbf{a}$, the mutual utility of agents $i,j$ from interacting within the first and second candidate coalitions are determined deterministically via $\mathcal{D}^1_{i,j}(\mathbf{a})$ and $\mathcal{D}^2_{i,j}(\mathbf{a})$ as follows (respectively):
    \begin{equation}
    \label{eq:deterministic utilities}
        \begin{aligned}
            \mathcal{D}_{i,j}^1(\mathbf{a}) &=\begin{cases}
            1 &, \text{ if } i,j \in C_1^{\mathbf{a}} \text{ and } |C_1^{\mathbf{a}}|=2 \\
            -1 &, \text{ if } i,j \in C_1^{\mathbf{a}} \text{ and } |C_1^{\mathbf{a}}|=3 \\
            1 &, \text{ if } i,j \in C_1^{\mathbf{a}} \text{ and } |C_1^{\mathbf{a}}|=4 \\
            1 &, \text{ if } i,j \in C_1^{\mathbf{a}} \text{ and } |C_1^{\mathbf{a}}|=5 \\
            1 &, \text{ if } i,j \in C_1^{\mathbf{a}} \text{ and } |C_1^{\mathbf{a}}|=6 
            \end{cases} \\
            \mathcal{D}_{i,j}^2(\mathbf{a}) &=\begin{cases}
                -\frac{1}{2} &, \text{ if } i,j \in C_2^{\mathbf{a}} \text{ and } |C_2^{\mathbf{a}}|=2 \\
                -\frac{1}{4} &, \text{ if } i,j \in C_2^{\mathbf{a}} \text{ and } |C_2^{\mathbf{a}}|=3 \\
                -\frac{1}{6} &, \text{ if } i,j \in C_2^{\mathbf{a}} \text{ and } |C_2^{\mathbf{a}}|=4 \\
                -\frac{1}{8} &, \text{ if } i,j \in C_2^{\mathbf{a}} \text{ and } |C_2^{\mathbf{a}}|=5 \\
                -\frac{1}{10} &, \text{ if } i,j \in C_2^{\mathbf{a}} \text{ and } |C_2^{\mathbf{a}}|=6 
            \end{cases}
        \end{aligned}
    \end{equation}
    Next, we formally characterize all the pure NS strategies of the first game $G_1$:
    \begin{lemma}
        \label{supp:lemma:NS game 1}
        The pure NS strategies of the first game $G_1$ consist only of all pure joint strategies where either exactly $2$ agents join the first candidate coalition or the grand coalition is formed within the first candidate coalition.
    \end{lemma}
    \begin{proof}
         We first prove that all pure joint strategies that are as in the above are Nash stable, and then show that no other pure joint strategy is Nash-stable. Formally:
        \begin{enumerate}
            \item \textbf{\underline{All pure strategies where only $2$ agents join the first candidate coalition are Nash-stable:}} Consider a joint action $\mathbf{a}$ with $C_1^{\mathbf{a}}=\{i,j\}$ for a pair of distinct agents $i,j$, meaning that any other agent $j'\in \mathcal{N}\setminus\{i,j\}$ does not join the first candidate coalition (i.e., $a_{j'}=\{2\}$). Hence, both agents $i,j$ receive a utility of $1$ from each other. Now, we will prove that no agent has an incentive to unilaterally deviate from $\mathbf{a}$:
            \begin{enumerate}
                \item \textbf{Both agents $i,j$ have no incentive to unilaterally deviate:} If agent $i$ unilaterally deviates, then she decides to join the second candidate coalition instead, which now becomes a singleton coalition consisting only of agent $j$, i.e., $C_1^{\mathbf{a}}=\{j\}$ while the second candidate coalition $C_2^{\mathbf{a}}$ is now of size $5$. Hence, each agent in $C_2^{\mathbf{a}}$ receives a utility of $-\frac{1}{8}$ from each other agent, meaning that she obtains a utility of $-\frac{1}{8}\cdot 4=-\frac{1}{2}$ from that coalition. Agent $i$'s unilateral deviation this decreases her utility from $1$ to $-\frac{1}{2}$. As such, agent $i$ has no incentive to unilaterally deviate. By \eqref{eq:deterministic utilities}, the analysis for the case where agent $j$ unilaterally deviates is identical.
    
                \item \textbf{No agent $j'\in \mathcal{N}\setminus\{i,j\}$ has an incentive to unilaterally deviate:} If agent $j'$ unilaterally deviates, then she decides to join only the first candidate coalition, which now consists of agents $i,j,j'$ (i.e., $C_1^{\mathbf{a}}=\{i,j,j'\}$). Before deviating, agent $j'$ obtains $-\frac{1}{6}\cdot3=-\frac{1}{2}$ utility from joining the second candidate coalition, as she received a utility of $-\frac{1}{6}$ from each of the other $3$ agents in the second candidate coalition. However, by \eqref{eq:deterministic utilities}, the utility of agent $j'$ decreases to $-1-1=-2$ since her utility from both agents $i,j$ is $-1$, and thus she has no incentive to unilaterally deviate.
            \end{enumerate}
            Overall, a joint action $\mathbf{a}$ with $C_1^{\mathbf{a}}=\{i,j\}$ for some pair of distinct agents $i,j$ is a pure Nash-stable strategy.
    
            \item \textbf{\underline{Forming the grand coalition within the first candidate coalition is a Nash-stable pure strategy:}} This strategy corresponds to the joint action $\mathbf{a}=(\{1\}, \{1\}, \{1\}, \{1\}, \{1\}, \{1\})$, where all $6$ agents decide to join the first candidate coalition. Note that each agent's utility from the grand coalition is $5\cdot 1=5$, as each agent's utility from interacting with any other agent is $1$ due to \eqref{eq:deterministic utilities}. However, no agent has an incentive to unilaterally deviate by deciding to join only the second candidate coalition because this would have decreased her utility to $0$ since an agent receives $0$ utility from being alone in a coalition. Accordingly, forming the grand coalition within the first candidate coalition is a Nash-stable pure strategy.
    
            \item \textbf{\underline{No other pure joint strategy is Nash-stable:}} Consider a joint action $\mathbf{a}$ with either $|C_1^{\mathbf{a}}|=1$ or $3\leq |C_1^{\mathbf{a}}| \leq 5$. We distinguish between the following cases:
            \begin{enumerate}
                \item \textbf{$|C_1^{\mathbf{a}}|=1$:} That is, the first candidate coalition consists of exactly $1$ agent, while all other agents choose to join the second candidate coalition. Each agent in the second candidate coalition receives a utility of $-\frac{1}{8}\cdot4=-\frac{1}{2}$ from that coalition. Further, as $v_{i,i}=0$ for any agent $i$, the agent in $C_1^{\mathbf{a}}$ gets $0$ utility from the partition induced by $\mathbf{a}$. By \eqref{eq:deterministic utilities}, any agent $i\in C_2^{\mathbf{a}}$ thus has an incentive to unilaterally deviate by joining the first candidate coalition. Particularly, this increases her utility to $1$ since she receives a utility of $1$ from the agent currently in $C_1^{\mathbf{a}}$.
                
                \item \textbf{$|C_1^{\mathbf{a}}|=3$:} That is, the first candidate coalition consists of exactly $3$ agents, while all other $3$ agents choose to join the second candidate coalition. By \eqref{eq:deterministic utilities}, any agent $i\in C_1^{\mathbf{a}}$ receives a utility of $-1-1=-2$ since her utility from any other agent is $-1$. As such, any agent $i\in C_1^{\mathbf{a}}$ has an incentive to unilaterally deviate by deciding to join the second candidate coalition instead, thus increasing her utility to $-\frac{1}{2}$. Further, any agent $i\in C_2^{\mathbf{a}}$ receives $-\frac{1}{2}$ utility, as she has a utility of $-\frac{1}{4}$ from any other agent in the second candidate coalition. By \eqref{eq:deterministic utilities}, any agent $i\in C_2^{\mathbf{a}}$ has an incentive to unilaterally deviate by joining the first candidate coalition. Particularly, this increases her utility to $2\cdot 1=1$ since she receives a utility of $1$ from any other coalition member.
    
                \item \textbf{$|C_1^{\mathbf{a}}|=4$:} That is, the first candidate coalition consists of exactly $4$ agents, while the remaining $2$ agents choose to join the second candidate coalition. Any agent $i\in C_2^{\mathbf{a}}$ receives $-\frac{1}{2}$ utility, as she has a utility of $-\frac{1}{2}$ from the other agent in the second candidate coalition. By \eqref{eq:deterministic utilities}, any agent $i\in C_2^{\mathbf{a}}$ has an incentive to unilaterally deviate by joining the first candidate coalition. Particularly, this increases her utility to $3\cdot 1=3$ since she receives a utility of $1$ from any other coalition member.
    
                \item \textbf{$|C_1^{\mathbf{a}}|=5$:} That is, the first candidate coalition consists of exactly $5$ agents, while one agent chooses to join the second candidate coalition. This one agent $i\in C_2^{\mathbf{a}}$ receives $0$ utility. By \eqref{eq:deterministic utilities}, the agent $i\in C_2^{\mathbf{a}}$ has an incentive to unilaterally deviate by joining the first candidate coalition. Particularly, this increases her utility to $5\cdot 1=5$ since she receives a utility of $1$ from any other coalition member. In fact, after the agent $i\in C_2^{\mathbf{a}}$ unilaterally deviates, the grand coalition is formed within the first candidate coalition, which we have already proven to be Nash-stable.
            \end{enumerate}
        \end{enumerate}
    
        In summary, the pure NS strategies of the first game $G_1$ consist only of all pure joint strategies where either exactly $2$ agents join the first candidate coalition or the grand coalition is formed within the first candidate coalition.
    \end{proof}

    \paragraph{Construction of the second game $G_2$.} In the second game, there is only one type of pure NS strategies, comprising all pure joint strategies where only $5$ agents join the first candidate coalition. Specifically, for each pair of distinct agents $i,j$ and any joint action $\mathbf{a}$, we define the mutual utility of agents $i,j$ within the first and second candidate coalitions is determined deterministically via $\tilde{\mathcal{D}}_{i,j}^1(\mathbf{a})$ and $\tilde{\mathcal{D}}_{i,j}^2(\mathbf{a})$ as follows (respectively):
    \begin{equation}
    \label{eq:deterministic utilities game 2}
        \begin{aligned}
            \tilde{\mathcal{D}}^1_{i,j}(\mathbf{a}) &=\begin{cases}
            1 &, \text{ if } i,j \in C_1^{\mathbf{a}} \text{ and } |C_1^{\mathbf{a}}|=2 \\
            1 &, \text{ if } i,j \in C_1^{\mathbf{a}} \text{ and } |C_1^{\mathbf{a}}|=3 \\
            1 &, \text{ if } i,j \in C_1^{\mathbf{a}} \text{ and } |C_1^{\mathbf{a}}|=4 \\
            1 &, \text{ if } i,j \in C_1^{\mathbf{a}} \text{ and } |C_1^{\mathbf{a}}|=5 \\
            -1 &, \text{ if } i,j \in C_1^{\mathbf{a}} \text{ and } |C_1^{\mathbf{a}}|=6 
        \end{cases}\\
        \tilde{\mathcal{D}}_{i,j}^2(\mathbf{a}) &=\begin{cases}
            -\frac{1}{2} &, \text{ if } i,j \in C_2^{\mathbf{a}} \text{ and } |C_2^{\mathbf{a}}|=2 \\
            -\frac{1}{4} &, \text{ if } i,j \in C_2^{\mathbf{a}} \text{ and } |C_2^{\mathbf{a}}|=3 \\
            -\frac{1}{6} &, \text{ if } i,j \in C_2^{\mathbf{a}} \text{ and } |C_2^{\mathbf{a}}|=4 \\
            -\frac{1}{8} &, \text{ if } i,j \in C_2^{\mathbf{a}} \text{ and } |C_2^{\mathbf{a}}|=5 \\
            -\frac{1}{10} &, \text{ if } i,j \in C_2^{\mathbf{a}} \text{ and } |C_2^{\mathbf{a}}|=6 
        \end{cases}
        \end{aligned}
    \end{equation}
    Next, we formally characterize all the pure NS strategies of the second game $G_2$:
    \begin{lemma}
        \label{supp:lemma:NS game 2}
        The pure NS strategies of the second game $G_2$ consist only of all pure joint strategies where exactly $5$ agents join the first candidate coalition.
    \end{lemma}
    \begin{proof}
         We first prove that all pure joint strategies that are as in the above are Nash stable, and then show that no other pure joint strategy is Nash-stable. Formally:
        \begin{enumerate}
        \item \textbf{\underline{All pure strategies where only $5$ agents join the first candidate coalition are Nash-stable:}} Consider a joint action $\mathbf{a}$ with $|C_1^{\mathbf{a}}| = 5$. That is, the first candidate coalition consists of exactly $5$ agents, while one agent chooses to join the second candidate coalition. Now, we will prove that no agent has an incentive to unilaterally deviate from $\mathbf{a}$:
        \begin{enumerate}
            \item \textbf{The single agent $i\in C_2^{\mathbf{a}}$ has no incentive to unilaterally deviate:} Agent $i$ receives $0$ utility from being alone in the second candidate coalition. By \eqref{eq:deterministic utilities game 2}, if agent $i$ unilaterally deviates, then she decreases her utility to to $5\cdot (-1)=-5$ since she receives a utility of $-1$ from any other coalition member, and thus agent $i$ has no incentive to unilaterally deviate.

            \item \textbf{No agent $i\in C_1^{\mathbf{a}}$ has an incentive to unilaterally deviate:} Each agent $i\in C_1^{\mathbf{a}}$ obtains a utility of $4 \cdot 1=4$, as she receives a utility of $1$ from any other coalition member. If an agent $i\in C_1^{\mathbf{a}}$ unilaterally deviates, then she decides to join the second candidate coalition instead, meaning that her utility decreases to $-\frac{1}{2}$ since she will obtain a utility of $-\frac{1}{2}$ from the single agent in $C_2^{\mathbf{a}}$. Therefore, each agent $i\in C_1^{\mathbf{a}}$ has no incentive to unilaterally deviate.
        \end{enumerate}

         \item \textbf{\underline{No other pure joint strategy is Nash-stable:}} Consider a joint action $\mathbf{a}$ with either $1\leq |C_1^{\mathbf{a}}| \leq 4$ or $|C_1^{\mathbf{a}}|=6$. We distinguish between the following cases:
         \begin{enumerate}
             \item \textbf{$|C_1^{\mathbf{a}}|=1$:} That is, the first candidate coalition consists of exactly $1$ agent, while all other agents choose to join the second candidate coalition. Each agent in the second candidate coalition receives a utility of $-\frac{1}{8}\cdot4=-\frac{1}{2}$ from that coalition. Further, as $v_{i,i}=0$ for any agent $i$, the agent in $C_1^{\mathbf{a}}$ gets $0$ utility from the partition induced by $\mathbf{a}$. By \eqref{eq:deterministic utilities}, any agent $i\in C_2^{\mathbf{a}}$ thus has an incentive to unilaterally deviate by joining the first candidate coalition. Particularly, this will increase her utility to $1$ since she will receive a utility of $1$ from the agent currently in $C_1^{\mathbf{a}}$.

             \item \textbf{$3 \leq |C_1^{\mathbf{a}}|\leq 5$:} That is, there are agents that choose to join the second candidate coalition. In any case, each such agent $i\in C_2^{\mathbf{a}}$ receives a utility of $-\frac{1}{2}$. By \eqref{eq:deterministic utilities}, any agent $i\in C_2^{\mathbf{a}}$ thus has an incentive to unilaterally deviate by joining the first candidate coalition instead. Particularly, this increases her utility to $2$ if $|C_1^{\mathbf{a}}|=3$, $3$ if $|C_1^{\mathbf{a}}|=4$ and $4$ if $|C_1^{\mathbf{a}}|=5$ since the utility that she receives from any other coalition member is $1$ in each case.
             
            \item \textbf{$|C_1^{\mathbf{a}}|=6$:} That is, the grand coalition is formed within the first candidate coalition (i.e., all agents join the first candidate coalition). Particularly, this strategy corresponds to the joint action given by $\mathbf{a}=(\{1\}, \{1\}, \{1\}, \{1\}, \{1\}, \{1\})$. Due to \eqref{eq:deterministic utilities}, each agent $i$ receives a utility of $5\cdot (-1)=-5$ from the partition induced by $\mathbf{a}$, as she obtains a utility of $-1$ from any other coalition member. Thereby, each agent $i$ has an incentive to unilaterally deviate by deciding to join the second candidate coalition, thus increasing her utility to $0$.
         \end{enumerate}
    \end{enumerate}

    In summary, the pure NS strategies of the second game $G_2$ consist solely of all pure joint strategies where only $5$ agents join the first candidate coalition.
    \end{proof}

    \paragraph{Construction of the exploration policy $\rho$.} For both games $G_1, G_2$, we construct the same exploration policy $\rho$, which picks a joint action uniformly at random from the set of all joint actions where the first candidate coalition consists of exactly $2$, $4$ or $5$ agents, whereas all other joint actions are assigned zero probability under $\rho$. Namely, the exploration policy $\rho$ picks joint action uniformly at random from $\mathcal{B}:=\{\mathbf{a} \in \mathcal{A}: |C_1^{\mathbf{a}}|=2 \text{ or } |C_1^{\mathbf{a}}|=4 \text{ or } |C_1^{\mathbf{a}}|=5\}$ and assigns zero probability to any joint action in $\mathcal{A}\setminus\mathcal{B}$. Since $|\mathcal{B}| = \binom{6}{2}+\binom{6}{4}+\binom{6}{5}=36$, the exploration policy $\rho$ is thereby given by:
    \begin{equation}
        \label{eq:exploration policy}
        \rho(\mathbf{a})=\begin{cases}
            \frac{1}{36} &, \mathbf{a} \in \mathcal{B}=\{\mathbf{a} \in \mathcal{A}: |C_1^{\mathbf{a}}|=2 \text{ or } |C_1^{\mathbf{a}}|=4 \text{ or } |C_1^{\mathbf{a}}|=5\} \\
            0 &, \text{ o.w.}
        \end{cases}
    \end{equation}

    \paragraph{Proving the statement in Theorem \ref{supp:thm:half gap}.} Subsequently, we prove the statement in Theorem \ref{supp:thm:half gap}. First, we show that the pairs $(G_1,\rho)$ and $(G_2,\rho)$ both belong to the class $\mathcal{G}$ stated in Theorem \ref{supp:thm:half gap}, i.e., both games with the exploration policy $\rho$ satisfy Assumption 1 in the full paper, except for at most one coalition size $\alpha\in[n]\cup\{0\}$. Note that the exploration policy $\rho$ in \eqref{eq:exploration policy} can only cover the following pure joint strategies:
    \begin{itemize}
        \item \underline{For the first game $G_1$:} The exploration policy $\rho$ in \eqref{eq:exploration policy} can only cover all NS pure joint strategies where exactly $2$ agents join the first candidate coalition, along with any unilateral deviation where one of those $2$ agents chooses to opt out of the coalition and instead join the second candidate coalition. However, the exploration policy $\rho$ does \textit{not} cover any unilateral deviation from such NS strategies that results in a pure joint strategy where exactly $3$ agents join the first candidate coalition. Hence, $(G_1,\rho)$ satisfies Assumption 1 in the full paper except for a coalition size of $3$, yielding that $(G_1,\rho) \in \mathcal{G}$.

        \item \underline{For the second game $G_2$:} The exploration policy $\rho$ in \eqref{eq:exploration policy} can only cover all NS pure joint strategies as specified in Lemma \ref{supp:lemma:NS game 2}, together with any unilateral deviation that leads to a pure joint strategy where exactly $4$ agents join the first candidate coalition. However, the exploration policy $\rho$ does \textit{not} cover any unilateral deviation that results in the pure joint strategy where the grand coalition is formed within the first candidate coalition. Hence, $(G_1,\rho)$ satisfies Assumption 1 in the full paper except for a coalition size of $6$, yielding that $(G_2,\rho) \in \mathcal{G}$.
    \end{itemize}  
    
    As such, consider any algorithm $\alg$. As mentioned earlier, the exploration policy $\rho$ fails to cover unilateral deviations from NS pure joint strategies that lead to a pure joint strategy where the first candidate coalition has size $3$ for game $G_1$, and size $6$ (i.e., the grand coalition) for game $G_2$. In particular, this also means that $\rho$ does \textit{not} cover the pure joint strategy of forming the grand coalition within the first candidate coalition, which is Nash-stable in game $G_1$ by Lemma \ref{supp:lemma:NS game 1}. Those unilateral deviations are crucial to differentiate the NS pure joint strategies of both games. In fact, by \eqref{eq:deterministic utilities} and \eqref{eq:deterministic utilities game 2}, any joint action sampled from $\rho$ will result in the same utility feedbacks for both games. Hence, \textbf{regardless of the size of the dataset obtained by $\alg$, the algorithm $\alg$ \textit{cannot} distinguish between the games $G_1,G_2$ and they appear behaviorally the same from the perspective of the data available under $\rho$}. 

    Therefore, we denote by $q$ the probability that a joint action $\mathbf{a}$ with $|C_1^{\mathbf{a}}|=5$ is sampled from the joint strategy $\bm{\varphi}$ produced by $\alg$. For game $G_1$, a joint action $\mathbf{a}$ with $|C_1^{\mathbf{a}}|=5$ is \textit{not} a Nash-stable pure joint strategy by Lemma \ref{supp:lemma:NS game 1}. Accordingly, since the joint strategy $\bm{\varphi}$ produced by $\alg$ plays such a joint action with probability $q$, then $\bm{\varphi}$ picks a joint action for which some agent has an incentive to unilaterally deviate with probability $q$, thereby inducing a non-zero duality gap. However, with probability $1-q$, the joint strategy $\bm{\varphi}$ produced by $\alg$ may pick one of the Nash-stable strategies dictated by Lemma \ref{supp:lemma:NS game 1}, and thus yielding zero duality gap. Formally, based on Section 1 in the full paper, the utility of each agent $i$ from her strategy $\varphi_i$ is denoted as $V_i^1(\bm{\varphi}) := \mathbb{E}_{\mathbf{a} \sim \bm{\varphi}} [d_i^1(\mathbf{a})]$ under the game $G_1$. As such, by arguments similar to \eqref{eq:gap upper bound} in the proof of Lemma \ref{supp:lemma:surrogate}, we obtain the following for the game $G_1$:
    \begin{subequations}
        \label{eq:gap game 1}
        \begin{align}
            \gap(\bm{\varphi})& = \max_{i \in \mathcal{N}} \left[\max_{\phi_i \in \Delta(\mathcal{A}_i)} V_i^1(\bm{\varphi}_{-i},\phi_i) - V_i^1(\bm{\varphi})\right] = \max_{i \in \mathcal{N}} \left[\max_{\phi_i \in \Delta(\mathcal{A}_i)}\mathbb{E}_{\mathbf{a} \sim (\bm{\varphi}_{-i},\phi_i)} [d_i^1(\mathbf{a})] -\mathbb{E}_{\mathbf{a} \sim \bm{\varphi}} [d_i^1(\mathbf{a})]\right] \\
            &= \max_{i \in \mathcal{N}}  \left[\max_{\phi_i \in \Delta(\mathcal{A}_i)}\mathbb{E}_{a_i' \sim \phi_i} [\mathbb{E}_{\mathbf{a}_{-i} \sim \bm{\varphi}_{-i}}[\mathbb{E}_{a_i \sim \varphi_i}[d_i^1(\mathbf{a}_{-i},a_i')|d_i^1(\mathbf{a}_{-i},a_i)]]] -\mathbb{E}_{\mathbf{a} \sim \bm{\varphi}} [d_i^1(\mathbf{a})]\right] \label{eq:law of total exp} \\
            &= \max_{i \in \mathcal{N}}  \max_{\phi_i \in \Delta(\mathcal{A}_i)}\mathbb{E}_{a_i' \sim \phi_i} \left[\mathbb{E}_{\mathbf{a} \sim \bm{\varphi}}[d_i^1(\mathbf{a}_{-i},a_i')|d_i^1(\mathbf{a})]-\mathbb{E}_{\mathbf{a} \sim \bm{\varphi}} [d_i^1(\mathbf{a})]\right] \\
            &= \max_{i \in \mathcal{N}} \max_{\phi_i \in \Delta(\mathcal{A}_i)} \mathbb{E}_{a_i' \sim \phi_i} \Big[\left[\mathbb{E}_{\mathbf{a} \sim \bm{\varphi}}[d_i^1(\mathbf{a}_{-i},a_i')|d_i^1(\mathbf{a}),|C_1^{\mathbf{a}}|=5]-\mathbb{E}_{\mathbf{a} \sim \bm{\varphi}} [d_i^1(\mathbf{a})||C_1^{\mathbf{a}}|=5]\right] \cdot \mathbb{P}_{\mathbf{a} \sim \bm{\varphi}}[|C_1^{\mathbf{a}}|=5] \\
            &+ \left[\mathbb{E}_{\mathbf{a} \sim \bm{\varphi}}[d_i^1(\mathbf{a}_{-i},a_i')|d_i^1(\mathbf{a}), |C_1^{\mathbf{a}}|\neq 5]-\mathbb{E}_{\mathbf{a} \sim \bm{\varphi}} [d_i^1(\mathbf{a})||C_1^{\mathbf{a}}|\neq 5]\right]\cdot \mathbb{P}_{\mathbf{a} \sim \bm{\varphi}}[|C_1^{\mathbf{a}}|\neq 5]\Big] \\
            &= q\max_{i \in \mathcal{N}} \max_{\phi_i \in \Delta(\mathcal{A}_i)} \mathbb{E}_{a_i' \sim \phi_i} \left[\left[\mathbb{E}_{\mathbf{a} \sim \bm{\varphi}}[d_i^1(\mathbf{a}_{-i},a_i')|d_i^1(\mathbf{a}),|C_1^{\mathbf{a}}|=5]-\mathbb{E}_{\mathbf{a} \sim \bm{\varphi}} [d_i^1(\mathbf{a})||C_1^{\mathbf{a}}|=5]\right] \right] \label{eq:prob 1} \\
            &+ (1-q) \max_{i \in \mathcal{N}} \max_{\phi_i \in \Delta(\mathcal{A}_i)} \mathbb{E}_{a_i' \sim \phi_i} \left[\left[\mathbb{E}_{\mathbf{a} \sim \bm{\varphi}}[d_i^1(\mathbf{a}_{-i},a_i')|d_i^1(\mathbf{a}), |C_1^{\mathbf{a}}|\neq 5]-\mathbb{E}_{\mathbf{a} \sim \bm{\varphi}} [d_i^1(\mathbf{a})||C_1^{\mathbf{a}}|\neq 5]\right] \right]\\
            &\geq q\label{eq:equality 5}
        \end{align}
    \end{subequations}
    where \eqref{eq:law of total exp} is due to the law of total expectation and rearranging the external expectation. Moreover, \eqref{eq:prob 1} follows from $\mathbb{P}_{\mathbf{a} \sim \bm{\varphi}}[|C_1^{\mathbf{a}}|=5]=q$ and $\mathbb{P}_{\mathbf{a} \sim \bm{\varphi}}[|C_1^{\mathbf{a}}|\neq 5]=1-q$. Further, the inequality in \eqref{eq:equality 5} stems from the fact that, for any joint action $\mathbf{a}$ with $|C_1^{\mathbf{a}}| = 5$, each agent $i$ obtains a utility of $4\cdot 1=4$ since she receives a utility of $1$ from any other coalition member due to \eqref{eq:deterministic utilities}, yielding that $ \mathbb{E}_{\mathbf{a} \sim \bm{\varphi}} \left[d_i^1(\mathbf{a})\middle| |C_1^{\mathbf{a}}|=5\right]=4$. However, for game $G_1$, a joint action $\mathbf{a}$ with $|C_1^{\mathbf{a}}|=5$ is \textit{not} a Nash-stable pure joint strategy by Lemma \ref{supp:lemma:NS game 1}. Particularly, the single agent that remains in $N\setminus C_1^{\mathbf{a}}$, can unilaterally deviate by joining the first candidate coalition, thus forming the grand coalition within the first candidate coalition and increasing her utility to $5\cdot 1=5$, as her utility from any other coalition member is $1$ due to \eqref{eq:deterministic utilities}. This means that $\max_{i \in \mathcal{N}} \max_{\phi_i \in \Delta(\mathcal{A}_i)} \mathbb{E}_{a_i' \sim \phi_i} \left[\left[\mathbb{E}_{\mathbf{a} \sim \bm{\varphi}}[d_i^1(\mathbf{a}_{-i},a_i')|d_i^1(\mathbf{a}),|C_1^{\mathbf{a}}|=5]-\mathbb{E}_{\mathbf{a} \sim \bm{\varphi}} [d_i^1(\mathbf{a})||C_1^{\mathbf{a}}|=5]\right] \right] \geq 1$, which implies \eqref{eq:equality 5}. Conversely, consider any joint action $\mathbf{a}$ with $|C_1^{\mathbf{a}}| \neq 5$. If $\mathbf{a}$ induces an NS partition, then the duality gap is $0$; otherwise, the duality gap is positive by argument similar to the above. As such, it holds that $\max_{i \in \mathcal{N}} \max_{\phi_i \in \Delta(\mathcal{A}_i)} \mathbb{E}_{a_i' \sim \phi_i} \left[\left[\mathbb{E}_{\mathbf{a} \sim \bm{\varphi}}[d_i^1(\mathbf{a}_{-i},a_i')|d_i^1(\mathbf{a}), |C_1^{\mathbf{a}}|\neq 5]-\mathbb{E}_{\mathbf{a} \sim \bm{\varphi}} [d_i^1(\mathbf{a})||C_1^{\mathbf{a}}|\neq 5]\right] \right] \geq 0$, yielding \eqref{eq:equality 5}.

    We therefore infer that \textbf{the joint strategy $\bm{\varphi}$ produced by $\alg$ satisfies $\gap(\bm{\varphi})\geq q$ the game $G_1$.}

    In contrast, for game $G_2$, a joint action $\mathbf{a}$ with $|C_1^{\mathbf{a}}|=5$ is a Nash-stable pure joint strategy by Lemma \ref{supp:lemma:NS game 2}. Hence, as the joint strategy $\bm{\varphi}$ produced by $\alg$ plays such a joint action with probability $q$, then $\bm{\varphi}$ picks an NS strategy with probability $q$, resulting in zero duality gap. Yet, with probability $1-q$, the joint strategy $\bm{\varphi}$ produced by $\alg$ is \textit{not} a Nash-stable strategy as specified by Lemma \ref{supp:lemma:NS game 2}, meaning that some agent has an incentive to unilaterally deviate and yield a non-zero duality gap. Thus, by arguments similar to \eqref{eq:gap game 1}, \textbf{the joint strategy $\bm{\varphi}$ produced by $\alg$ satisfies $\gap(\bm{\varphi})\geq 1-q$ for the game $G_2$.}

    Finally, since either $q\geq \frac{1}{2}$ or $1-q\geq \frac{1}{2}$, the joint strategy $\bm{\varphi}$ produced by $\alg$ is always satisfies $\gap(\bm{\varphi})\geq \frac{1}{2}$ for at least one game, regardless of the dataset size, as desired.   
    
\end{proof}

\section{Proof of Theorem 3}
The proof of Theorem 3 under semi-bandit feedback relies on Lemma \ref{supp:lemma:bound on diffs semi bandit}, which leverages the following well-known results concerning Hoeffding’s inequality to prove that, for any agent $i$, our constructed $b_i^\delta$ is an exploration bonus for the utility estimator $\hat{v}_i$. Beforehand, we recall that a joint action $\mathbf{a}$ induces a partition of the agents $\pi^{\mathbf{a}} = (C_\ell^{\mathbf{a}})_{\ell \in [k]}$ with possibly overlapping coalitions, where, for any $\ell \in [k]$, $C_\ell^{\mathbf{a}}$ is the set of agents joining the $\ell$th candidate coalition, i.e., $C_\ell^{\mathbf{a}} = \{i \in \mathcal{N}: \ell\in a_i \}$.
\begin{lemma}[\textbf{Hoeffding’s Inequality}]
    \label{supp:lemma:Hoeffding’s Inequality}
    Let $\mathcal{Y}_1, \dots, \mathcal{Y}_K$ be $K$ i.i.d. random variables and set $\bar{\mathcal{Y}}_K = \frac{1}{K} \sum_{k\in [K]} \mathcal{Y}_k$. For any $k \in [K]$, assume that there are constants $a,b$ such that $Pr[a \leq \mathcal{Y}_k \leq b] = 1$ and denote $\mu := \mathbb{E}[\mathcal{Y}_k]; c := (b-a)^2$. Then, for any $\delta \in (0,1]$, the following holds with probability at least $1-\delta$:
    \begin{equation}
        \label{eq:Hoeffding semi bandit}
        |\bar{\mathcal{Y}}_K - \mu| \leq \sqrt{\frac{c \log(2/\delta)}{2K}}
    \end{equation}
\end{lemma}

\begin{lemma}
    \label{supp:lemma:bound on diffs semi bandit}
    For any $\delta \in (0,1]$, the following holds with probability at least $1-\delta$ simultaneously for each agent $i$ and any joint action $\mathbf{a}$:
    \begin{equation}
        \label{eq:variance approx lemma 4 semi bandit}
        |\hat{v}_i(\mathbf{a}) - d_i(\mathbf{a})| \leq  \sum_{\ell\in a_i} \sum_{i \neq j \in C_\ell^{\mathbf{a}}: N_{i,j}^\ell \geq 1} \sqrt{\frac{2  \log(4(n+1)k/\delta)}{N_{i,j}^\ell\vee 1}} = b_i(\mathbf{a})
    \end{equation}
\end{lemma}
\begin{proof}
    Recall that our utility estimator under semi-bandit feedback is given by:
    \begin{equation}
        \label{eq:grad approx recall lemma 2 semi bandit}
        \begin{aligned}
            \hat{v}_i(\mathbf{a}) &= \sum_{\ell\in a_i} \sum_{i \neq j \in C_\ell^{\mathbf{a}}} \hat{v}_i^\ell(j)  = \sum_{\ell\in a_i} \sum_{i \neq j \in C_\ell^{\mathbf{a}}} \frac{\sum_{m=1}^M v_{i,j}^{\ell,m} \mathds{1}\{\ell \in a_i^m\cap a_j^m \}}{N_{i,j}^\ell \vee 1} \\
            &=  \sum_{\ell\in a_i} \sum_{i \neq j \in C_\ell^{\mathbf{a}}: N_{i,j}^\ell \geq 1} \frac{\sum_{m=1}^M v_{i,j}^{\ell,m} \mathds{1}\{\ell \in a_i^m\cap a_j^m \}}{N_{i,j}^\ell \vee 1} 
        \end{aligned}
    \end{equation}
    where the last equality is due to the fact that each addend for which $N_{i,j}^\ell = 0$ is essentially zero, i.e., if $N_{i,j}^\ell=0$ for some agent $j \neq i$, then $\frac{\sum_{m=1}^M v_{i,j}^{\ell,m} \mathds{1}\{\ell \in a_i^m\cap a_j^m \}}{N_{i,j}^\ell \vee 1}  = 0$. Specifically, if $N_{i,j}^\ell = 0$ for any agent $i \neq j \in \cup_{C \in \pi^{\mathbf{a}}(i)} C$ and $\ell \in a_i$, then our statement clearly holds since $\hat{v}_i(\mathbf{a}) = 0$. Hereafter, for any $\ell \in a_i$, we therefore assume that $N_{i,j}^\ell \geq 1$ for at least one agent $i \neq j \in C_\ell^{\mathbf{a}}$. As such, we will apply Hoeffding's inequality from Lemma \ref{supp:lemma:Hoeffding’s Inequality} to agent $i$'s utility estimator for each other agent $i \neq j \in \cup_{C \in \pi^{\mathbf{a}}(i)} C$ and then sum those bounds to obtain the desired result. As such, denoting $\mathcal{Y}_{i,j}^{\ell,k} = v_{i,j}^{\ell,m}$ for each  $m\in[M]$ where $\mathds{1}\{\ell \in a_i^m\cap a_j^m \} = 1$, then we can apply Lemma \ref{supp:lemma:Hoeffding’s Inequality} to the sequence $\{\mathcal{Y}_{i,j}^{\ell,k}\}_{m\in[M]: \mathds{1}\{\ell \in a_i^m\cap a_j^m \} = 1}$ for each fixed $\ell \in a_i^m$ and any fixed agent $i \neq j \in C_\ell^{\mathbf{a}}$, which comprises of $N_{i,j}^\ell$ elements. Further, recall that $\mathbb{E}[\hat{v}_{i}^\ell(j)] = d_{i,j}^\ell$, where $\mathbb{E}[\cdot]$ is the expectation over all the utilities' randomness. For any $m\in[M]$, observe that $Pr[-1 \leq \mathcal{Y}_{i,j}^{\ell,k} \leq 1] = 1$ due to our assumption that $v_{i,j}^{\ell,m} \in [-1,1]$. 
    Accordingly, for each fixed $m\in[M]$ and for any $\delta \in (0,1]$ we can apply Hoeffding’s inequality as in \eqref{eq:Hoeffding semi bandit} to obtain that the following holds with probability at least $1-\delta$ for each fixed agent $i$ and fixed joint action $\mathbf{a}$:
    \begin{equation}
        \label{eq:variance grad proof lemma 2 semi bandit}
        |\hat{v}_i^\ell(j) - d_{i,j}^\ell| \leq \sqrt{\frac{4 \log(2/\delta)}{2 N_{i,j}^\ell}} = \sqrt{\frac{2 \log(2/\delta)}{N_{i,j}^\ell\vee 1}}
    \end{equation}
    Further, note that:
    \begin{equation}
        \label{eq:diff estimate mean}
        \begin{aligned}
            \sum_{\ell\in a_i} \sum_{i \neq j \in C_\ell^{\mathbf{a}}} [\hat{v}_i^\ell(j) - d_{i,j}^\ell] &=  \sum_{\ell\in a_i} \sum_{i \neq j \in C_\ell^{\mathbf{a}}: N_{i,j}^\ell \geq 1}  [\hat{v}_i^\ell(j) - d_{i,j}^\ell] - \sum_{\ell\in a_i} \sum_{i \neq j \in C: N_{i,j}^\ell = 0}  d_{i,j}^\ell \\
            &\leq  \sum_{\ell\in a_i} \sum_{i \neq j \in C_\ell^{\mathbf{a}}: N_{i,j}^\ell \geq 1} [\hat{v}_i^\ell(j) - d_{i,j}^\ell]
        \end{aligned}
    \end{equation}
    since $\hat{v}_i^\ell(j) = \frac{\sum_{m=1}^M v_{i,j}^{\ell,m} \mathds{1}\{\ell \in a_i^m\cap a_j^m \}}{N_{i,j}^\ell \vee 1}=0$ when $N_{i,j}^\ell=0$. Hence, by \eqref{eq:grad approx recall lemma 2 semi bandit}, the following holds with probability at least $1-\delta$ for each fixed agent $i$ and fixed joint action $\mathbf{a}$:
    \begin{equation}
        \label{eq:variance grad proof lemma 2 semi bandit 2}
        \begin{aligned}
            |\hat{v}_i(\mathbf{a}) - d_i(\mathbf{a})| = \Bigg|\sum_{\ell\in a_i} \sum_{i \neq j \in C} [\hat{v}_i^\ell(j) - d_{i,j}^\ell]\Bigg| \leq \Bigg| \sum_{\ell\in a_i} \sum_{i \neq j \in C_\ell^{\mathbf{a}}: N_{i,j}^\ell \geq 1}[\hat{v}_i^\ell(j) - d_{i,j}^\ell]\Bigg| \leq \\
            \leq  \sum_{\ell\in a_i} \sum_{i \neq j \in C_\ell^{\mathbf{a}}: N_{i,j}^\ell \geq 1} |\hat{v}_i^\ell(j) - d_{i,j}^\ell| \leq  \sum_{\ell\in a_i} \sum_{i \neq j \in C_\ell^{\mathbf{a}}: N_{i,j}^\ell \geq 1} \sqrt{\frac{2 \log(2/\delta)}{N_{i,j}^\ell\vee 1}}
        \end{aligned}
    \end{equation}
    where the first inequality follows from \eqref{eq:diff estimate mean} and in the second one we used the triangle inequality. One can then apply a union bound to obtain \eqref{eq:variance approx lemma 4 semi bandit}. 
\end{proof}

Theorem 3's proof also relies on Lemma \ref{supp:lemma:counter}, which is derived from Lemma \ref{supp:lemma:binomial concentration}:

\begin{lemma}[\textbf{\cite[Lemma A.1]{xie2021policy}}]
    \label{supp:lemma:binomial concentration}
    Suppose $N\sim \text{Bin}(m,p)$ where $m\geq 1$ and $p\in[0,1]$. Then, for any $\delta\in(0,1]$, we have $\frac{p}{N\vee 1} \leq \frac{8\log(1/\delta)}{M}$ with probability at least $1-\delta$.
\end{lemma}

\begin{lemma}
    \label{supp:lemma:counter}
    Given a dataset of size $M$, for any $\delta \in (0,1]$, the following holds with probability at least $1-\delta$ simultaneously for each agent $i$ and any joint action $\mathbf{a}$:
    \begin{equation}
        \label{eq:frac counter lemma}
        \frac{1}{N_{i,j}^\ell\vee 1} \leq \frac{8  \log(2(n+1)k/\delta)}{M d_\ell^{\rho}(i,j)}
    \end{equation}
    where $d_\ell^{\rho}(i,j) = \sum_{\mathbf{a}\in\mathcal{A}: i,j\in C_\ell^{\mathbf{a}}} \rho(\mathbf{a})$ for any pair of distinct agents $i,j$.
\end{lemma}
\begin{proof}
    The following holds with probability at least $1-\delta$ by plugging $p=d_\ell^{\rho}(i,j)$ in Lemma \ref{supp:lemma:binomial concentration} and applying a union bound:
    \begin{equation}
        \label{eq:frac counter}
        \frac{1}{N_{i,j}^\ell\vee 1} \leq \frac{8  \log(2(n+1)k/\delta)}{M d_\ell^{\rho}(i,j)}
    \end{equation}

\end{proof}

We are now ready to prove Theorem 3:
\begin{customthm}{3}
    \label{supp:thm:semi-bandit}
    Under semi-bandit feedback and Assumption 1, for any $\delta \in (0,1]$, any dataset size $M \in \mathbb{N}$ and any NS joint strategy $\bm{\varphi}^\star$, the joint mixed strategy $\varphi^{\out}$ produced by Algorithm 1 with utility estimators and exploration bonuses defined in equations (8) and (9) within the full paper satisfies the following with probability at least $1-\delta$: 
    \begin{equation}
        \gap(\varphi^{\out}) \leq 8kn(n+1) c_{\size}^{\bm{\varphi}^\star} \log(4(n+1)k/\delta) \sqrt{\frac{2(n-1)}{M}} \left[\frac{n-1}{2} +\sqrt{\frac{n}{2}}\right] +\epsilon_{\opt}
    \end{equation}
    where $c_{\size}^{\bm{\varphi}^\star}$ is defined in equation (7) within the full paper. 
\end{customthm}
\begin{proof}
    By Theorem \ref{supp:thm:general duality gap}, we need to bound $\mathbb{E}_{\mathbf{a}\sim (\bm{\varphi}_{-i}^\star,\varphi_i')} [b_i^\delta(\mathbf{a})]$ and $ \mathbb{E}_{\mathbf{a}\sim \bm{\varphi}^\star} [b_i^\delta(\mathbf{a})]$ for some NS (possibly mixed) joint strategy $\bm{\varphi}^\star$ and any strategy $\varphi_i \in \Delta(\mathcal{A}_i)$ (in particular, we can $\varphi_i$ choose to be deterministic by Theorem \ref{supp:thm:general duality gap}). Indeed, for any $\delta \in (0,1]$, the following holds with probability at least $1-\delta$ simultaneously for each agent $i$:

    \begin{subequations}
        \label{eq:bound bonus}
        \begin{align}
            \mathbb{E}_{\mathbf{a}\sim (\bm{\varphi}_{-i}^\star,\varphi_i')} [b_i^\delta(\mathbf{a})] &= \mathbb{E}_{\mathbf{a}\sim (\bm{\varphi}_{-i}^\star,\varphi_i')}\left[ \sum_{\ell\in a_i} \sum_{i \neq j \in C_\ell^{\mathbf{a}}: N_{i,j}^\ell \geq 1} \sqrt{\frac{2  \log(4(n+1)k/\delta)}{N_{i,j}^\ell\vee 1}}\right] \\
            &\leq  \mathbb{E}_{\mathbf{a}\sim (\bm{\varphi}_{-i}^\star,\varphi_i')}\left[\sum_{\ell\in [k]}\sum_{i \neq j \in \mathcal{N}: N_{i,j}^\ell \geq 1}\sqrt{\frac{2  \log(4(n+1)k/\delta)}{N_{i,j}^\ell\vee 1}}\right] \\
            &= \sum_{\ell\in [k]} \sum_{i \neq j \in \mathcal{N}: N_{i,j}^\ell \geq 1}\mathbb{E}_{\mathbf{a}\sim (\bm{\varphi}_{-i}^\star,\varphi_i')}\left[\sqrt{\frac{2  \log(4(n+1)k/\delta)}{N_{i,j}^\ell\vee 1}}\right] \label{eq:at most k} \\
            &\leq \sum_{\ell\in [k]} \sum_{i \neq j \in \mathcal{N}: N_{i,j}^\ell \geq 1}\sum_{\alpha=0}^n d^{(\bm{\varphi}_{-i}^\star,\varphi_i')}_\ell(\alpha) \sqrt{\frac{8  \log^2(4(n+1)k/\delta)}{M d_\ell^{\rho}(i,j)}}  \label{eq:apply lemma} \\
            &\leq c_{\size}^{\bm{\varphi}^\star} \sum_{\ell\in [k]}\sum_{i \neq j \in \mathcal{N}: N_{i,j}^\ell \geq 1}\sum_{\alpha=0}^n d_\ell^{\rho}(\alpha) \sqrt{\frac{8  \log^2(4(n+1)k/\delta)}{M d_\ell^{\rho}(i,j)}} \label{eq:assumption 1}
        \end{align}
    \end{subequations}
    where the inequality in \eqref{eq:apply lemma} is due to $d^{(\bm{\varphi}_{-i}^\star,\varphi_i')}_\ell(\alpha) = \sum_{\mathbf{a}\in\mathcal{A}: |C_\ell^{\mathbf{a}}|=\alpha} (\bm{\varphi}_{-i}^\star,\varphi_i')(\mathbf{a})$ and Lemma \ref{supp:lemma:counter}. Further, the inequality in \eqref{eq:assumption 1} is derived from Assumption 1.
    
    Now, recalling that $d_\ell^{\rho}(\alpha) = \sum_{\mathbf{a}\in\mathcal{A}:|C_\ell^{\mathbf{a}}|=\alpha} \rho(\mathbf{a})$ for any coalition of size $\alpha\in[n]\cup\{0\}$, observe that:
    \begin{equation}
        \label{eq:bound rho}
        d_\ell^{\rho}(\alpha) = \sum_{\mathbf{a}\in\mathcal{A}:|C_\ell^{\mathbf{a}}|=\alpha} \rho(\mathbf{a}) \leq \sum_{i,j \in \mathcal{N}: i \neq j} \text{ } \sum_{\mathbf{a}\in\mathcal{A}: |C_\ell^{\mathbf{a}}|=\alpha \wedge i,j\in C_\ell^{\mathbf{a}}} \rho(\mathbf{a}) \leq \sum_{i,j \in \mathcal{N}: i \neq j}d_\ell^{\rho}(i,j)
    \end{equation}
    where the first inequality follows from the fact that, for each $\mathbf{a}\in\mathcal{A}$ with $ C \in \pi^{\mathbf{a}}$ such that $|C|=\alpha$, $\rho(\mathbf{a})$ may be summed more than once in $\sum_{i,j \in \mathcal{N}: i \neq j} \text{ } \sum_{\mathbf{a}\in\mathcal{A}: |C_\ell^{\mathbf{a}}|=\alpha \wedge i,j\in C_\ell^{\mathbf{a}}} \rho(\mathbf{a})$. Moreover, the second inequality in \eqref{eq:bound rho} follows from the fact that, for any pair of distinct agents $i,j$, it holds that: 
    \begin{equation}
        d_\ell^{\rho}(i,j) = \sum_{\mathbf{a}\in\mathcal{A}: i,j\in C_\ell^{\mathbf{a}}} \rho(\mathbf{a}) \geq \sum_{\mathbf{a}\in\mathcal{A}: |C_\ell^{\mathbf{a}}|=\alpha \wedge i,j\in C_\ell^{\mathbf{a}}} \rho(\mathbf{a})
    \end{equation}
    
    Combining \eqref{eq:bound bonus} with \eqref{eq:bound rho}, we obtain:
    \begin{subequations}
        \label{eq:bound bonus cont}
        \begin{align}
            \mathbb{E}_{\mathbf{a}\sim (\bm{\varphi}_{-i}^\star,\varphi_i')} [b_i^\delta(\mathbf{a})] &\leq c_{\size}^{\bm{\varphi}^\star} \sum_{\ell\in [k]} \sum_{i \neq j \in \mathcal{N}: N_{i,j}^\ell \geq 1} \sqrt{\frac{8  \log^2(4(n+1)k/\delta)}{M d_\ell^{\rho}(i,j)}}  \cdot\left[\sum_{\alpha=0}^n \sum_{i',j' \in \mathcal{N}: i' \neq j'}d_\ell^{\rho}(i',j')\right] \\
            &=(n+1)c_{\size}^{\bm{\varphi}^\star} \sum_{\ell\in [k]} \sum_{i \neq j \in \mathcal{N}: N_{i,j}^\ell \geq 1} \sqrt{\frac{8  \log^2(4(n+1)k/\delta)}{M d_\ell^{\rho}(i,j)}}  \cdot \left[\sum_{i',j' \in \mathcal{N}: i' \neq j' \wedge d_\ell^{\rho}(i',j') \leq d_\ell^{\rho}(i,j)} d_\ell^{\rho}(i',j')\right] \\
            &+ (n+1)c_{\size}^{\bm{\varphi}^\star} \sum_{\ell\in [k]} \sum_{i \neq j \in \mathcal{N}: N_{i,j}^\ell \geq 1} \left[\sum_{i',j' \in \mathcal{N}: i' \neq j' \wedge d_\ell^{\rho}(i',j') > d_\ell^{\rho}(i,j)} \sqrt{\frac{8  \log^2(4(n+1)k/\delta)}{M d_\ell^{\rho}(i,j)}} \cdot d_\ell^{\rho}(i',j')\right] \\
            &\leq (n+1)c_{\size}^{\bm{\varphi}^\star} \sum_{\ell\in [k]} \sum_{i \neq j \in \mathcal{N}: N_{i,j}^\ell \geq 1} \sqrt{\frac{8  \log^2(4(n+1)k/\delta)}{M d_\ell^{\rho}(i,j)}}  \cdot \left[\sum_{i',j' \in \mathcal{N}: i' \neq j' \wedge d_\ell^{\rho}(i',j') \leq d_\ell^{\rho}(i,j)} d_\ell^{\rho}(i,j)\right] \\
            &+ (n+1)c_{\size}^{\bm{\varphi}^\star} \sum_{\ell\in [k]} \sum_{i \neq j \in \mathcal{N}: N_{i,j}^\ell \geq 1} \left[\sum_{i',j' \in \mathcal{N}: i' \neq j' \wedge d_\ell^{\rho}(i',j') > d_\ell^{\rho}(i,j)} \sqrt{\frac{8  \log^2(4(n+1)k/\delta)}{M d_\ell^{\rho}(i',j')}} \cdot d_\ell^{\rho}(i',j')\right]  \\
            &=(n+1)c_{\size}^{\bm{\varphi}^\star} \sum_{\ell\in [k]} \sum_{i \neq j \in \mathcal{N}: N_{i,j}^\ell \geq 1} \sum_{i',j' \in \mathcal{N}: i' \neq j' \wedge d_\ell^{\rho}(i',j') \leq d_\ell^{\rho}(i,j)} \sqrt{\frac{8  \log^2(4(n+1)k/\delta)d_\ell^{\rho}(i,j)}{M}} \\
            &+ (n+1)c_{\size}^{\bm{\varphi}^\star} \sum_{\ell\in [k]} \sum_{i \neq j \in \mathcal{N}: N_{i,j}^\ell \geq 1} \sum_{i',j' \in \mathcal{N}: i' \neq j' \wedge d_\ell^{\rho}(i',j') > d_\ell^{\rho}(i,j)} \sqrt{\frac{8  \log^2(4(n+1)k/\delta) d_\ell^{\rho}(i',j')}{M}} \\
            &\leq (n+1)c_{\size}^{\bm{\varphi}^\star} \sum_{\ell\in [k]} \sum_{i \neq j \in \mathcal{N}} \sum_{i',j' \in \mathcal{N}: i' \neq j'} \sqrt{\frac{8  \log^2(4(n+1)k/\delta)d_\ell^{\rho}(i,j)}{M}} \\
            &+ (n+1)c_{\size}^{\bm{\varphi}^\star} \sum_{\ell\in [k]} \sum_{i \neq j \in \mathcal{N}} \sum_{i',j' \in \mathcal{N}: i' \neq j'} \sqrt{\frac{8  \log^2(4(n+1)k/\delta) d_\ell^{\rho}(i',j')}{M}} \\
            &= \binom{n}{2} (n+1) c_{\size}^{\bm{\varphi}^\star} \sum_{\ell\in[k]} \sum_{i \neq j \in \mathcal{N}} \sqrt{\frac{8  \log^2(4(n+1)k/\delta)d_\ell^{\rho}(i,j)}{M}}\\
            &+ n (n+1)c_{\size}^{\bm{\varphi}^\star} \sum_{\ell\in [k]} \sum_{i',j' \in \mathcal{N}: i' \neq j'} \sqrt{\frac{8  \log^2(4(n+1)k/\delta)d_\ell^{\rho}(i',j')}{M}} \\
            &\leq \frac{n(n-1)}{2} (n+1) c_{\size}^{\bm{\varphi}^\star} \sum_{\ell\in[k]} \sqrt{\frac{8 (n-1) \log^2(4(n+1)k/\delta)}{M} \sum_{i \neq j \in \mathcal{N}} d_\ell^{\rho}(i,j)} \label{eq:cauchy1}\\
            &+ n (n+1)c_{\size}^{\bm{\varphi}^\star} \sum_{\ell\in [k]} \sqrt{\frac{8  \log^2(4(n+1)k/\delta) \binom{n}{2}}{M}\sum_{i',j' \in \mathcal{N}: i' \neq j'} d_\ell^{\rho}(i',j')}\label{eq:cauchy}\\
            &\leq 2kn(n+1) c_{\size}^{\bm{\varphi}^\star} \log(4(n+1)k/\delta) \sqrt{\frac{2(n-1)}{M}} \left[\frac{n-1}{2} +\sqrt{\frac{n}{2}}\right]\label{eq:1}
        \end{align}
    \end{subequations}
    where the inequalities in \eqref{eq:cauchy1} and \eqref{eq:cauchy} follow from the Cauchy-Schwarz inequality, according to which $\sum_{\ell=1}^n \sqrt{a_\ell} \leq \sqrt{n \sum_{\ell=1}^n a_\ell}$ for any sequence $a_1, \dots, a_n$. The last inequality in \eqref{eq:1} is due to:
    \begin{equation}
        \begin{aligned}
            \sum_{i \neq j \in \mathcal{N}} d_\ell^{\rho}(i,j) &\leq \sum_{i',j' \in \mathcal{N}: i'\neq j'} \text{ } \sum_{\mathbf{a}\in\mathcal{A}: \exists C \in \pi^{\mathbf{a}} \text{ s.t. } i',j'\in C} \rho(\mathbf{a}) = \sum_{\mathbf{a}\in\mathcal{A}} \rho(\mathbf{a})=1 \\
            \sum_{i',j' \in \mathcal{N}: i' \neq j'} d_\ell^{\rho}(i',j') &\leq \sum_{i',j' \in \mathcal{N}: i' \neq j'} \text{ } \sum_{\mathbf{a}\in\mathcal{A}: \exists C \in \pi^{\mathbf{a}} \text{ s.t. } i',j'\in C} \rho(\mathbf{a}) = \sum_{\mathbf{a}\in\mathcal{A}} \rho(\mathbf{a}) =1
        \end{aligned}
    \end{equation}
    By arguments similar to the above:
    \begin{equation}
        \mathbb{E}_{\mathbf{a}\sim \bm{\varphi}^\star} [b_i^\delta(\mathbf{a})]\leq 2kn(n+1) c_{\size}^{\bm{\varphi}^\star} \log(4(n+1)k/\delta) \sqrt{\frac{2(n-1)}{M}} \left[\frac{n-1}{2} +\sqrt{\frac{n}{2}}\right]
    \end{equation}
    Since our constructed $b_i^\delta$ is an exploration bonus for the utility estimator $\hat{v}_i$ by Lemma \ref{supp:lemma:bound on diffs semi bandit}, applying Theorem \ref{supp:thm:general duality gap} implies that:
    \begin{equation}
    \label{eq:gap out upper bound}
        \gap(\varphi^{\out}) \leq 8kn(n+1) c_{\size}^{\bm{\varphi}^\star} \log(4(n+1)k/\delta) \sqrt{\frac{2(n-1)}{M}} \left[\frac{n-1}{2} +\sqrt{\frac{n}{2}}\right] +\epsilon_{\opt}
    \end{equation}
    as desired.
\end{proof}

\subsection{The Minimum Value of Coalition Size Coefficient is at most $3$}

\begin{lemma}
    For a Nash stable joint strategy $\bm{\varphi}^\star$, the minimum value of $c_{\size}^{\bm{\varphi}^\star}$ is at most $3$. 
\end{lemma}
\begin{proof}
    Recall that:
    \begin{equation}
        \label{eq:coefficient}
        c_{\size}^{\bm{\varphi}^\star} = \max_{i\in \mathcal{N}, \ell \in [k], \phi_i\in\Delta(\mathcal{A}_i), \alpha\in [n]\cup\{0\}:d_\ell^{\rho}(\alpha) > 0} \frac{d_\ell^{\bm{\varphi}_{-i}^\star,\phi_i}(\alpha)}{d_\ell^{\rho}(\alpha)}
    \end{equation}
    By Theorem \ref{supp:thm:general duality gap}, $\bm{\varphi}^\star$ can be chosen to be \textit{deterministic}. Thus, assume that the exploration policy $\rho$ picks a joint action uniformly at random from the set of all joint actions reachable from $\bm{\varphi}^\star$ through unilateral deviations. Since $\bm{\varphi}^\star$ induces a pure joint strategy, then any unilateral deviation may increase or decrease the number of agents within a given coalition by at most $1$, as the agent that unilaterally deviates either remains in her current coalition, leaves that coalition or joins a new coalition. Formally, for any joint action $\mathbf{a}$, any agent $i$ and any $\ell\in[k]$, if agent $i$ unilaterally deviates to some action $a_i'$, then the size of the $\ell$-th candidate coalition can only be either of the following:
    \begin{enumerate}
        \item If agent $i$ originally joined the $\ell$-th candidate coalition (i.e., $\ell \in a_i$), but now she decides to leave it (i.e., $\ell \notin a_i'$), then the size of the $\ell$-th candidate coalition decreases by exactly $1$, meaning that $|C_\ell^{\mathbf{a}_{-i},a_i'}|=|C_\ell^{\mathbf{a}}|-1$.

        \item If agent $i$ originally joined the $\ell$-th candidate coalition (i.e., $\ell \in a_i$), and she decides to remain in that coalition (i.e., $\ell \in a_i'$), then the size of the $\ell$-th candidate coalition remains unchanged, meaning that $|C_\ell^{\mathbf{a}_{-i},a_i'}|=|C_\ell^{\mathbf{a}}|$.

        \item If agent $i$ originally did not join the $\ell$-th candidate coalition (i.e., $\ell \notin a_i$), but now she decides to join it (i.e., $\ell \in a_i'$), then the size of the $\ell$-th candidate coalition increases by exactly $1$, meaning that $|C_\ell^{\mathbf{a}_{-i},a_i'}|=|C_\ell^{\mathbf{a}}|+1$.
    \end{enumerate}

    Therefore, for any $\ell\in[k]$, any unilateral deviation only affects the size of the $\ell$-th candidate coalition within a range of $3$ possible values. As such, the minimum possible non-zero value of $d_\ell^{\rho}(\alpha)$ is at least $\frac{1}{3}$. Since the numerator in \eqref{eq:coefficient} is at most $1$, then the minimum value of $ c_{\size}^{\bm{\varphi}^\star}$ is at most $3$.
\end{proof}

\subsection{Proof of Corollary 1}

\begin{customcoro}{1}
    \label{supp:coro:optimal}
    Under {\normalfont semi-bandit} feedback and Assumption 1, for any $\delta \in (0,1]$, any $\varepsilon > \epsilon_{\opt}$ with $\epsilon_{\opt} = o(\varepsilon)$ and any NS joint {\normalfont mixed} strategy $\bm{\varphi}^\star$, Algorithm 1 with utility estimators and exploration bonuses defined in equations (8) and (9) within the full paper has a sample complexity bound with {\normalfont \textbf{optimal}} dependence on $\varepsilon$ (up to logarithmic factors): if the dataset size is at least
    \begin{equation}
        M \geq \frac{128k^2n^2(n+1)^2(n-1) (c_{\size}^{\bm{\varphi}^\star})^2}{(\varepsilon -\epsilon_{\opt})^2} \log^2(4(n+1)k/\delta)  \left[\frac{n-1}{2} +\sqrt{\frac{n}{2}}\right]^2,
    \end{equation}
    then the joint strategy $\varphi^{\out}$ is $\varepsilon$-NS (i.e., $\gap(\varphi^{\out}) \leq \varepsilon$) with probability at least $1-\delta$.
\end{customcoro}
\begin{proof}
    The proof is a continuation of the proof for Theorem \ref{supp:thm:semi-bandit}. Indeed, if the right-hand side in \eqref{eq:gap out upper bound} is at most $\varepsilon$ for some $\varepsilon \geq \epsilon_{\opt}$ with $\epsilon_{\opt} = o(\varepsilon)$, then Algorithm 1 produces an $\varepsilon$-NS joint strategy with probability at least $1-\delta$ (i.e., $\gap(\varphi^{\out}) \leq \varepsilon$). This is equivalent to requiring that:
    \begin{equation}
        \begin{aligned}
            &\varepsilon \sqrt{M}\geq 8kn(n+1) c_{\size}^{\bm{\varphi}^\star} \log(4(n+1)k/\delta) \sqrt{{2(n-1)}} \left[\frac{n-1}{2} +\sqrt{\frac{n}{2}}\right] +\epsilon_{\opt}\sqrt{M} \\
            &\Leftrightarrow  (\varepsilon -\epsilon_{\opt})\sqrt{M}\geq 8kn(n+1) c_{\size}^{\bm{\varphi}^\star} \log(4(n+1)k/\delta) \sqrt{{2(n-1)}} \left[\frac{n-1}{2} +\sqrt{\frac{n}{2}}\right]\\
            &\Leftrightarrow \sqrt{M} \geq \frac{8kn(n+1)\sqrt{{2(n-1)}} c_{\size}^{\bm{\varphi}^\star}}{\varepsilon -\epsilon_{\opt}} \log(4(n+1)k/\delta)  \left[\frac{n-1}{2} +\sqrt{\frac{n}{2}}\right] \\
            &\Leftrightarrow M \geq \frac{128k^2n^2(n+1)^2(n-1) (c_{\size}^{\bm{\varphi}^\star})^2}{(\varepsilon -\epsilon_{\opt})^2} \log^2(4(n+1)k/\delta)  \left[\frac{n-1}{2} +\sqrt{\frac{n}{2}}\right]^2
        \end{aligned} 
    \end{equation}
    as desired. Since $\epsilon_{\opt} = o(\varepsilon)$, we have {\normalfont \textbf{optimal}} dependence on $\varepsilon$ (up to logarithmic factors) due to \cite{hassani2020stochastic,bai20provable}. 
\end{proof}

\subsection{Learning Approximate Nash-Stable Pure Strategies}
\label{supp:Learning Approximate Nash-Stable Pure Strategies - semi-bandit}
\setcounter{algorithm}{1}
\begin{algorithm}[t!]
    \caption{\textbf{Surrogate Minimization in POCF Games (Pure Strategies)}}
    \label{alg:surrogate min pure}   
    \textbf{Input:} An offline dataset $\mathcal{S}=\{(\mathbf{a}^m, \mathbf{v}^m)\}_{m=1}^M$.
    \begin{algorithmic}[1] 
        \State{Construct $\hat{v}_i$ and $b_i^\delta$ for any agent $i$ based on $\mathcal{S}$.}
        \State{Return an deterministic strategy $\bm{\varphi}^{\out}$ that \textit{approximately} solves:
        \begin{equation}
            \label{supp:eq:approx NS}
            \min_{\bm{\varphi}\in \Gamma}\text{ }\max_{i \in \mathcal{N}} \left[\overline{V}_i^{\star,\delta}(\bm{\varphi}_{-i}) - \underline{V}_{i}^\delta(\bm{\varphi})\right]
        \end{equation}
        where $\overline{V}_{i}(\bm{\varphi}), \underline{V}_{i}^\delta(\bm{\varphi}),\overline{V}_i^{\star,\delta}(\bm{\varphi}_{-i})$ are as in equations (4)-(5) in the main text.}
    \end{algorithmic}
\end{algorithm}

In this section, we concentrate on learning an approximate NS \textit{pure} strategy under semi-bandit feedback. As mentioned in Section 4.2.2 in the main text, since we focus on \textit{symmetric} POCF games, they always admit a pure NS strategy by Lemma \ref{supp:lemma:potential game}, which legitimates learning an approximate NS \textit{pure} strategy. Hence, we can adapt Algorithm 1 to produce a pure strategy. Indeed, let $\Gamma_i \subset \Delta(\mathcal{A}_i)$ be the set of agent $i$'s \textit{deterministic} strategies, i.e., each deterministic strategy $\varphi_i \in \Gamma_i$ corresponds to exactly one joint \textit{pure} strategy $a_i \in \mathcal{A}_i$, such that $\varphi_i(a_i)=1$ for any agent $i$ and $\varphi_i(a_i')=0$ for any other pure strategy $a_i \neq a_i' \in \mathcal{A}_i$. Thus, instead of approximately solving \eqref{supp:eq:approx NS} over \textit{mixed} strategies, Algorithm 1 in the main text can be modified to find a joint \textit{deterministic} strategy $\bm{\varphi}^{\out}\in \Gamma := \prod_{i=1}^n \Gamma_i$ that approximately solves the minimization problem $\min_{\bm{\varphi}\in\Gamma} \max_{i\in \mathcal{N}} [\overline{V}_i^{\star,\delta}(\bm{\varphi}_{-i})-\underline{V}_{i}^\delta(\bm{\varphi})]$. Algorithm \ref{alg:surrogate min pure} summarizes the resulting algorithm.

Now, similarly to (7) in the full paper, we can measure how well the dataset covers all coalition sizes through \textit{deterministic} unilateral deviations from a joint deterministic strategy $\bm{\varphi}$ via:
\begin{equation*}
    \label{eq:coalition coefficient pure}
    c_{\size}^{\text{pure}}(\bm{\varphi}) = \max_{i\in \mathcal{N}, \ell \in [k], \phi_i\in\Gamma_i, \alpha\in [n]\cup\{0\}:d_\ell^{\rho}(\alpha) > 0} \frac{d_\ell^{\bm{\varphi}_{-i},\phi_i}(\alpha)}{d_\ell^{\rho}(\alpha)}
\end{equation*}

Next, we establish Algorithm \ref{alg:surrogate min pure}’s approximation to Nash stability under \textit{pure} strategies, which replaces the factor $n+1$ in equation (10) of the main text with the constant $3$:

\begin{theorem}
    \label{supp:thm:semi-bandit pure}
    Under semi-bandit feedback, Assumption 1 and pure strategies, for any $\delta \in (0,1]$, any dataset size $M \in \mathbb{N}$ and any NS joint strategy $\bm{\varphi}^\star$, the joint pure strategy $\varphi^{\out}$ produced by Algorithm \ref{alg:surrogate min pure} with utility estimators and exploration bonuses defined in equations (8) and (9) within the full paper satisfies the following with probability at least $1-\delta$: 
    \begin{equation}
        \gap(\varphi^{\out}) \leq 24kn c_{\size}^{\text{pure}}(\bm{\varphi}^\star) \log(4(n+1)k/\delta) \sqrt{\frac{2(n-1)}{M}} \left[\frac{n-1}{2} +\sqrt{\frac{n}{2}}\right] +\epsilon_{\opt}
    \end{equation}
    where $c_{\size}^{\text{pure}}(\bm{\varphi}^\star)$ is defined in equation (7) within the full paper. 
\end{theorem}
\begin{proof}
    By Theorem \ref{supp:thm:general duality gap}, we need to bound $\mathbb{E}_{\mathbf{a}\sim (\bm{\varphi}_{-i}^\star,\varphi_i')} [b_i^\delta(\mathbf{a})]$ and $ \mathbb{E}_{\mathbf{a}\sim \bm{\varphi}^\star} [b_i^\delta(\mathbf{a})]$ for some NS pure joint strategy $\bm{\varphi}^\star$ and any strategy $\varphi_i \in \Gamma_i$. Particularly, there exists a joint action $\mathbf{a}^\star \in \mathcal{A}$ such that, for any agent $i$, $\varphi_i^\star(a_i^\star)=1$ and $\varphi_i^\star(a_i)=0$ for any other action $a_i^\star \neq a_i \in \mathcal{A}$. As such, for any $\delta \in (0,1]$, the following holds with probability at least $1-\delta$ simultaneously for each agent $i$:

    \begin{subequations}
        \label{eq:bound bonus pure}
        \begin{align}
            \mathbb{E}_{\mathbf{a}\sim (\bm{\varphi}_{-i}^\star,\varphi_i')} [b_i^\delta(\mathbf{a})] &= \mathbb{E}_{\mathbf{a}\sim (\bm{\varphi}_{-i}^\star,\varphi_i')}\left[ \sum_{\ell\in a_i} \sum_{i \neq j \in C_\ell^{\mathbf{a}}: N_{i,j}^\ell \geq 1} \sqrt{\frac{2  \log(4(n+1)k/\delta)}{N_{i,j}^\ell\vee 1}}\right] \\
            &\leq  \mathbb{E}_{\mathbf{a}\sim (\bm{\varphi}_{-i}^\star,\varphi_i')}\left[\sum_{\ell\in [k]}\sum_{i \neq j \in \mathcal{N}: N_{i,j}^\ell \geq 1}\sqrt{\frac{2  \log(4(n+1)k/\delta)}{N_{i,j}^\ell\vee 1}}\right] \\
            &= \sum_{\ell\in [k]} \sum_{i \neq j \in \mathcal{N}: N_{i,j}^\ell \geq 1} \mathbb{E}_{\mathbf{a}\sim (\bm{\varphi}_{-i}^\star,\varphi_i')}\left[\sqrt{\frac{2  \log(4(n+1)k/\delta)}{N_{i,j}^\ell\vee 1}}\right] \label{eq:at most k pure} \\
            &\leq \sum_{\ell\in [k]} \sum_{i \neq j \in \mathcal{N}: N_{i,j}^\ell \geq 1} \sum_{\alpha=|C_\ell^{\mathbf{a}}|-1}^{|C_\ell^{\mathbf{a}}|+1} d^{(\bm{\varphi}_{-i}^\star,\varphi_i')}_\ell(\alpha) \sqrt{\frac{8  \log^2(4(n+1)k/\delta)}{M d_\ell^{\rho}(i,j)}}  \label{eq:apply lemma pure} \\
            &\leq c_{\size}^{\text{pure}}(\bm{\varphi}^\star) \sum_{\ell\in [k]}\sum_{i \neq j \in \mathcal{N}: N_{i,j}^\ell \geq 1}\sum_{\alpha=|C_\ell^{\mathbf{a}}|-1}^{|C_\ell^{\mathbf{a}}|+1} d_\ell^{\rho}(\alpha) \sqrt{\frac{8  \log^2(4(n+1)k/\delta)}{M d_\ell^{\rho}(i,j)}} \label{eq:assumption 1 pure}
        \end{align}
    \end{subequations}
    where the inequality in \eqref{eq:apply lemma pure} is due to $d^{(\bm{\varphi}_{-i}^\star,\varphi_i')}_\ell(\alpha) = \sum_{\mathbf{a}\in\mathcal{A}: |C_\ell^{\mathbf{a}}|=\alpha} (\bm{\varphi}_{-i}^\star,\varphi_i')(\mathbf{a})$ and Lemma \ref{supp:lemma:counter}. Further, the inequality in \eqref{eq:assumption 1 pure} is derived from Assumption 1 and the definition of $c_{\size}^{\text{pure}}(\bm{\varphi}^\star)$. Moreover, unlike \eqref{eq:bound bonus pure} from the proof of Theorem \ref{supp:thm:semi-bandit} for mixed strategies, note that the summation over $\alpha$ is from $|C_\ell^{\mathbf{a}}|-1$ to $|C_\ell^{\mathbf{a}}|+1$. The key intuition is as follows. For any joint action $\mathbf{a}$ sampled from some joint strategy $\bm{\varphi}$ and coalition $C\in\pi^{\mathbf{a}}$, consider a unilateral deviation of some agent $i$. If $i \in C$, then agent $i$ may either remain in or leave $C$ after she deviates. Otherwise, if $i \notin C$, then agent $i$ may either join $C$ or stay out after deviating. Thus, each unilateral deviation affects coalition sizes by either increasing or decreasing them by $1$, or leaving them unchanged.

    Combining \eqref{eq:bound bonus pure} with \eqref{eq:bound rho} from the proof of Theorem \ref{supp:thm:semi-bandit}, we obtain:
    \begin{subequations}
        \label{eq:bound bonus cont pure}
        \begin{align}
            \mathbb{E}_{\mathbf{a}\sim (\bm{\varphi}_{-i}^\star,\varphi_i')} [b_i^\delta(\mathbf{a})] &\leq3 c_{\size}^{\text{pure}}(\bm{\varphi}^\star) \sum_{\ell\in [k]} \sum_{i \neq j \in \mathcal{N}: N_{i,j}^\ell \geq 1} \sqrt{\frac{8  \log^2(4(n+1)k/\delta)}{M d_\ell^{\rho}(i,j)}}  \cdot\left[\sum_{\alpha=|C_\ell^{\mathbf{a}}|-1}^{|C_\ell^{\mathbf{a}}|+1} \sum_{i',j' \in \mathcal{N}: i' \neq j'}d_\ell^{\rho}(i',j')\right] \\
            &=3c_{\size}^{\text{pure}}(\bm{\varphi}^\star) \sum_{\ell\in [k]} \sum_{i \neq j \in \mathcal{N}: N_{i,j}^\ell \geq 1} \sqrt{\frac{8  \log^2(4(n+1)k/\delta)}{M d_\ell^{\rho}(i,j)}}  \cdot \left[\sum_{i',j' \in \mathcal{N}: i' \neq j' \wedge d_\ell^{\rho}(i',j') \leq d_\ell^{\rho}(i,j)} d_\ell^{\rho}(i',j')\right] \\
            &+ 3c_{\size}^{\text{pure}}(\bm{\varphi}^\star) \sum_{\ell\in [k]} \sum_{i \neq j \in \mathcal{N}: N_{i,j}^\ell \geq 1} \left[\sum_{i',j' \in \mathcal{N}: i' \neq j' \wedge d_\ell^{\rho}(i',j') > d_\ell^{\rho}(i,j)} \sqrt{\frac{8  \log^2(4(n+1)k/\delta)}{M d_\ell^{\rho}(i,j)}} \cdot d_\ell^{\rho}(i',j')\right] \\
            &\leq 3c_{\size}^{\text{pure}}(\bm{\varphi}^\star) \sum_{\ell\in [k]} \sum_{i \neq j \in \mathcal{N}: N_{i,j}^\ell \geq 1} \sqrt{\frac{8  \log^2(4(n+1)k/\delta)}{M d_\ell^{\rho}(i,j)}}  \cdot \left[\sum_{i',j' \in \mathcal{N}: i' \neq j' \wedge d_\ell^{\rho}(i',j') \leq d_\ell^{\rho}(i,j)} d_\ell^{\rho}(i,j)\right] \\
            &+ 3c_{\size}^{\text{pure}}(\bm{\varphi}^\star) \sum_{\ell\in [k]} \sum_{i \neq j \in \mathcal{N}: N_{i,j}^\ell \geq 1} \left[\sum_{i',j' \in \mathcal{N}: i' \neq j' \wedge d_\ell^{\rho}(i',j') > d_\ell^{\rho}(i,j)} \sqrt{\frac{8  \log^2(4(n+1)k/\delta)}{M d_\ell^{\rho}(i',j')}} \cdot d_\ell^{\rho}(i',j')\right]  \\
            &=3c_{\size}^{\text{pure}}(\bm{\varphi}^\star) \sum_{\ell\in [k]} \sum_{i \neq j \in \mathcal{N}: N_{i,j}^\ell \geq 1} \sum_{i',j' \in \mathcal{N}: i' \neq j' \wedge d_\ell^{\rho}(i',j') \leq d_\ell^{\rho}(i,j)} \sqrt{\frac{8  \log^2(4(n+1)k/\delta)d_\ell^{\rho}(i,j)}{M}} \\
            &+ 3c_{\size}^{\text{pure}}(\bm{\varphi}^\star) \sum_{\ell\in [k]} \sum_{i \neq j \in \mathcal{N}: N_{i,j}^\ell \geq 1} \sum_{i',j' \in \mathcal{N}: i' \neq j' \wedge d_\ell^{\rho}(i',j') > d_\ell^{\rho}(i,j)} \sqrt{\frac{8  \log^2(4(n+1)k/\delta) d_\ell^{\rho}(i',j')}{M}} \\
            &\leq 3c_{\size}^{\text{pure}}(\bm{\varphi}^\star) \sum_{\ell\in [k]} \sum_{i \neq j \in \mathcal{N}} \sum_{i',j' \in \mathcal{N}: i' \neq j'} \sqrt{\frac{8  \log^2(4(n+1)k/\delta)d_\ell^{\rho}(i,j)}{M}} \\
            &+ 3c_{\size}^{\text{pure}}(\bm{\varphi}^\star) \sum_{\ell\in [k]} \sum_{i \neq j \in \mathcal{N}} \sum_{i',j' \in \mathcal{N}: i' \neq j'} \sqrt{\frac{8  \log^2(4(n+1)k/\delta) d_\ell^{\rho}(i',j')}{M}} \\
            &= \binom{n}{2} 3 c_{\size}^{\text{pure}}(\bm{\varphi}^\star) \sum_{\ell\in[k]} \sum_{i \neq j \in \mathcal{N}} \sqrt{\frac{8  \log^2(4(n+1)k/\delta)d_\ell^{\rho}(i,j)}{M}}\\
            &+ 3nc_{\size}^{\text{pure}}(\bm{\varphi}^\star) \sum_{\ell\in [k]} \sum_{i',j' \in \mathcal{N}: i' \neq j'} \sqrt{\frac{8  \log^2(4(n+1)k/\delta)d_\ell^{\rho}(i',j')}{M}} \\
            &\leq \frac{n(n-1)}{2} 3 c_{\size}^{\text{pure}}(\bm{\varphi}^\star) \sum_{\ell\in[k]} \sqrt{\frac{8 (n-1) \log^2(4(n+1)k/\delta)}{M} \sum_{i \neq j \in \mathcal{N}} d_\ell^{\rho}(i,j)} \label{eq:cauchy1 pure}\\
            &+ 3nc_{\size}^{\text{pure}}(\bm{\varphi}^\star) \sum_{\ell\in [k]} \sqrt{\frac{8  \log^2(4(n+1)k/\delta) \binom{n}{2}}{M}\sum_{i',j' \in \mathcal{N}: i' \neq j'} d_\ell^{\rho}(i',j')}\label{eq:cauchy pure}\\
            &\leq 6kn c_{\size}^{\text{pure}}(\bm{\varphi}^\star) \log(4(n+1)k/\delta) \sqrt{\frac{2(n-1)}{M}} \left[\frac{n-1}{2} +\sqrt{\frac{n}{2}}\right]\label{eq:1 pure}
        \end{align}
    \end{subequations}
    where the inequalities in \eqref{eq:cauchy1 pure} and \eqref{eq:cauchy pure} follow from the Cauchy-Schwarz inequality, according to which $\sum_{\ell=1}^n \sqrt{a_\ell} \leq \sqrt{n \sum_{\ell=1}^n a_\ell}$ for any sequence $a_1, \dots, a_n$. The last inequality in \eqref{eq:1 pure} is due to:
    \begin{equation}
        \begin{aligned}
            \sum_{i \neq j \in \mathcal{N}} d_\ell^{\rho}(i,j) &\leq \sum_{i',j' \in \mathcal{N}: i'\neq j'} \text{ } \sum_{\mathbf{a}\in\mathcal{A}: \exists C \in \pi^{\mathbf{a}} \text{ s.t. } i',j'\in C} \rho(\mathbf{a}) = \sum_{\mathbf{a}\in\mathcal{A}} \rho(\mathbf{a})=1 \\
            \sum_{i',j' \in \mathcal{N}: i' \neq j'} d_\ell^{\rho}(i',j') &\leq \sum_{i',j' \in \mathcal{N}: i' \neq j'} \text{ } \sum_{\mathbf{a}\in\mathcal{A}: \exists C \in \pi^{\mathbf{a}} \text{ s.t. } i',j'\in C} \rho(\mathbf{a}) = \sum_{\mathbf{a}\in\mathcal{A}} \rho(\mathbf{a}) =1
        \end{aligned}
    \end{equation}
    By arguments similar to the above:
    \begin{equation}
        \mathbb{E}_{\mathbf{a}\sim \bm{\varphi}^\star} [b_i^\delta(\mathbf{a})]\leq 6kn c_{\size}^{\text{pure}}(\bm{\varphi}^\star) \log(4(n+1)k/\delta) \sqrt{\frac{2(n-1)}{M}} \left[\frac{n-1}{2} +\sqrt{\frac{n}{2}}\right]
    \end{equation}
    Since our constructed $b_i^\delta$ is an exploration bonus for the utility estimator $\hat{v}_i$ by Lemma \ref{supp:lemma:bound on diffs semi bandit}, applying Theorem \ref{supp:thm:general duality gap} implies that:
    \begin{equation}
    \label{eq:gap out upper bound pure}
        \gap(\varphi^{\out}) \leq 24kn c_{\size}^{\text{pure}}(\bm{\varphi}^\star) \log(4(n+1)k/\delta) \sqrt{\frac{2(n-1)}{M}} \left[\frac{n-1}{2} +\sqrt{\frac{n}{2}}\right] +\epsilon_{\opt}
    \end{equation}
    as desired.
\end{proof}

\subsection{Sample Complexity of Algorithm \ref{alg:surrogate min pure}}

\begin{corollary}
    \label{supp:coro:optimal pure}
    Under {\normalfont semi-bandit} feedback and Assumption 1, for any $\delta \in (0,1]$, any $\varepsilon > \epsilon_{\opt}$ with $\epsilon_{\opt} = o(\varepsilon)$ and any NS joint {\normalfont pure} strategy $\bm{\varphi}^\star$, Algorithm 1 with utility estimators and exploration bonuses defined in equations (8) and (9) within the full paper has a sample complexity bound with {\normalfont \textbf{optimal}} dependence on $\varepsilon$ (up to logarithmic factors): if the dataset size is at least
    \begin{equation}
        M \geq \frac{576k^2n^2(n-1) (c_{\size}^{\text{pure}}(\bm{\varphi}^\star))^2}{(\varepsilon -\epsilon_{\opt})^2} \log^2(4(n+1)k/\delta)  \left[\frac{n-1}{2} +\sqrt{\frac{n}{2}}\right]^2,
    \end{equation}
    then the joint strategy $\varphi^{\out}$ is $\varepsilon$-NS (i.e., $\gap(\varphi^{\out}) \leq \varepsilon$) with probability at least $1-\delta$.
\end{corollary}
\begin{proof}
    The proof is a continuation of the proof for Theorem \ref{supp:thm:semi-bandit pure}. Indeed, if the right-hand side in \eqref{eq:gap out upper bound pure} is at most $\varepsilon$ for some $\varepsilon \geq \epsilon_{\opt}$ with $\epsilon_{\opt} = o(\varepsilon)$, then Algorithm 1 produces an $\varepsilon$-NS joint strategy with probability at least $1-\delta$ (i.e., $\gap(\varphi^{\out}) \leq \varepsilon$). This is equivalent to requiring that:
    \begin{equation}
        \begin{aligned}
            &\varepsilon \sqrt{M}\geq 24kn c_{\size}^{\text{pure}}(\bm{\varphi}^\star) \log(4(n+1)k/\delta) \sqrt{{2(n-1)}} \left[\frac{n-1}{2} +\sqrt{\frac{n}{2}}\right] +\epsilon_{\opt}\sqrt{M} \\
            &\Leftrightarrow  (\varepsilon -\epsilon_{\opt})\sqrt{M}\geq 24kn c_{\size}^{\text{pure}}(\bm{\varphi}^\star) \log(4(n+1)k/\delta) \sqrt{{2(n-1)}} \left[\frac{n-1}{2} +\sqrt{\frac{n}{2}}\right]\\
            &\Leftrightarrow \sqrt{M} \geq \frac{24kn\sqrt{{2(n-1)}} c_{\size}^{\text{pure}}(\bm{\varphi}^\star)}{\varepsilon -\epsilon_{\opt}} \log(4(n+1)k/\delta)  \left[\frac{n-1}{2} +\sqrt{\frac{n}{2}}\right] \\
            &\Leftrightarrow M \geq \frac{576k^2n^2(n-1) (c_{\size}^{\text{pure}}(\bm{\varphi}^\star))^2}{(\varepsilon -\epsilon_{\opt})^2} \log^2(4(n+1)k/\delta)  \left[\frac{n-1}{2} +\sqrt{\frac{n}{2}}\right]^2
        \end{aligned} 
    \end{equation}
    as desired. Since $\epsilon_{\opt} = o(\varepsilon)$, we have {\normalfont \textbf{optimal}} dependence on $\varepsilon$ (up to logarithmic factors) due to \cite{hassani2020stochastic,bai20provable}. 
\end{proof}

\section{Proof of Theorem 4}

\begin{customthm}{4}
    \label{supp:thm:1/8 gap}
    Let $\mathcal{G}'$ be the class of all pairs $(G,\rho)$ consisting of a POCF game $G$ and an exploration policy $\rho$ satisfying Assumption 1. Then, for any algorithm $\alg$ with {\normalfont bandit} feedback, there is $(G,\rho)\in \mathcal{G}$ such that any joint strategy $\bm{\varphi}$ produced by $\alg$ satisfies $\gap(\bm{\varphi})\geq\frac{1}{20}$ for the POCF game $G$, regardless of the dataset size.
\end{customthm}
\begin{proof}
    We construct two symmetric POCF games $G_1$ and $G_2$, each with $3$ agents whose utility distributions are deterministic, where the number of coalitions is at most $3$ (i.e., $k=3$), the action space of each agent is $\mathcal{A}_1=\mathcal{A}_2=\mathcal{A}_3=\{\{1\}, \{2\}, \{3\}, \{1,2\}\}$, and both games share the same exploration policy $\rho$. Formally, since any partition can contain at most three coalitions, any joint action $\mathbf{a}$ induces a partition of the agents $\pi^{\mathbf{a}} = (C_1^{\mathbf{a}},C_2^{\mathbf{a}}, C_3^{\mathbf{a}})$, where $C_1^{\mathbf{a}}$, $C_2^{\mathbf{a}}$ and $C_3^{\mathbf{a}}$ are the sets of agents that joined the first, second and third candidate coalition, respectively. If no agent joins the $\ell$-th coalition for $\ell \in \{1,2,3\}$ (i.e., $\ell \notin a_i$ for any agent $i$), then that coalition is not formed. Therefore, we construct the two games and their common exploration policy as follows:

    \paragraph{Construction of the first game $G_1$.} In the first game, denoted as $G_1$, all pure NS strategies are strategies where either: (1) two agents choose to join the first and second candidate coalitions while the third one decides to join only the first candidate coalition, or (2) two agents join the second candidate coalition and the third agent only joins the third candidate coalition. Specifically, for each pair of distinct agents $i,j$ and any joint action $\mathbf{a}$, the mutual utility of agents $i,j$ within the first, second and third candidate coalition are determined deterministically via $\mathcal{D}_{i,j}^{1}(\mathbf{a})$, $\mathcal{D}_{i,j}^{2}(\mathbf{a})$ and $\mathcal{D}_{i,j}^{3}(\mathbf{a})$ as follows (respectively):
    \begin{equation}
    \label{eq:deterministic utilities semi bandit}
        \begin{aligned}
            \mathcal{D}_{i,j}^1(\mathbf{a}) =\begin{cases}
            1 &, \text{ if } i,j \in C_1^{\mathbf{a}} \text{ and } |C_1^{\mathbf{a}}|=2 \\
            \frac{1}{2} &, \text{ if } i,j \in C_1^{\mathbf{a}} \text{ and } |C_1^{\mathbf{a}}|=3 
        \end{cases}
        \\        
        \mathcal{D}_{i,j}^2(\mathbf{a}) =\begin{cases}           
            1 &, \text{ if } i,j \in C_2^{\mathbf{a}} \text{ and } |C_2^{\mathbf{a}}|=2 \\    
            -1 &, \text{ if } i,j \in C_2^{\mathbf{a}} \text{ and } |C_2^{\mathbf{a}}|=3 
        \end{cases} \\
        \mathcal{D}_{i,j}^3(\mathbf{a}) =\begin{cases}           
            -\frac{1}{2} &, \text{ if } i,j \in C_2^{\mathbf{a}} \text{ and } |C_2^{\mathbf{a}}|=2 \\    
            -\frac{1}{4} &, \text{ if } i,j \in C_2^{\mathbf{a}} \text{ and } |C_2^{\mathbf{a}}|=3 
        \end{cases}
        \end{aligned}
    \end{equation}
    For instance, if agents $1$ and $2$ join the first and second candidate coalitions while agent $3$ only joins the first one, then each pair of agents receives a utility of $\frac{1}{2}$ from each other within the first candidate coalition, while agents $1$ and $2$ obtain a utility of $1$ from each other within the second candidate coalition.
    
    Next, we formally characterize all the pure NS strategies of the first game $G_1$:
    \begin{lemma}
        \label{supp:lemma:NS game 1 bandit}
        The pure NS strategies of the first game $G_1$ consist only of all pure joint strategies where either: 
        \begin{enumerate}
            \item two agents choose to join the first and second candidate coalitions, while the third one decides to join only the first candidate coalition, or 
            \item two agents join the second candidate coalition and the third agent only joins the third candidate coalition.
        \end{enumerate}       
    \end{lemma}
    \begin{proof}
         We first prove that all pure joint strategies that are as in the above are Nash stable, and then show that no other pure joint strategy is Nash-stable. Formally:
        \begin{enumerate}
            \item \textbf{{All pure strategy where two agents join the first and second candidate coalitions, while the third joins only the first candidate coalition, are Nash-stable:}} Without loss of generality, consider the joint action $\mathbf{a}$ where $a_1=\{1,2\}, a_2=\{1,2\}, a_3=\{1\}$, i.e., $C_1^{\mathbf{a}}=\{1,2,3\}$, $C_2^{\mathbf{a}}=\{1,2\}$ and $C_3^{\mathbf{a}}=\emptyset$ (other cases are symmetric by \eqref{eq:deterministic utilities semi bandit}). Now, we will prove that no agent has an incentive to unilaterally deviate from $\mathbf{a}$:
            \begin{enumerate}
                \item \textbf{Agents $1$ and $2$ have no incentive to unilaterally deviate:} We supply the proof for agent $1$, while the proof for agent $2$ is identical due to \eqref{eq:deterministic utilities semi bandit}. Specifically:
                \begin{enumerate}
                    \item If agent $1$ unilaterally deviates to $a_1'=\{1\}$, then the resulting coalitions are $C_1^{(a_1',a_2,a_3)}=\{1,2,3\}$, $C_2^{(a_1',a_2,a_3)}=\{1\}$ and $C_3^{(a_1',a_2,a_3)}=\emptyset$. Thus, agent $1$'s utility from the resulting partition is $1+1=2$ as she receives a utility of $1$ from both agents $2,3$ within the first candidate coalition. However, before deviating, agent $1$'s utility was $1+1+1=3$, as she obtained a utility of $1$ from agent $2$ within the second candidate coalition. Namely, agent $1$'s utility decreased after deviating, meaning that she has no incentive to unilaterally deviate.

                    \item If agent $1$ unilaterally deviates to $a_1'=\{2\}$, then the resulting coalitions are $C_1^{(a_1',a_2,a_3)}=\{2,3\}$, $C_1^{(a_1',a_2,a_3)}=\{1,2\}$ and $C_3^{(a_1',a_2,a_3)}=\emptyset$. Thus, agent $1$'s utility from the resulting partition is $1$ as she receives a utility of $1$ from agent $2$ within the second candidate coalition. However, before deviating, agent $1$'s utility was $3$, as discussed in the previous case. Namely, agent $1$'s utility decreased after deviating, meaning that she has no incentive to unilaterally deviate.

                    \item If agent $1$ unilaterally deviates to $a_1'=\{3\}$, then the resulting coalitions are $C_1^{(a_1',a_2,a_3)}=\{2,3\}$,  $C_1^{(a_1',a_2,a_3)}=\{2\}$ and $C_3^{(a_1',a_2,a_3)}=\{1\}$. Thus, agent $1$'s utility from the resulting partition is $0$ as she is alone in the third candidate coalition. However, before deviating, agent $1$'s utility was $3$, as discussed earlier. Namely, agent $1$'s utility decreased after deviating, meaning that she has no incentive to unilaterally deviate.
                \end{enumerate}

                \item \textbf{Agent $3$ has no incentive to unilaterally deviate:} Before deviating, agent $3$'s utility is $1+1=2$ since she obtains a utility of $1$ from both agents $2,3$ within the first candidate coalition. However, agent $3$ has not incentive to deviate. 
                \begin{enumerate}
                    \item By deviating to $a_3'=\{3\}$, agent $3$'s utility decreases to $0$, as she is alone in the third candidate coalition.
                    
                    \item If agent $3$ deviates to $a_3'=\{1,2\}$, then the resulting coalitions are $C_1^{(a_1,a_2,a_3')}=\{1,2,3\}$, $C_1^{(a_1,a_2,a_3')}=\{1,2,3\}$ and $C_3^{(a_1,a_2,a_3')}=\emptyset$. Hence, agent $3$'s utility from the resulting partition is $2(\frac{1}{2}-1)=-1$, as she receives a utility of $\frac{1}{2}$ from both agents $1,2$ within the first candidate coalition and a utility of $-1$ from both agents $1,2$ within the second candidate coalition. That is, agent $3$'s utility decrease after unilaterally deviating.

                    \item If agent $3$ deviates to $a_3'=\{2\}$, then the resulting coalitions are $C_1^{(a_1,a_2,a_3')}=\{1,2\}$, $C_1^{(a_1,a_2,a_3')}=\{1,2,3\}$ and $C_3^{(a_1,a_2,a_3')}=\emptyset$. Hence, agent $3$'s utility from the resulting partition is $2(-1)=-2$, as she receives a utility of a utility of $-1$ from both agents $1,2$ within the second candidate coalition. That is, agent $3$'s utility decrease after unilaterally deviating.
                \end{enumerate}                
            \end{enumerate}

            \item \textbf{All pure strategies where two agents join the second candidate coalition and the third agent only joins the third candidate coalition are Nash stable:}  Without loss of generality, assume that agent $1$ only joins the third candidate coalition, i.e., $a_1=\{3\}$ and $a_2=a_3=\{2\}$, meaning that $C_1^{\mathbf{a}}=\emptyset$, $C_2^{\mathbf{a}}=\{2,3\}$ and $C_3^{\mathbf{a}}=\{1\}$ (other cases are symmetric by \eqref{eq:deterministic utilities semi bandit}). Thus, agent $1$ obtains zero utility from being alone. However, she has no incentive to deviate. Indeed, if agent $1$ unilaterally deviates by joining the first candidate coalition, then her utility remains unchanged. If agent $1$ unilaterally deviates by joining the second candidate coalition, then her utility decreases to $-2$ since she obtains a utility of $-1$ from each other agent.

            \item \textbf{No other pure joint strategy is Nash-stable:} Consider any joint action $\mathbf{a}$ where at most one agent joins both coalitions. We distinguish between the following cases:
            \begin{enumerate}
                \item \textbf{One agent joins the first and second candidate coalitions, while the others join exactly one coalition:} Without loss of generality, assume that $a_1=\{1,2\}$ (other cases are symmetric by \eqref{eq:deterministic utilities semi bandit}). The following holds:
                \begin{enumerate}
                    \item If $a_2=\{1\}, a_3=\{1\}$, then $C_1^{\mathbf{a}}=\{1,2,3\}$, $C_2^{\mathbf{a}}=\{1\}$ and $C_3^{\mathbf{a}}=\emptyset$. By arguments similar to the above, both agents $2,3$ have an incentive to unilaterally deviate by joining the first and second candidate coalitions (i.e., picking the action $\{1,2\}$). For instance, such unilateral deviation would increase agent $2$'s utility from $2$ to $3$.

                    \item If $a_2=\{2\}, a_3=\{2\}$, then $C_1^{\mathbf{a}}=\{1\}$, $C_2^{\mathbf{a}}=\{1,2,3\}$ and $C_3^{\mathbf{a}}=\emptyset$. By arguments similar to the above, all agents $1,2,3$ have an incentive to unilaterally deviate by joining the first candidate coalitions. For instance, if agent $2$ unilaterally deviates to $a_2'=\{1\}$, then the resulting coalitions are $C_1^{(a_1,a_2',a_3)}=\{1,2\}$, $C_1^{(a_1,a_2',a_3)}=\{1,3\}$ and $C_3^{(a_1,a_2',a_3)}=\emptyset$. This increases agent $2$'s utility from $-2$ to $1$ by \eqref{eq:deterministic utilities semi bandit}.

                    \item If either $a_2=\{1\}, a_3=\{2\}$ or $a_2=\{2\}, a_3=\{1\}$, then the first and second candidate coalitions are of size $2$. Without loss of generality, assume that $a_2=\{1\}, a_3=\{2\}$, i.e., $C_1^{\mathbf{a}}=\{1,2\}$, $C_2^{\mathbf{a}}=\{1,3\}$  and $C_3^{\mathbf{a}}=\emptyset$ (the other case is symmetric by \eqref{eq:deterministic utilities semi bandit}). Thus, both agents $1,2$ receive a utility of $1$ from each other by \eqref{eq:deterministic utilities semi bandit}. As we have proven earlier, agent $2$ has an incentive to unilaterally deviate by joining the first and second candidate coalitions (i.e., unilaterally deviating to the action $a_2'=\{1,2\}$ instead).

                    \item If either agent $2$, agent $3$ or both only joins the third candidate coalition, then she clearly has an incentive to deviate.
                \end{enumerate}
                \item \textbf{Each agent joins exactly one coalition:} We distinguish between the following cases:
                \begin{enumerate}
                    \item \textbf{All agents only join the first candidate coalition:} If $a_1=a_2=a_3=\{1\}$, then $C_1^{\mathbf{a}}=\{1,2,3\}$ and $C_2^{\mathbf{a}}=C_3^{\mathbf{a}}=\emptyset$. Each agent obtains a utility of $2\cdot\frac{1}{2}=1$, as she receives a utility of $\frac{1}{2}$ from every other agent within the first candidate coalition. However, each agent has an incentive to unilaterally deviate by joining the first and second candidate coalitions since we have already proven that the resulting partition is Nash stable.

                    \item \textbf{All agents only join the second candidate coalition:} If $a_1=a_2=a_3=\{2\}$, then $C_1^{\mathbf{a}}=C_3^{\mathbf{a}}=\emptyset$ and $C_2^{\mathbf{a}}=\{1,2,3\}$. Each agent obtains a utility of $2\cdot(-1)=-2$, as she receives a utility of $-1$ from every other agent within the second candidate coalition. However, each agent has an incentive to unilaterally deviate by, e.g., not joining any coalition, thus increasing her utility to $0$.

                    \item \textbf{Two agents join the first candidate coalition and the third agent joins the second candidate coalition:} Without loss of generality, assume that $a_1=a_2=\{1\}$ and $a_3=\{2\}$, i.e., $C_1^{\mathbf{a}}=\{1,2\}$, $C_2^{\mathbf{a}}=\{3\}$ and $C_3^{\mathbf{a}}=\emptyset$ (other cases are symmetric by \eqref{eq:deterministic utilities semi bandit}). Agent $3$ obtains a utility of $0$ as she is in a singleton coalition. Thus, she has an incentive to unilaterally deviate by, e.g., joining the first and second candidate coalitions, thereby increasing her utility to $2\cdot\frac{1}{2}=1$, as she receives a utility of $\frac{1}{2}$ from every other agent within the first candidate coalition.

                    \item \textbf{Two agents join the second candidate coalition and the third agent joins the first candidate coalition:} Without loss of generality, assume that $a_1=a_2=\{2\}$ and $a_3=\{1\}$, i.e., $C_1^{\mathbf{a}}=\{3\}$, $C_2^{\mathbf{a}}=\{1,2\}$  and $C_3^{\mathbf{a}}=\emptyset$ (other cases are symmetric by \eqref{eq:deterministic utilities semi bandit}). Agent $1$ obtains a utility of $1$, as she receives a utility of $1$ from agent $2$ within the second candidate coalition. Thus, she has an incentive to unilaterally deviate by, e.g., joining the first and second candidate coalitions, thereby increasing her utility to $2$, as she will then further receive a utility of $1$ from agent $3$ within the first candidate coalition.

                    \item \textbf{One agent only joins the third candidate coalition, while the others join either the first or the second candidate coalition:} Without loss of generality, assume that agent $1$ only joins the third candidate coalition and the others join the first candidate coalition, i.e., $a_1=\{3\}$ and $a_2=a_3=\{1\}$, meaning that $C_1^{\mathbf{a}}=\{2,3\}$, $C_2^{\mathbf{a}}=\emptyset$ and $C_2^{\mathbf{a}}=\{1\}$ (other cases are symmetric by \eqref{eq:deterministic utilities semi bandit}). Thus, agent $1$ has an incentive to unilaterally deviate by also joining the first candidate coalition, thereby increasing her utility from $0$ to $1$.

                    \item \textbf{Two agents only join the third candidate coalition, while the third joins either the first or the second candidate coalition:} In any case, note that each of the two agents that only join the third candidate coalition always has an incentive to join the third agent's coalition, thus increasing her utility from $-\frac{1}{2}$ to $1$.
    
                    \item \textbf{All agents only join the third candidate coalition:} In any case, note that each agent always has an incentive to join either the first or the second candidate coalition, thus increasing her utility from $-\frac{1}{2}$ to $0$ since she will be alone.
                \end{enumerate}

            \end{enumerate}               
        \end{enumerate}
    \end{proof}

    \paragraph{Construction of the second game $G_2$.} In the second game, denoted as $G_2$, all pure NS strategies are strategies where either: (1) both the first and second candidate coalitions are of size exactly $2$, or (2) two agents join either the first or second candidate coalitions, while the third agent only joins the third candidate coalition. Specifically, for each pair of distinct agents $i,j$ and any joint action $\mathbf{a}$, the mutual utility of agents $i,j$ within the first, second and third candidate coalition are determined deterministically via $\tilde{\mathcal{D}}_{i,j}^{1}(\mathbf{a})$, $\tilde{\mathcal{D}}_{i,j}^{2}(\mathbf{a})$ and $\tilde{\mathcal{D}}_{i,j}^{3}(\mathbf{a})$ as follows (respectively):
    \begin{equation}
    \label{eq:deterministic utilities semi bandit 2}
        \begin{aligned}
            \tilde{\mathcal{D}}_{i,j}^1(\mathbf{a}) =\begin{cases}
            1 &, \text{ if } i,j \in C_1^{\mathbf{a}} \text{ and } |C_1^{\mathbf{a}}|=2 \\
            -\frac{1}{4} &, \text{ if } i,j \in C_1^{\mathbf{a}} \text{ and } |C_1^{\mathbf{a}}|=3 
        \end{cases}
        \\        
        \tilde{\mathcal{D}}_{i,j}^2(\mathbf{a}) =\begin{cases}           
            1 &, \text{ if } i,j \in C_2^{\mathbf{a}} \text{ and } |C_2^{\mathbf{a}}|=2 \\    
            -\frac{1}{4} &, \text{ if } i,j \in C_2^{\mathbf{a}} \text{ and } |C_2^{\mathbf{a}}|=3 
        \end{cases}\\
        \tilde{\mathcal{D}}_{i,j}^3(\mathbf{a}) =\begin{cases}           
            -\frac{1}{2} &, \text{ if } i,j \in C_2^{\mathbf{a}} \text{ and } |C_2^{\mathbf{a}}|=2 \\    
            -\frac{1}{4} &, \text{ if } i,j \in C_2^{\mathbf{a}} \text{ and } |C_2^{\mathbf{a}}|=3 
        \end{cases}
        \end{aligned}
    \end{equation}
    For instance, if agents $1$ and $2$ join the first candidate coalition while agent $3$ only joins the second one, then agents $1$ and $2$ both receive a utility of $1$ from each other within the first candidate coalition, while agent $3$ obtains a utility of $0$.
    
    Next, we formally characterize all the pure NS strategies of the second game $G_2$:
    \begin{lemma}
        \label{supp:lemma:NS game 2 bandit}
        The pure NS strategies of the second game $G_2$ consist only of all pure joint strategies where either:
        \begin{enumerate}
            \item both the first and second candidate coalitions are of size exactly $2$, or

            \item  two agents join either the first or second candidate coalitions, while the third agent only joins the third candidate coalition.
        \end{enumerate}
    \end{lemma}
    \begin{proof}
         We first prove that all pure joint strategies that are as in the above are Nash stable, and then show that no other pure joint strategy is Nash-stable. Formally:
        \begin{enumerate}
            \item \textbf{All pure joint strategies where both the first and second candidate coalitions are of size exactly $2$ are Nash stable:} Without loss of generality, consider the joint action $\mathbf{a}$ where $a_1=\{1,2\}$, $a_2=\{1\}$ and $a_3=\{2\}$, i.e., $C_1^{\mathbf{a}}=\{1,2\}$, $C_2^{\mathbf{a}}=\{1,3\}$ and $C_3^{\mathbf{a}}=\emptyset$ (other cases are symmetric by \eqref{eq:deterministic utilities semi bandit 2}). Note that each agent obtains a utility of $1$ since her utility from the other agent in her coalition is $1$. We therefore distinguish between the following cases:
            \begin{enumerate}
                \item If agent $1$ unilaterally deviates to $a_1'=\{1\}$, then the resulting coalitions are $C_1^{(a_1',a_2,a_3)}=\{1,2\}$, $C_2^{(a_1',a_2,a_3)}=\{3\}$ and $C_3^{(a_1',a_2,a_3)}=\emptyset$. Thus, agent $1$'s utility from the resulting partition is $1$ as she receives a utility of $1$ from agent $2$ within the first candidate coalition. However, before deviating, agent $1$'s utility was $1+1=2$, as she obtained a utility of $1$ from agent $2$ within the first candidate coalition and a utility of $1$ from agent $3$ within the second candidate coalition. Namely, agent $1$'s utility decreased after deviating, meaning that she has no incentive to unilaterally deviate to $a_1'=\{1\}$. Agent $1$ has no incentive to unilaterally deviate to $a_1'=\{2\}$ by similar arguments.

                \item If agent $1$ unilaterally deviates to $a_1'=\{3\}$, then the resulting coalitions are $C_1^{(a_1',a_2,a_3)}=\{2\}$, $C_2^{(a_1',a_2,a_3)}=\{3\}$ and $C_3^{(a_1',a_2,a_3)}=\{1\}$. Thus, agent $1$'s utility decreased from $2$ to $0$, meaning that she has no incentive to unilaterally deviate to $a_1'=\{3\}$. 

                \item If agent $2$ unilaterally deviates to $a_2'=\{2\}$, then the resulting coalitions are $C_1^{(a_1,a_2',a_3)}=\{1\}$, $C_1^{(a_1,a_2',a_3)}=\{1,2,3\}$ and $C_3^{(a_1,a_2',a_3)}=\emptyset$. Thus, agent $2$'s utility from the resulting partition is $-\frac{1}{2}$ as she receives a utility of $-\frac{1}{4}$ from agents $1$ and $3$ within the second candidate coalition. However, before deviating, agent $2$'s utility was $1$, as discussed in the previous case. Namely, agent $2$'s utility decreased after deviating, meaning that she has no incentive to unilaterally deviate to $a_2'=\{2\}$. Agent $3$ has no incentive to unilaterally deviate to $a_3'=\{1\}$ by similar arguments.

                \item If agent $2$ unilaterally deviates to $a_2'=\{3\}$, then the resulting coalitions are $C_1^{(a_1,a_2',a_3)}=\{1\}$, $C_1^{(a_1,a_2',a_3)}=\{1,3\}$ and $C_3^{(a_1,a_2',a_3)}=\{2\}$. Thus, agent $2$'s utility decreased from $1$ to $0$, meaning that she has no incentive to unilaterally deviate to $a_2'=\{3\}$. Agent $3$ has no incentive to unilaterally deviate to $a_3'=\{3\}$ by similar arguments.

                \item If agent $2$ unilaterally deviates to $a_2'=\{1,2\}$, then the resulting coalitions are $C_1^{(a_1,a_2',a_3)}=\{1,2\}$, $C_1^{(a_1,a_2',a_3)}=\{1,2,3\}$ and $C_3^{(a_1,a_2',a_3)}=\emptyset$. Thus, agent $2$'s utility from the resulting partition is $1-2\cdot\frac{1}{4}=\frac{1}{2}$, as she obtains a utility of $1$ from agent $1$ within the first candidate coalition and a utility of $-\frac{1}{4}$ from both agents $1$ and $3$ within the second candidate coalition. However, before deviating, agent $2$'s utility was $1$. Namely, agent $2$'s utility decreased after deviating, meaning that she has no incentive to unilaterally deviate to $a_2'=\{1,2\}$. Agent $3$ has no incentive to unilaterally deviate to $a_3'=\{1,2\}$ by similar arguments.

            \end{enumerate}

            \item \textbf{All pure joint strategies where two agents join either the first or second candidate coalitions while the third agent only joins the third candidate coalition are Nash-stable:} Without loss of generality, consider the joint action $\mathbf{a}$ where $a_1=\{1\}$, $a_2=\{1\}$ and $a_3=\{3\}$, i.e., $C_1^{\mathbf{a}}=\{1,2\}$, $C_2^{\mathbf{a}}=\emptyset$ and $C_3^{\mathbf{a}}=\{3\}$ (other cases are symmetric by \eqref{eq:deterministic utilities semi bandit 2}). Note that both agents $1,2$ receive a utility of $1$ from each other, while agent $3$ obtains $0$ utility. We distinguish between the following cases:
            \begin{enumerate}
                \item If agent $1$ unilaterally deviates to $a_1'=\{2\}$, then the resulting coalitions are $C_1^{(a_1',a_2,a_3)}=\{2\}$, $C_2^{(a_1',a_2,a_3)}=\{1\}$ and $C_3^{(a_1',a_2,a_3)}=\{3\}$. Thus, agent $1$'s utility decreased from $1$ to $0$, meaning that she has no incentive to unilaterally deviate to $a_1'=\{2\}$. Agent $2$ has no incentive to unilaterally deviate to $a_2'=\{2\}$ by similar arguments.

                \item If agent $1$ unilaterally deviates to $a_1'=\{3\}$, then the resulting coalitions are $C_1^{(a_1',a_2,a_3)}=\{2\}$, $C_2^{(a_1',a_2,a_3)}=\emptyset$ and $C_3^{(a_1',a_2,a_3)}=\{1,3\}$. Thus, agent $1$'s utility from the resulting partition is $-\frac{1}{2}$ as she receives a utility of $-\frac{1}{2}$ from agent $3$ within the third candidate coalition. However, before deviating, agent $1$'s utility was $1$. Namely, agent $1$'s utility decreased after deviating, meaning that she has no incentive to unilaterally deviate to $a_1'=\{3\}$. Agent $2$ has no incentive to unilaterally deviate to $a_2'=\{3\}$ by similar arguments.

                \item If agent $1$ unilaterally deviates to $a_1'=\{1,2\}$, then the resulting coalitions are $C_1^{(a_1',a_2,a_3)}=\{1,2\}$, $C_2^{(a_1',a_2,a_3)}=\{1\}$ and $C_3^{(a_1',a_2,a_3)}=\{3\}$. Thus, agent $1$'s utility from the resulting partition remains $1$, as she only receives a utility of $1$ from agent $2$ within the first candidate coalition. Namely, agent $1$'s utility remained unchanged after deviating, meaning that she has no incentive to unilaterally deviate to $a_1'=\{1,2\}$. Agent $2$ has no incentive to unilaterally deviate to $a_2'=\{1,2\}$ by similar arguments.

                \item If agent $3$ unilaterally deviates to $a_3'=\{1\}$, then the resulting coalitions are $C_1^{(a_1,a_2,a_3')}=\{1,2,3\}$, $C_2^{(a_1,a_2,a_3')}=\emptyset$ and $C_3^{(a_1,a_2,a_3')}=\emptyset$.  Thus, agent $3$'s utility from the resulting partition is $-2\cdot\frac{1}{4}=-\frac{1}{2}$, as she obtains a utility of $-\frac{1}{4}$ from both agents $1$ and $2$ within the first candidate coalition. However, before deviating, agent $3$'s utility was $0$. Namely, agent $3$'s utility decreased after deviating, meaning that she has no incentive to unilaterally deviate to $a_3'=\{1\}$.

                \item If agent $3$ unilaterally deviates to $a_3'=\{2\}$, then the resulting coalitions are $C_1^{(a_1,a_2,a_3')}=\{1,2\}$, $C_2^{(a_1,a_2,a_3')}=\{3\}$ and $C_3^{(a_1,a_2,a_3')}=\emptyset$. Thus, agent $3$'s remains $0$, meaning that she has no incentive to unilaterally deviate to $a_3'=\{2\}$. 

                \item If agent $3$ unilaterally deviates to $a_3'=\{1,2\}$, then the resulting coalitions are $C_1^{(a_1,a_2,a_3')}=\{1,2,3\}$, $C_2^{(a_1,a_2,a_3')}=\{3\}$ and $C_3^{(a_1,a_2,a_3')}=\emptyset$. Thus, agent $3$'s utility from the resulting partition is $-2\cdot\frac{1}{4}=-\frac{1}{2}$, as she obtains a utility of $-\frac{1}{4}$ from both agents $1$ and $2$ within the first candidate coalition and she gets zero utility from being alone in the second candidate coalition. However, before deviating, agent $3$'s utility was $0$. Namely, agent $3$'s utility decreased after deviating, meaning that she has no incentive to unilaterally deviate to $a_3'=\{1\}$.
            \end{enumerate}

            \item \textbf{No other pure joint strategy is Nash-stable:} We distinguish between the following cases:
            \begin{enumerate}
                \item \textbf{Two agents join one of the first two candidate coalitions, while the third agent joins the other among these two coalitions:} Without loss of generality, consider the joint action $\mathbf{a}$ with $a_1=a_2=\{1\}$ and $a_3=\{2\}$, i.e., $C_1^{\mathbf{a}}=\{1,2\}$, $C_2^{\mathbf{a}}=\{3\}$ and $C_3^{\mathbf{a}}=\emptyset$ (other cases are symmetric by \eqref{eq:deterministic utilities semi bandit 2}). While both agents $1$ and $2$ obtain a utility of $1$, agent $3$ obtain zero utility. As we have already proven, this means that agent $3$ has an incentive to unilaterally deviate by joining both the first and second candidate coalitions.

                \item \textbf{Both the first and second candidate coalitions are of size exactly $3$:} Suppose the partition induced by $\mathbf{a}$ contains two coalitions of size exactly $3$, i.e., $a_1=a_2=a_3=\{1,2\}$, meaning that $C_1^{\mathbf{a}}=\{1,2,3\}$, $C_2^{\mathbf{a}}=\{1,2,3\}$ and $C_3^{\mathbf{a}}=\emptyset$. Each agent obtains a utility of $4\cdot (-\frac{1}{4})=-1$, as she receives a utility of $-\frac{1}{4}$ from any other agent within each coalition. Each agent therefore has an incentive to unilaterally deviate by, e.g., only joining one coalition, thus increasing her utility to $2\cdot (-\frac{1}{4})=-\frac{1}{2}$.
            \end{enumerate}
        \end{enumerate}
    \end{proof}

    \paragraph{Construction of the exploration policy $\rho$.} For both games $G_1, G_2$, we construct the same exploration policy $\rho$, which picks a joint action uniformly at random from the set of all joint actions where either: (1) all agents join both the first and second candidate coalitions; (2) only two agents join both the first and second candidate coalitions while the third agent joins the third candidate coalition; or (3) only two agents join the third candidate coalition while the third agent joins either the first or the second candidate coalition. 
    Note that there are exactly $10$ such joint actions, as listed later in \eqref{eq:exploration policy semi}. All other joint actions are assigned zero probability under $\rho$. Hence, the exploration policy $\rho$ is given by:
    \begin{equation}
        \label{eq:exploration policy semi}
        \rho(\mathbf{a})=\begin{cases}
            \frac{1}{10} &, a_1=a_2=a_3=\{1,2\}\\
            \frac{1}{10} &, a_1=a_2=\{1,2\},a_3=\{3\} \text{ \textbf{or} } a_1=a_3=\{1,2\},a_2=\{3\} \text{ \textbf{or} } a_2=a_3=\{1,2\},a_1=\{3\}\\
            \frac{1}{10} &, a_1=a_2=\{3\},a_3=\{1\} \text{ \textbf{or} } a_1=a_2=\{3\},a_3=\{2\} \text{ \textbf{or} } a_1=a_3=\{3\},a_2=\{1\} \text{ \textbf{or} } \\ & \text{ } \text{ } a_1=a_3=\{3\},a_2=\{2\} \text{ \textbf{or} } a_2=a_3=\{3\},a_1=\{1\} \text{ \textbf{or} } a_2=a_3=\{3\},a_1=\{2\} \\
            0 &,\text{ o.w.}
        \end{cases}
    \end{equation}

    \paragraph{Proving the statement in Theorem \ref{supp:thm:1/8 gap}.} Subsequently, we prove the statement in Theorem \ref{supp:thm:1/8 gap}. First, we note that the pairs $(G_1,\rho)$ and $(G_2,\rho)$ both belong to the class $\mathcal{G}'$ stated in Theorem \ref{supp:thm:1/8 gap}, i.e., both games with the exploration policy $\rho$ satisfy Assumption 1 in the full paper. Indeed:
    \begin{itemize}
        \item \underline{For the first game $G_1$:} By Lemma \ref{supp:lemma:NS game 1 bandit}, consider a Nash stable pure joint strategy $\bm{\varphi}^\star$ where two agents join the first and second candidate coalitions, while the third joins only the first candidate coalition. Further, consider a Nash stable pure joint strategy $\tilde{\bm{\varphi}}^\star$ where two agents join the second candidate coalition and the third agent only joins the third candidate coalition. We will show that, for any agent $i$, any $\ell \in [3]$ and any coalition size $\alpha\in[n]\cup\{0\}$, if there is a strategy $\phi_i \in \Delta(\mathcal{A}_i)$ with either $d_\ell^{\bm{\varphi}_{-i}^\star,\phi_i}(\alpha)>0$ or $d_\ell^{\tilde{\bm{\varphi}}_{-i}^\star,\phi_i}(\alpha)>0$, then $d_\ell^{\rho}(\alpha) > 0$. Indeed:
        \begin{enumerate}
            \item Considering $\bm{\varphi}^\star$, this already holds for $\ell= 1, \alpha=3$ and $\ell= 2, \alpha=2$ and $\ell=3, \alpha=0$ by picking $\phi_i=\varphi_i^\star$ for any agent $i$.

            \item Considering $\tilde{\bm{\varphi}}^\star$, this already holds for $\ell= 1, \alpha=0$ and $\ell= 2, \alpha=2$ and $\ell=3, \alpha=1$ by picking $\phi_i=\tilde{\varphi}_i^\star$ for any agent $i$. The case $\ell=3, \alpha=1$ is also covered when each agent $i$ unilaterally deviates from $\bm{\varphi}^\star$ by only joining the third candidate coalition, and it is also covered by $\rho$.

            \item $\ell=1, \alpha=1$ is covered when each agent $i$ unilaterally deviates from $\bm{\varphi}^\star$ by only joining the first candidate coalition, and it is also covered by $\rho$.

            \item $\ell=1, \alpha=2$ is covered when each agent $i$ unilaterally deviates from $\bm{\varphi}^\star$ by only joining the second or the third candidate coalition, and it is also covered by $\rho$.  


            \item $\ell=2, \alpha=0$ is not covered by any unilateral deviation, and it is also uncovered by $\rho$. 

            \item $\ell=2, \alpha=1$ is covered when one of the agents in the second candidate coalition induced by ${\bm{\varphi}}^\star$ unilaterally deviates from $\bm{\varphi}^\star$ by only joining, e.g., the first candidate coalition, and it is also covered by $\rho$.

            \item $\ell=2, \alpha=3$ is not covered by any unilateral deviation, and it is also uncovered by $\rho$.    
            
            \item $\ell=3, \alpha=2$ is covered when one of the agents in the second candidate coalition induced by $\tilde{\bm{\varphi}}^\star$ unilaterally deviates from $\tilde{\bm{\varphi}}^\star$ by joining the third candidate coalition, and it is also covered by $\rho$.

            \item $\ell=3, \alpha=3$ is not covered by any unilateral deviation, and it is also uncovered by $\rho$.    
        \end{enumerate}
        Hence, $(G_1,\rho)$ satisfies Assumption 1 in the full paper, yielding that $(G_1,\rho) \in \mathcal{G}$.

        \item \underline{For the second game $G_2$:} By Lemma \ref{supp:lemma:NS game 2 bandit}, consider a Nash stable pure joint strategy $\bm{\varphi}^\star$ where both the first and second candidate coalitions are of size exactly $2$. Further, consider a Nash stable pure joint strategy $\tilde{\bm{\varphi}}^\star$ where two agents join either the first or second candidate coalitions, while the third agent only joins the third candidate coalition. We will show that, for any agent $i$, any $\ell \in [3]$ and any coalition size $\alpha\in[n]\cup\{0\}$, if there is a strategy $\phi_i \in \Delta(\mathcal{A}_i)$ with either $d_\ell^{\bm{\varphi}_{-i}^\star,\phi_i}(\alpha)>0$ or $d_\ell^{\tilde{\bm{\varphi}}_{-i}^\star,\phi_i}(\alpha)>0$, then $d_\ell^{\rho}(\alpha) > 0$. Indeed:
        \begin{enumerate}
            \item Considering $\bm{\varphi}^\star$, this already holds for $\ell= 1, \alpha=2$ and $\ell= 2, \alpha=2$ and $\ell=3, \alpha=0$ by picking $\phi_i=\varphi_i^\star$ for any agent $i$.

            \item Considering $\tilde{\bm{\varphi}}^\star$, this already holds for $\ell= 1, \alpha=2$ and $\ell= 1, \alpha=0$ and $\ell= 2, \alpha=2$ and $\ell= 2, \alpha=0$ and $\ell=3, \alpha=1$ by picking $\phi_i=\tilde{\varphi}_i^\star$ for any agent $i$. 

            \item $\ell=1, \alpha=1$ is covered when each agent in the first candidate coalition induced by $\bm{\varphi}^\star$ unilaterally deviates from $\bm{\varphi}^\star$ by only joining the second or third candidate coalition, and it is also covered by $\rho$.

            \item $\ell=1, \alpha=3$ is covered when the agent that only joins the second candidate coalition induced by $\bm{\varphi}^\star$ unilaterally deviates from $\bm{\varphi}^\star$ by only joining both the first and the second candidate coalition, and it is also covered by $\rho$.  


            \item $\ell=2, \alpha=1$ is covered when each agent in the second candidate coalition induced by $\bm{\varphi}^\star$ unilaterally deviates from $\bm{\varphi}^\star$ by only joining the first or third candidate coalition, and it is also covered by $\rho$.

            \item $\ell=2, \alpha=3$ is covered when the agent that only joins the first candidate coalition induced by $\bm{\varphi}^\star$ unilaterally deviates from $\bm{\varphi}^\star$ by only joining both the first and the second candidate coalition, and it is also covered by $\rho$.      
            
            \item $\ell=3, \alpha=2$ is covered when one of the agents that joins the first or second candidate coalition induced by $\tilde{\bm{\varphi}}^\star$ unilaterally deviates from $\tilde{\bm{\varphi}}^\star$ by joining the third candidate coalition, and it is also covered by $\rho$.

            \item $\ell=3, \alpha=3$ is not covered by any unilateral deviation, and it is also uncovered by $\rho$.    
        \end{enumerate}
        Hence, $(G_2,\rho)$ satisfies Assumption 1 in the full paper, yielding that $(G_2,\rho) \in \mathcal{G}$.
    \end{itemize}

    As such, consider any algorithm $\alg$. All the utility information that we may obtain from any dataset sampled from the exploration policy $\rho$ is: 
    \begin{enumerate}
        \item Consider the joint action $\mathbf{a}$ for which $a_1=a_2=a_3=\{1,2\}$ sampled from $\rho$ with probability $\frac{1}{10}$. In the first game $G_1$, by \eqref{eq:deterministic utilities semi bandit}, the utility information observed for each agent $i$ is $\sum_{i \neq j \in \mathcal{N}} [v_{i,j}^1+v_{i,j}^2] = \sum_{i \neq j \in \mathcal{N}} [\frac{1}{2}-1]= 2\cdot(-\frac{1}{2})=-1$. In the second game $G_2$, by \eqref{eq:deterministic utilities semi bandit 2}, the utility information for each agent $i$ is $\sum_{i \neq j \in \mathcal{N}} [v_{i,j}^1+v_{i,j}^2] = \sum_{i \neq j \in \mathcal{N}} [-\frac{1}{4}-\frac{1}{4}]= 2\cdot(-\frac{1}{2})=-1$. 

        \item Consider the joint action $\mathbf{a}$ where two agents join both the first and second candidate coalitions while the third agent joins the third candidate coalition, where $\mathbf{a}$ is sampled from $\rho$ with probability $\frac{1}{10}$. Without loss of generality, assume that $a_1=a_2=\{1,2\},a_3=\{3\}$ (other cases are symmetric). In both games, by \eqref{eq:deterministic utilities semi bandit} and \eqref{eq:deterministic utilities semi bandit 2}, the utility information observed for agents $1$ and $2$ is the same and equals $4$, while agent $3$'s utility in both games is $0$.

        \item Consider the joint action $\mathbf{a}$ where two agents join both the third candidate coalition while the third agent joins either the first or the second candidate coalition, where $\mathbf{a}$ is sampled from $\rho$ with probability $\frac{1}{10}$. Without loss of generality, assume that $a_1=a_2=\{3\},a_3=\{1\}$ (other cases are symmetric). In both games, by \eqref{eq:deterministic utilities semi bandit} and \eqref{eq:deterministic utilities semi bandit 2}, the utility information observed for agents $1$ and $2$ is the same and equals $-\frac{1}{2}$, while agent $3$'s utility in both games is $0$.
    \end{enumerate}
    That is, \textbf{in any case, the utility feedbacks in both games are the same, meaning that, regardless of the size of the dataset obtained by $\alg$, the algorithm $\alg$ \textit{cannot} distinguish between the games $G_1,G_2$ and they appear behaviorally the same from the perspective of the data available under $\rho$}. 

    Therefore, we denote by $q$ the probability that the joint action $\mathbf{a}$ where all $3$ agents join the first candidate coalition is sampled from the joint strategy $\bm{\varphi}$ produced by $\alg$. Since it has a probability of $\frac{1}{10}$ for being sample from $\rho$, then we can obtain that $\gap(\varphi)\geq \frac{1-p}{10}$ in the first game $G_1$ and $\gap(\varphi)\geq \frac{p}{10}$ in the first game $G_2$ by arguments similar to Theorem \ref{supp:thm:half gap}. since either $q\geq \frac{1}{2}$ or $1-q\geq \frac{1}{2}$, the joint strategy $\bm{\varphi}$ produced by $\alg$ is always satisfies $\gap(\bm{\varphi})\geq \frac{1}{20}$ for at least one game, regardless of the dataset size, as desired.   
\end{proof}

\section{Proof of Theorem 5}
\label{supp:thm 5}
The proof of Theorem 5 under semi-bandit feedback relies on Lemma \ref{supp:lemma:bonus bandit}, which proves that, for any agent $i$, our constructed $b_i^\delta$ is an exploration bonus for the utility estimator $\hat{v}_i$.
\begin{lemma}
    \label{supp:lemma:bonus bandit}
    For any $\delta \in (0,1]$, the following holds with probability at least $1-\delta$ simultaneously for each agent $i$ and any joint action $\mathbf{a}$:
    \begin{equation}
        \label{eq:bonus estimation}
        |d_i(\mathbf{a})-\hat{v}_i(\mathbf{a})|\leq b_i^\delta(\mathbf{a})
    \end{equation}
\end{lemma}
\begin{proof}
    Under the notation of Section 5.1 in the full paper, the following holds with probability at least $1-\delta$ as a special case of Theorem 20.5 by Lattimore and Szepesv\'{a}ri \cite{Lattimore_Szepesvari_2020}:
    \begin{equation}
        \label{eq:lat}
        \|\hat{\bm{\theta}}-\bm{\theta}\|_V \leq \|\bm{\theta}\|_2+\sqrt{\log(\det(V))+2\log(1/\delta)}
    \end{equation}
    Thus, for any agent $i$ and any joint action $\mathbf{a}$, the following is satisfied with probability at least $1-\delta$:
    \begin{equation}
        \label{eq:diffush}
        |d_i(\mathbf{a})-\hat{v}_i(\mathbf{a})|= |\langle\mathbf{z}_{i}(\mathbf{a}), \hat{\bm{\theta}}-\bm{\theta} \rangle| \leq \|\mathbf{z}_{i}(\mathbf{a})\|_{V^{-1}} \|\hat{\bm{\theta}}-\bm{\theta}\|_V \leq \|\mathbf{z}_{i}(\mathbf{a})\|_{V^{-1}}  \left[\|\bm{\theta}\|_2+\sqrt{\log(\det(V))+2\log(1/\delta)}\right]
    \end{equation}
    where the first inequality is by Cauchy-Schwartz inequality and the last inequality is by \eqref{eq:lat}. Next, we upper bound $\det(V)$, which is the determinant of the covariance matrix $V=I+\sum_{m\in[M]} \sum_{i\in \mathcal{N}} \mathbf{z}_{i}(\mathbf{a}^k)\mathbf{z}_{i}(\mathbf{a}^k)^\top$ of the dataset. Note that $V$ is a matrix of order $(n^2k)\times (n^2k)$. Thus, letting $\lambda_1,\dots,\lambda_{n^2k}$ be the eigenvalues of $V$, it holds that:
    \begin{equation}
        \label{eq:bound det}
        \det(V) = \prod_{\kappa=1}^{n^2k} \lambda_\kappa \leq \left(\frac{\tr(V)}{n^2k}\right)^{n^2k} = \left(\frac{\tr(I)+\sum_{m\in[M]} \sum_{i\in \mathcal{N}} \|\mathbf{z}_{i}(\mathbf{a}^k)\|_2^2}{n^2k}\right)^{n^2k} \leq \left(1+\frac{nkM}{n^2k}\right)^{n^2k} = \left(1+\frac{M}{n}\right)^{n^2k}
    \end{equation}
    where the first inequality is by the AM-GM inequality and the last inequality is due to $\|\mathbf{z}_{i}(\mathbf{a}^k)\|_2^2 \leq k$.

    Combining \eqref{eq:bound det} with \eqref{eq:diffush} and $\|\bm{\theta}\|_2\leq 2\sqrt{n^2k}$, we obtain the desired by applying a union bound.
\end{proof}

We are now ready to prove Theorem 5:
\begin{customthm}{5}
    \label{supp:thm:bandit}
    Under {\normalfont bandit} feedback and Assumption 2, for any $\delta \in (0,1]$ and dataset size $M \in \mathbb{N}$, Algorithm 1 with utility estimators and exploration bonuses defined in equation (11) within the full paper outputs a joint strategy $\bm{\varphi}^{\out}$ satisfying $\gap(\varphi^{\out}) \leq 4\sqrt{\frac{n^2k\beta}{c_{\act}M}}+\epsilon_{\opt}$ with probability at least $1-\delta$, where $\sqrt{\beta}$ is as in equation (11) within the full paper.
\end{customthm}
\begin{proof}
    By Theorem \ref{supp:thm:general duality gap}, we need to bound $\mathbb{E}_{\mathbf{a}\sim (\bm{\varphi}_{-i}^\star,\varphi_i')} [b_i^\delta(\mathbf{a})]$ and $ \mathbb{E}_{\mathbf{a}\sim \bm{\varphi}^\star} [b_i^\delta(\mathbf{a})]$ for some NS (possibly mixed) joint strategy $\bm{\varphi}^\star$ and any strategy $\varphi_i \in \Delta(\mathcal{A}_i)$ (in particular, we can $\varphi_i$ choose to be deterministic by Theorem \ref{supp:thm:general duality gap}). Indeed, for any $\delta \in (0,1]$, the following holds with probability at least $1-\delta$ simultaneously for each agent $i$:
    \begin{subequations}
        \label{eq:bound bonus bandit}
        \begin{align}
             \mathbb{E}_{\mathbf{a}\sim (\bm{\varphi}_{-i}^\star,\varphi_i')} [b_i^\delta(\mathbf{a})] &= \sqrt{\beta} \mathbb{E}_{\mathbf{a}\sim (\bm{\varphi}_{-i}^\star,\varphi_i')} \sqrt{\mathbf{z}_{i}(\mathbf{a})^\top V^{-1} \mathbf{z}_{i}(\mathbf{a})} \\
             &\leq \sqrt{\beta}\mathbb{E}_{\mathbf{a}\sim (\bm{\varphi}_{-i}^\star,\varphi_i')} \sqrt{\mathbf{z}_{i}(\mathbf{a})^\top [I+M c_{\act} \mathbb{E}_{\mathbf{a}'\sim (\bm{\varphi}_{-i}^\star,\phi_i)}\left[\mathbf{z}_{i}(\mathbf{a}')\mathbf{z}_{i}(\mathbf{a}')^\top]\right]^{-1} \mathbf{z}_{i}(\mathbf{a})} \label{eq:assump2}  \\
             &=\sqrt{\beta}\mathbb{E}_{\mathbf{a}\sim (\bm{\varphi}_{-i}^\star,\varphi_i')} \sqrt{\tr[ [I+M c_{\act} \mathbb{E}_{\mathbf{a}'\sim (\bm{\varphi}_{-i}^\star,\phi_i)}\left[\mathbf{z}_{i}(\mathbf{a}')\mathbf{z}_{i}(\mathbf{a}')^\top]\right]^{-1} \mathbf{z}_{i}(\mathbf{a}^k)\mathbf{z}_{i}(\mathbf{a}^k)^\top]} \\
             &\leq\sqrt{\beta} \sqrt{\tr[ [I+M c_{\act} \mathbb{E}_{\mathbf{a}'\sim (\bm{\varphi}_{-i}^\star,\phi_i)}\left[\mathbf{z}_{i}(\mathbf{a}')\mathbf{z}_{i}(\mathbf{a}')^\top]\right]^{-1} \mathbb{E}_{\mathbf{a}\sim (\bm{\varphi}_{-i}^\star,\varphi_i')}[\mathbf{z}_{i}(\mathbf{a}^k)\mathbf{z}_{i}(\mathbf{a}^k)^\top]]} \label{eq:jensen} \\
             &= \sqrt{\frac{\beta}{M c_{\act}}} \sqrt{\tr[ I-(I+M c_{\act} \mathbb{E}_{\mathbf{a}'\sim (\bm{\varphi}_{-i}^\star,\phi_i)}\left[\mathbf{z}_{i}(\mathbf{a}')\mathbf{z}_{i}(\mathbf{a}')^\top])^{-1}\right]} \leq \sqrt{\frac{\beta n^2 k}{M c_{\act}}}
        \end{align}
    \end{subequations}
    where the inequality in \eqref{eq:assump2} follows from Assumption 2 and the inequality in \eqref{eq:jensen} is by Jensen's inequality. As $\mathbb{E}_{\mathbf{a}\sim (\bm{\varphi}_{-i}^\star,\varphi_i')} [b_i^\delta(\mathbf{a})] \leq \sqrt{\frac{\beta n^2 k}{M c_{\act}}}$, we easily obtain the desired from Theorem \ref{supp:thm:general duality gap}.
\end{proof}

\subsection{Analysis of the Universal Constant in the Action Coverage Assumption (Assumption 2)}
Consider a POCF game $G$ where each agent $i$ can join any possible non-empty subset of candidate coalitions. That is, the action space of each agent $i$ is $\mathcal{A}_i = \mathcal{P}([k])\setminus \{\emptyset\}$, which is exponential in the number of candidate coalitions (i.e., $|\mathcal{A}_i|=2^k-1$). By Theorem \ref{supp:thm:general duality gap}, the Nash stable joint strategy $\bm{\varphi}^\star$ in Assumption 2 can be chosen to be \textit{deterministic}. Consider an exploration policy $\rho$, which picks a joint action uniformly at random from the set of all joint actions where only one agent unilaterally deviates from $\bm{\varphi}^\star$ by either joining or leaving at most one coalition. For example, if agent $1$'s action according to the pure NS strategy $\bm{\varphi}^\star$ is $\{1,2\}$, then the exploration policy $\rho$ covers the joint actions induced by agent $1$'s unilateral deviations to $\{1,2\}, \{1\}, \{2\}, \{1,2,3\}, \{1,2,4\}, \dots, \{1,2,k\}$. Besides the action dictated by $\bm{\varphi}^\star$ for each agent, there are $k$ possible unilateral deviations for every agent, and thus the exploration policy $\rho$ selects a joint action uniformly at random from a set of $kn+1$ action. That is, for any joint action $\mathbf{a}$ for which $\rho$ assigns a positive probability, it holds that $\rho(\mathbf{a})=\frac{1}{kn+1}$.

Given a dataset $\mathcal{S}=\{(\mathbf{a}^m, \mathbf{v}^m)\}_{m=1}^M$ of $M$ samples independently drawn from $\rho$, we denote the number of samples where a joint action $\mathbf{a}$ appears by $N_\mathcal{S}(\mathbf{a}):=|\{k \in [M] : \mathbf{a}=\mathbf{a}^k\}|$. Hence, the following holds:
\begin{lemma}
    \label{supp:lemma:lower bound counter}
    For any $\delta \in (0,1]$ and any joint action $\mathbf{a} \in \mathcal{A}$ with $\rho(\mathbf{a})>0$, if the dataset size is at least $M \geq 8 (nk+1)\log(\frac{nk+1}{\delta})$, then joint action $\mathbf{a}$ appears at least $\frac{\rho(\mathbf{a})M}{2}=\frac{M}{2(nk+1)}$ times in the dataset (i.e., $N_\mathcal{S}(\mathbf{a}) \geq \frac{M}{2(nk+1)}$) with probability at least $1-\delta$.
\end{lemma}
\begin{proof}
    Note that $N_\mathcal{S}(\mathbf{a})$ is the number of times a joint action $\mathbf{a}$ appears in a sample of $M$ independently drawn joint actions from $\rho$, meaning that $N_\mathcal{S}(\mathbf{a})$ follows a binomial distribution with parameters $M$ and $\rho(\mathbf{a})$, i.e., $N_\mathcal{S}(\mathbf{a}) \sim \text{Binomial}(M,\rho(\mathbf{a}))$. Using the Chernoff bound, the following is satisfied for any $\gamma >0$:
    \begin{equation}
        \label{eq:chernoff}
        \mathbb{P}[N_\mathcal{S}(\mathbf{a}) \geq (1-\gamma)M\rho(\mathbf{a})]\geq 1-\exp\left(-\frac{\gamma^2 M\rho(\mathbf{a})}{2}\right)
    \end{equation}
    Therefore, for any $\mathbf{a} \in \mathcal{A}$ with $\rho(\mathbf{a})>0$, if $\gamma^2 \geq -\frac{2\log(\delta)}{M\rho(\mathbf{a})}$, then the following holds:
    \begin{equation}
        \mathbb{P}[N_\mathcal{S}(\mathbf{a}) \geq (1-\gamma)M\rho(\mathbf{a})]\geq 1-\exp\left(\log(\delta)\right)=1-\delta
    \end{equation}
    Using a union bound and the fact that $\rho(\mathbf{a})=\frac{1}{kn+1}$ for any $\mathbf{a} \in \mathcal{A}$ with $\rho(\mathbf{a})>0$, we obtain the desired.
\end{proof}

Next, we prove that, if the dataset is sufficiently large, then Assumption 2 in the full paper holds for $c_{\act}=\frac{1}{2nk^4}$ with high probability.
\begin{lemma}
    For any $\delta \in (0,1]$, if the dataset size is at least $M \geq 8 (nk+1)\log(\frac{nk+1}{\delta})$ and $c_{\act}=\frac{1}{2nk^4}$, then Assumption 2 in the full paper holds for $c_{\act}=\frac{1}{2nk^4}$ with probability at least $1-\delta$.
\end{lemma}
\begin{proof}
    Recall that Assumption 2 requires the existence of a universal constant $c_{\act}>0$ and an NS joint strategy $\bm{\varphi}^\star$ such that, for any agent $i$ and any strategy $\phi_i \in \Delta(\mathcal{A}_i)$, it holds that:
    \begin{equation}
        \label{supp:eq:action coverage}
        V\succeq I+M c_{\act} \mathbb{E}_{\mathbf{a}\sim (\bm{\varphi}_{-i}^\star,\phi_i)}[\mathbf{z}_{i}(\mathbf{a})\mathbf{z}_{i}(\mathbf{a})^\top]
    \end{equation}
    where $V=I+\sum_{m\in[M]} \sum_{i\in \mathcal{N}} \mathbf{z}_{i}(\mathbf{a}^k)\mathbf{z}_{i}(\mathbf{a}^k)^\top$ is the covariance matrix of the dataset. By Theorem \ref{supp:thm:general duality gap}, the Nash stable joint strategy $\bm{\varphi}^\star$ in Assumption 2 can be chosen to be \textit{deterministic}. Further, as the right-hand side in \eqref{supp:eq:action coverage} is linear in each entry of any strategy $\phi_i \in \Delta(\mathcal{A}_i)$ of agent $i$, we need to focus only on pure strategies. 
    
    That is, we hereafter consider some agent $i$ that unilaterally deviates from $\bm{\varphi}^\star$ to a pure strategy $\phi_i \in \Delta(\mathcal{A}_i)$, resulting in a pure joint strategy $(\bm{\varphi}^\star_{-i}, \phi_i)$ from which a joint action $\mathbf{a}'$ is deterministically picked. Further, let $\mathbf{a}^\star$ be the joint action deterministically picked by $\bm{\varphi}^\star$. Without loss of generality, assume that the candidate coalitions affected by the unilateral deviation from $\bm{\varphi}^\star$ that leads to the joint action $\mathbf{a}'$ are those that correspond to the indices $1,2,\dots,\ell'$ for some $\ell' \in [k]$. To estimate the change in utility resulting from such a deviation, we require some joint action $\mathbf{a}^\ell$ resulting from a unilateral deviation from $\mathbf{a}^\star$ that solely affects the $\ell$-th candidate coalition for any $\ell \in [\ell']$.

    Now, recall that agent $i$'s mean utility from a joint action $\mathbf{a} \in \mathcal{A}$ can be written as $d_i(\mathbf{a}) =\langle\mathbf{z}_{i}(\mathbf{a}), \bm{\theta} \rangle$, where $\bm{\theta}$ concatenates all agents' individual utilities into an $n^2 k$-dimensional vector. Without loss of generality, we assume that the change in utility is given by $\langle\mathbf{z}_{i}(\mathbf{a}^\ell)-\mathbf{z}_{i}(\mathbf{a}^\star), \bm{\theta} \rangle$, and thus $\mathbf{z}_{i}(\mathbf{a}') = \mathbf{z}_{i}(\mathbf{a}^\star)+\sum_{\ell'\in[\ell]} (\mathbf{z}_{i}(\mathbf{a}^{\ell'})-\mathbf{z}_{i}(\mathbf{a}^\star))= (1-\ell)\mathbf{z}_{i}(\mathbf{a}^\star)+ \sum_{\ell'\in[\ell]} \mathbf{z}_{i}(\mathbf{a}^{\ell'})$. Thereby, by Lemma \ref{supp:lemma:lower bound counter} and since $(\bm{\varphi}^\star_{-i}, \phi_i)$ is deterministic, we need to require that the following holds with probability at least $1-\delta$ for any $\delta \in (0,1]$ so that the condition in \eqref{eq:action coverage} of Assumption 2 to hold:
    \begin{equation}
       I+\frac{M(1-\ell)^2}{2(nk+1)}\mathbf{z}_{i}(\mathbf{a}^\star)\mathbf{z}_{i}(\mathbf{a}^\star)^\top+ \frac{M}{2(nk+1)}\sum_{\ell'\in[\ell]} \mathbf{z}_{i}(\mathbf{a}^{\ell'}) \mathbf{z}_{i}(\mathbf{a}^{\ell'})^\top \succeq I+M c_{\act} \mathbf{z}_{i}(\mathbf{a}')\mathbf{z}_{i}(\mathbf{a}')^\top
    \end{equation}
    which is equivalent to requiring the following for any vector $\mathbf{x}\in \mathbb{R}^{n^2k}$:
    \begin{equation}
        \label{eq:semidefinite}
        \frac{M(1-\ell)^2}{2(nk+1)} \mathbf{x}^\top \mathbf{z}_{i}(\mathbf{a}^\star)\mathbf{z}_{i}(\mathbf{a}^\star)^\top \mathbf{x}+ \frac{M}{2(nk+1)}\sum_{\ell'\in[\ell]} \mathbf{x}^\top \mathbf{z}_{i}(\mathbf{a}^{\ell'}) \mathbf{z}_{i}(\mathbf{a}^{\ell'})^\top \mathbf{x} \geq M c_{\act} \mathbf{x}^\top \mathbf{z}_{i}(\mathbf{a}')\mathbf{z}_{i}(\mathbf{a}')^\top \mathbf{x}
    \end{equation}
    Letting $x_{\ell'}=\mathbf{x}^\top \mathbf{z}_{i}(\mathbf{a}^{\ell'})$ and $x_{\star}=\mathbf{x}^\top \mathbf{z}_{i}(\mathbf{a}^\star)$, \eqref{eq:semidefinite} can be rephrase as follows:
    \begin{equation}
        \label{eq:semidefinite2}
        \frac{M(1-\ell)^2}{2(nk+1)} x_{\star}^2+ \frac{M}{2(nk+1)}\sum_{\ell'\in[\ell]} x_{\ell'}^2 \geq M c_{\act} \left[(1-\ell)x_{\star}+\sum_{\ell'\in[\ell]} x_{\ell'}\right]^2 \geq M c_{\act} (\ell+1) \left[(1-\ell)^2 x_{\star}^2+\sum_{\ell'\in[\ell]} x_{\ell'}^2\right]^2
    \end{equation}
    where the last inequality is by Jensen’s inequality. Accordingly, it is sufficient to require:
    \begin{equation}
        \label{eq:semidefinite3}
         c_{\act} (k+1)(k-1)^2 \leq \frac{1}{2(nk+1)}
    \end{equation}
    from which the desired readily follows.
\end{proof}

\subsection{Proof of Corollary 2}
\begin{customcoro}{2}
    \label{supp:coro:optimal2}
    Under {\normalfont bandit} feedback and Assumption 2, for any $\delta \in (0,1]$, any $\varepsilon > \epsilon_{\opt}$ with $\epsilon_{\opt} = o(\varepsilon)$ and any NS joint strategy $\bm{\varphi}^\star$, Algorithm 1 with utility estimators and exploration bonuses as in equation (11) in the full paper has a sample complexity bound that {\normalfont \textbf{optimally}} depends on $\varepsilon$ (up to logarithmic factors): for a dataset of size $M \geq \frac{16n^2 k\beta}{c_{\act}(\varepsilon -\epsilon_{\opt})^2}$, $\varphi^{\out}$ is $\varepsilon$-NS (i.e., $\gap(\varphi^{\out}) \leq \varepsilon$) with probability at least $1-\delta$.
\end{customcoro}
\begin{proof}
    The proof is a continuation of the proof for Theorem \ref{supp:thm:bandit}. Indeed, if the right-hand side in $\gap(\varphi^{\out}) \leq 4\sqrt{\frac{n^2k\beta}{c_{\act}M}}+\epsilon_{\opt}$ is at most $\varepsilon$ for some $\varepsilon \geq \epsilon_{\opt}$ with $\epsilon_{\opt} = o(\varepsilon)$, then Algorithm 1 produces an $\varepsilon$-NS joint strategy with probability at least $1-\delta$ (i.e., $\gap(\varphi^{\out}) \leq \varepsilon$). This is equivalent to requiring that:
    \begin{equation}
        \begin{aligned}
            \varepsilon \sqrt{M}\geq 4\sqrt{\frac{n^2k\beta}{c_{\act}}}+\sqrt{M}\epsilon_{\opt} \Leftrightarrow  (\varepsilon -\epsilon_{\opt})\sqrt{M}\geq 4\sqrt{\frac{n^2k\beta}{c_{\act}}} \Leftrightarrow \sqrt{M} \geq \frac{4\sqrt{\frac{n^2k\beta}{c_{\act}}}}{\varepsilon -\epsilon_{\opt}} \Leftrightarrow M \geq \frac{16 n^2k\beta}{c_{\act}(\varepsilon -\epsilon_{\opt})^2}
        \end{aligned} 
    \end{equation}
    as desired. Since $\epsilon_{\opt} = o(\varepsilon)$, we have {\normalfont \textbf{optimal}} dependence on $\varepsilon$ (up to logarithmic factors) due to \cite{hassani2020stochastic,bai20provable}. 
\end{proof}

\subsection{Learning Approximate Nash-Stable Pure Strategies}
\label{supp:Learning Approximate Nash-Stable Pure Strategies - bandit}

Similarly to Appendix \ref{supp:Learning Approximate Nash-Stable Pure Strategies - bandit}, we herein concentrate on learning an approximate NS \textit{pure} strategy under \textit{bandit} feedback. As mentioned in Appendix \ref{supp:Learning Approximate Nash-Stable Pure Strategies - bandit}, since we focus on \textit{symmetric} POCF games, they always admit a pure NS strategy by Lemma \ref{supp:lemma:potential game}, which legitimates learning an approximate NS \textit{pure} strategy. Hence, we can adapt Algorithm 1 to produce a pure strategy. Indeed, let $\Gamma_i \subset \Delta(\mathcal{A}_i)$ be the set of agent $i$'s \textit{deterministic} strategies, i.e., each deterministic strategy $\varphi_i \in \Gamma_i$ corresponds to exactly one joint \textit{pure} strategy $a_i \in \mathcal{A}_i$, such that $\varphi_i(a_i)=1$ for any agent $i$ and $\varphi_i(a_i')=0$ for any other pure strategy $a_i \neq a_i' \in \mathcal{A}_i$. Thus, instead of approximately solving \eqref{eq:approx NS} over \textit{mixed} strategies, Algorithm 1 in the main text can be modified to find a joint \textit{deterministic} strategy $\bm{\varphi}^{\out}\in \Gamma := \prod_{i=1}^n \Gamma_i$ that approximately solves the minimization problem $\min_{\bm{\varphi}\in\Gamma} \max_{i\in \mathcal{N}} [\overline{V}_i^{\star,\delta}(\bm{\varphi}_{-i})-\underline{V}_{i}^\delta(\bm{\varphi})]$. Algorithm \ref{alg:surrogate min pure} summarizes the resulting algorithm.

Under bandit feedback, Algorithm \ref{alg:surrogate min pure}’s approximation to Nash stability under \textit{pure} strategies is identical to its approximation to Nash stability under \textit{mixed} strategies due to Theorem \ref{supp:thm:bandit}. Accordingly, our algorithm's sample complexity bound under bandit feedback is also identical for both \textit{pure} and \textit{mixed} strategies, which is proven in Corollary \ref{coro:optimal}. We formulate this in the following corollaries:
\begin{corollary}
    Under {\normalfont bandit} feedback and Assumption 2, for any $\delta \in (0,1]$ and dataset size $M \in \mathbb{N}$, Algorithm \ref{alg:surrogate min pure} with utility estimators and exploration bonuses defined in equation (11) within the full paper outputs a joint strategy $\bm{\varphi}^{\out}$ satisfying $\gap(\varphi^{\out}) \leq 4\sqrt{\frac{n^2k\beta}{c_{\act}M}}+\epsilon_{\opt}$ with probability at least $1-\delta$, where $\sqrt{\beta}$ is as in equation (11) within the full paper.
\end{corollary}
\begin{corollary}
    Under {\normalfont bandit} feedback and Assumption 2, for any $\delta \in (0,1]$, any $\varepsilon > \epsilon_{\opt}$ with $\epsilon_{\opt} = o(\varepsilon)$ and any NS joint strategy $\bm{\varphi}^\star$, Algorithm \ref{alg:surrogate min pure} with utility estimators and exploration bonuses as in equation (11) in the full paper has a sample complexity bound that {\normalfont \textbf{optimally}} depends on $\varepsilon$ (up to logarithmic factors): for a dataset of size $M \geq \frac{16n^2 k\beta}{c_{\act}(\varepsilon -\epsilon_{\opt})^2}$, $\varphi^{\out}$ is $\varepsilon$-NS (i.e., $\gap(\varphi^{\out}) \leq \varepsilon$) with probability at least $1-\delta$.
\end{corollary}

\section{Additional Experimental Results}

\subsection{Additional Empirical Evaluations under Semi-Bandit Feedback}
In this section, we supply additional experimental results for semi-bandit feedback. In Appendices \ref{supp:setup}--\ref{supp:results}, we consider both size-independent and size-dependent utility generation models, which show trends similar to the size-dependent versions discussed in the main text. In Appendix \ref{supp:mixed effects}, we further consider a utility generation model with mixed coalition-size-effects, which also exhibits trends similar to prior models.
\subsubsection{Experimental Setup}
\label{supp:setup}
For each run, we generate a game with $n$ agents with the same action set of size at least $3$, sampled uniformly at random from $\mathcal{P}([k])\setminus \emptyset$. Given a joint action $\mathbf{a}$, we sample the mutual utility $v_{i,j}^{\ell}$ of each pair of distinct agents $i,j$ in the $\ell$-th candidate coalition $C_\ell^{\mathbf{a}}$ using one of the following four utility generation models inspired by Boehmer et al. \cite{boehmer2025causes}: 
\begin{enumerate}
    \item {\textbf{Uniform:} $v_{i,j}^{\ell}$ is drawn uniformly at random from $[-1,1]$.}
    \item {\textbf{Gaussian:} We first draw a mean $\mu_{i,j}$ uniformly at random from $[-1,1]$. Then, $v_{i,j}^{\ell}$ is drawn from the Gaussian distribution with mean $\mu_{i,j}$ and standard deviation $1-\mu_{i,j}$ if $\mu_{i,j}\geq 0$ and $|-1-\mu_{i,j}|$ if $\mu_{i,j}< 0$, ensuring that $v_{i,j}^{\ell} \in [-1,1]$.}\label{model:gaussian}
    \item {\textbf{Size-Dependent Uniform:} We first draw $u_{i,j}^{\ell}$ uniformly at random from $[-1,1]$ and then set $v_{i,j}^{\ell} = \frac{|C_\ell^{\mathbf{a}}|}{n+1}\cdot u_{i,j}^{\ell}$.}
    \item {\textbf{Size-Dependent Gaussian:} We first draw $u_{i,j}^{\ell}$ from the Gaussian distribution constructed in the Gaussian model in \ref{model:gaussian} and then set $v_{i,j}^{\ell} = \frac{|C_\ell^{\mathbf{a}}|}{n+1}\cdot u_{i,j}^{\ell}$.}
\end{enumerate}

{For each configuration of parameters, we consider two exploration policies for constructing an offline dataset of size $M$:}
\begin{enumerate}[label={(\arabic*)}]
    \item {A \textit{\textbf{uniformly random}} policy $\rho^{\text{rand}}$, where joint actions are sampled uniformly at random (i.e., $\rho^{\text{rand}}(\mathbf{a}) = \frac{1}{|\mathcal{A}|}$ for any joint action $\mathbf{a} \in \mathcal{A}$). $\rho^{\text{rand}}$ satisfies $c_{\size}^{\bm{\varphi}^\star}=|\mathcal{A}|^n$ by equation (7) in the main text and covers all coalition sizes, thus satisfying Assumption 1.}
    \item {To validate the need of Assumption 1, we also consider a policy $\rho^{\text{1Rand}}$ that does not necessarily satisfy it, which induces a uniformly random strategy $\rho_1^{\text{1Rand}}$ for agent $1$ (i.e., $\rho_1^{\text{1Rand}}(a_1)=\frac{1}{|\mathcal{A}_1|}$ for any $a_1 \in \mathcal{A}_1$); others always deterministically follow the second action that was inserted to their action set during its random generation.}
\end{enumerate}

{Algorithm 1 under semi-bandit feedback is then implemented using the exploration bonuses in equation (9) of the main text with a confidence level of $\delta = 10^{-2}$. As exactly solving the optimization problem (6) in Algorithm 1 is computationally infeasible (see Footnote 3), especially for large joint action spaces, we employ a practical \textit{coordinate-descent} update scheme, known to converge to stationary points even in some non-convex settings and is widely used due to high efficiency in large-scale optimization problems (see, e.g., \cite{nesterov2012efficiency,wright2015coordinate}). It approximates the global surrogate minimization via iterative smoothing toward per-agent optimistic best responses. Formally, at each round, we sequentially update each agent's mixed strategy by taking a convex combination between her current strategy and an optimistic best response to the estimated empirical utilities of the other agents, yielding a smoothed best-response or fictitious-play-like dynamics with gradual convergence. Here, expectations of utility estimators and exploration bonuses are computed by Monte Carlo sampling with $100$ samples per term.}

\subsubsection{Empirical Results}
\label{supp:results}
Figure \ref{supp:fig:experiments} reports the mean approximate duality gap $\widehat{\gap}^\delta(\bm{\varphi}^{\out})$ of the strategy $\bm{\varphi}^{\out}$ produced by Algorithm 1 versus the size of datasets generated by $\rho^{\text{rand}}$ over $5$ runs with different seeds. Further, Figure \ref{fig:experiments2} reports the mean approximate duality gap $\widehat{\gap}^\delta(\bm{\varphi}^{\out})$ of the strategy $\bm{\varphi}^{\out}$ produced by Algorithm 1 versus the size of datasets generated by $\rho^{\text{1Rand}}$ over $5$ runs with different seeds. By Lemma \ref{supp:lemma:surrogate}, $\widehat{\gap}^\delta(\bm{\varphi}^{\out})$ upper bounds the \textit{true} duality gap with high probability, thus quantifying how close the learned strategy is to Nash stability. The top two rows examine the effect of varying the number of agents $n \in \{5,10,15,20,25\}$ while fixing $k=5$, whereas the bottom two rows vary the number of candidate coalitions $k \in \{5,10,15,20,25\}$ while fixing $n=10$. In each experiment, we evaluate our algorithm on datasets of sizes $M \in \{10^2, 5\cdot 10^3, 10^4, 2\cdot 10^4, 3\cdot 10^4\}$.

By Figure \ref{supp:fig:experiments}, Algorithm 1 consistently reaches a low approximation to Nash stability under $\rho^{\text{rand}}$, but fails to do so under $\rho^{\text{1Rand}}$ as noted in Figure \ref{fig:experiments2}. This supports the practical relevance of Assumption 1 in the main text: $\rho^{\text{rand}}$ satisfies it, allowing effective learning, whereas $\rho^{\text{1Rand}}$ may violate it, yielding poorer performance. For $\rho^{\text{rand}}$, the scaling in $n,k,M$ also matches Theorem \ref{supp:thm:semi-bandit}, Corollary~\ref{supp:coro:optimal} and Remark 5 in the main text. For any fixed number of agents $n$, the gap decreases rapidly when the dataset size $M$ is small, but slowly for larger $M$, consistent with the $1/\sqrt{M}$ factor in Theorem \ref{supp:thm:semi-bandit}. Further, obtaining a certain approximation requires larger datasets as $n$ increases, in line with Corollary~\ref{supp:coro:optimal} and Remark 5 in the main text.

\begin{figure}[hp!]
    \begin{tabular}{cc}
        \begin{subfigure}[b]{0.5\textwidth}
            \centering
            \includegraphics[width=\linewidth]{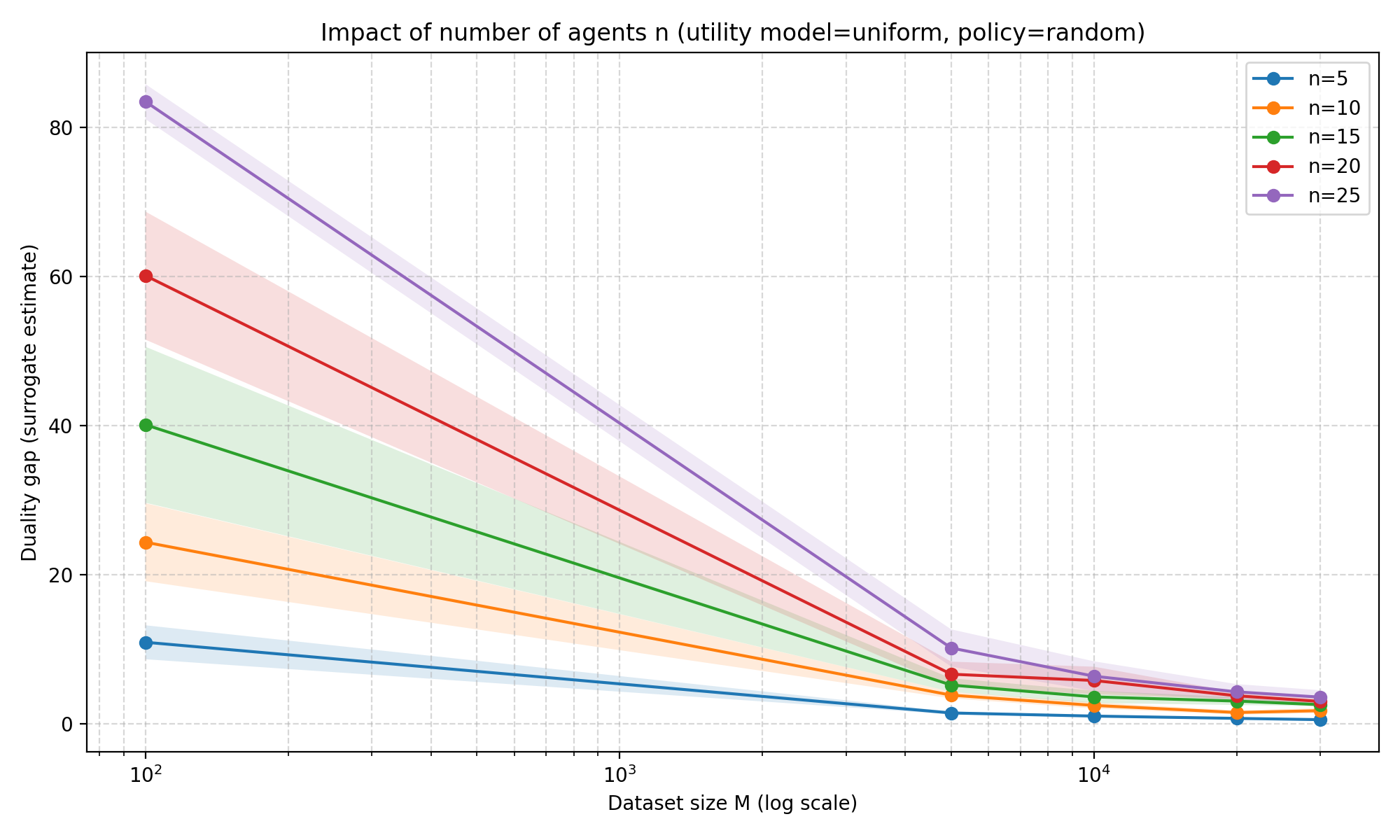}
        \end{subfigure}
        \begin{subfigure}[b]{0.5\textwidth}
            \centering
            \includegraphics[width=\linewidth]{figures/impact_n_size_uniform_random.png}
        \end{subfigure}
        \\
        \begin{subfigure}[b]{0.5\textwidth}
            \centering
            \includegraphics[width=\linewidth]{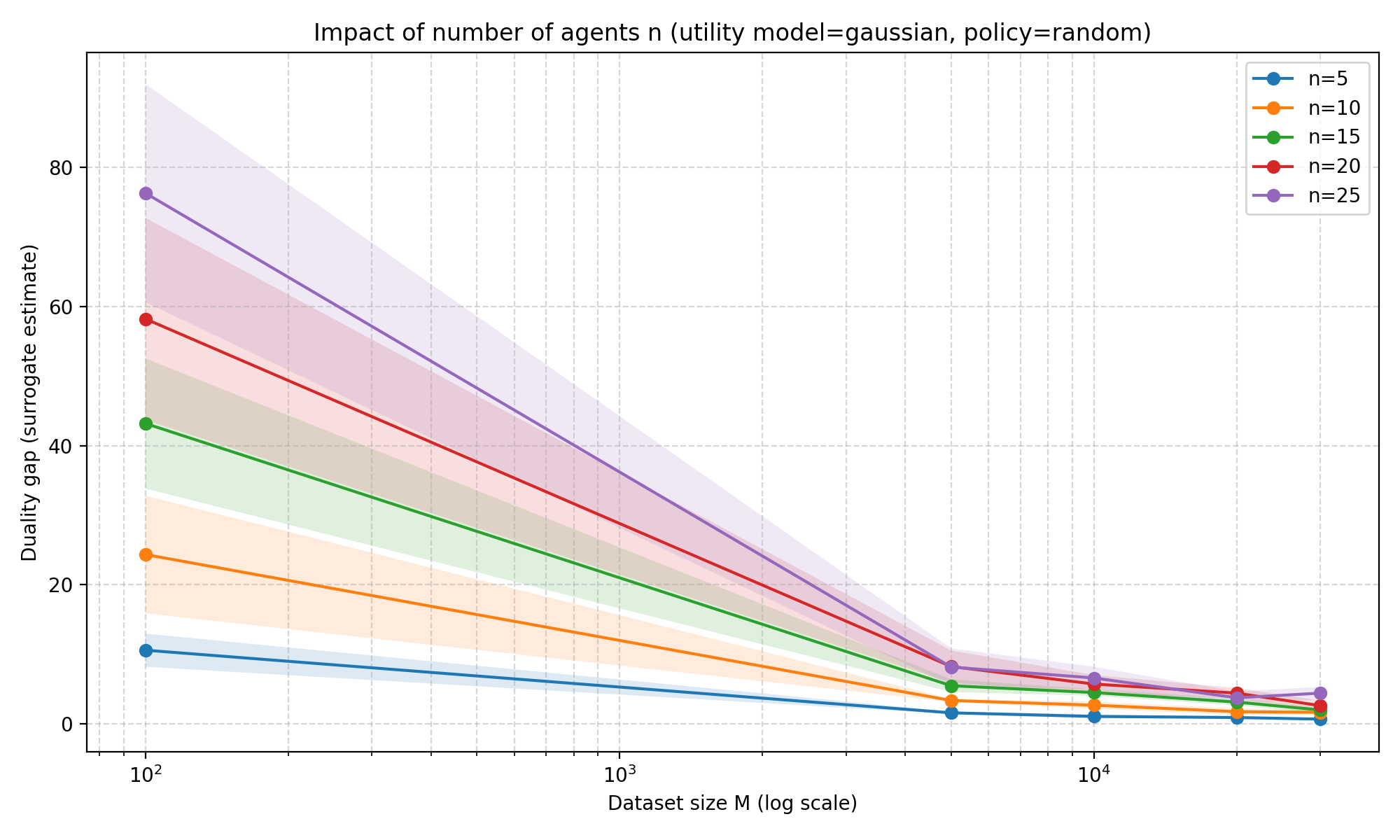}
        \end{subfigure}
        \begin{subfigure}[b]{0.5\textwidth}
            \centering
            \includegraphics[width=\linewidth]{figures/impact_n_size_gaussian_random.png}
        \end{subfigure}
        \\
        \begin{subfigure}[b]{0.5\textwidth}
            \centering
            \includegraphics[width=\linewidth]{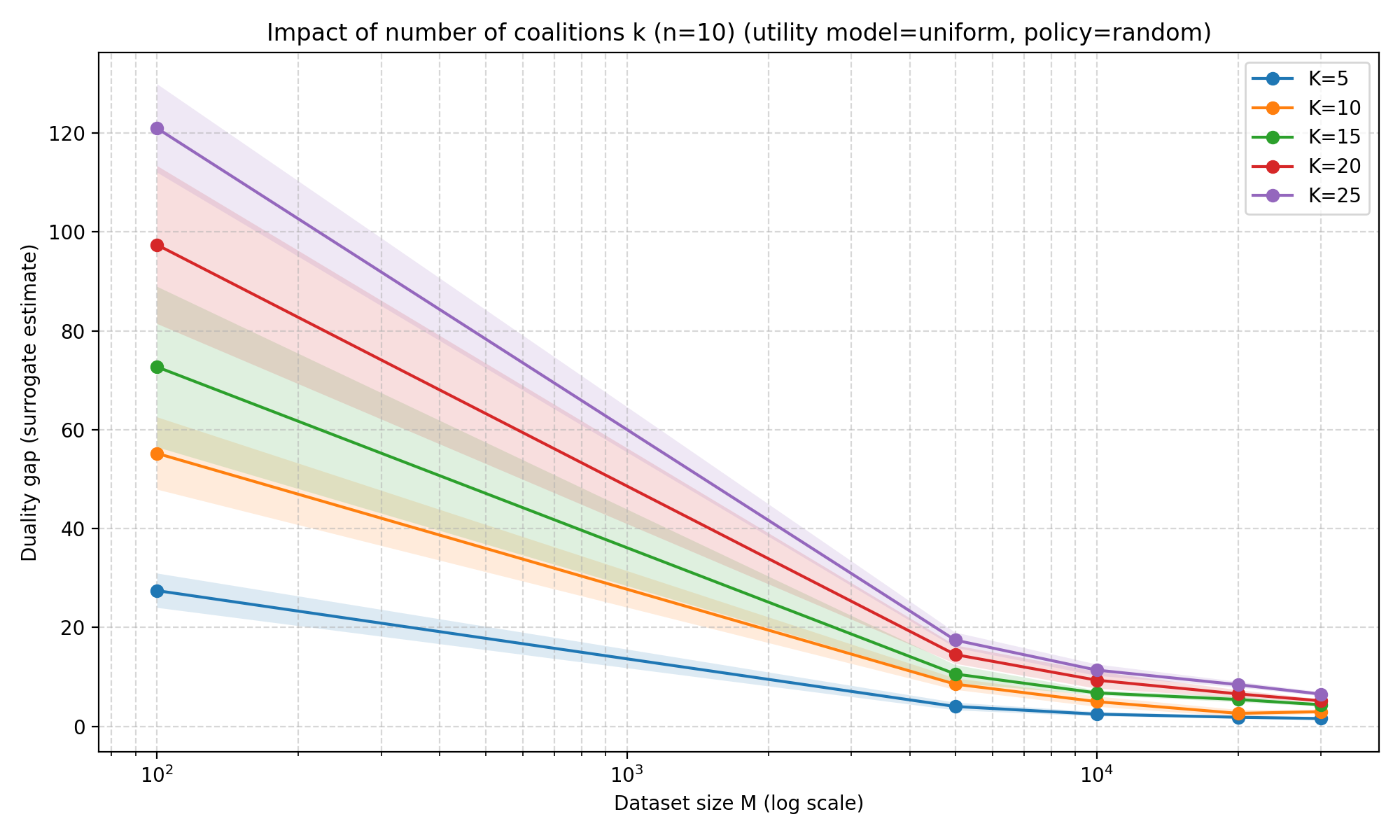}
        \end{subfigure}
        \begin{subfigure}[b]{0.5\textwidth}
            \centering
            \includegraphics[width=\linewidth]{figures/impact_K_n10_size_uniform_random.png}
        \end{subfigure}
        \\
        \begin{subfigure}[b]{0.5\textwidth}
            \centering
            \includegraphics[width=\linewidth]{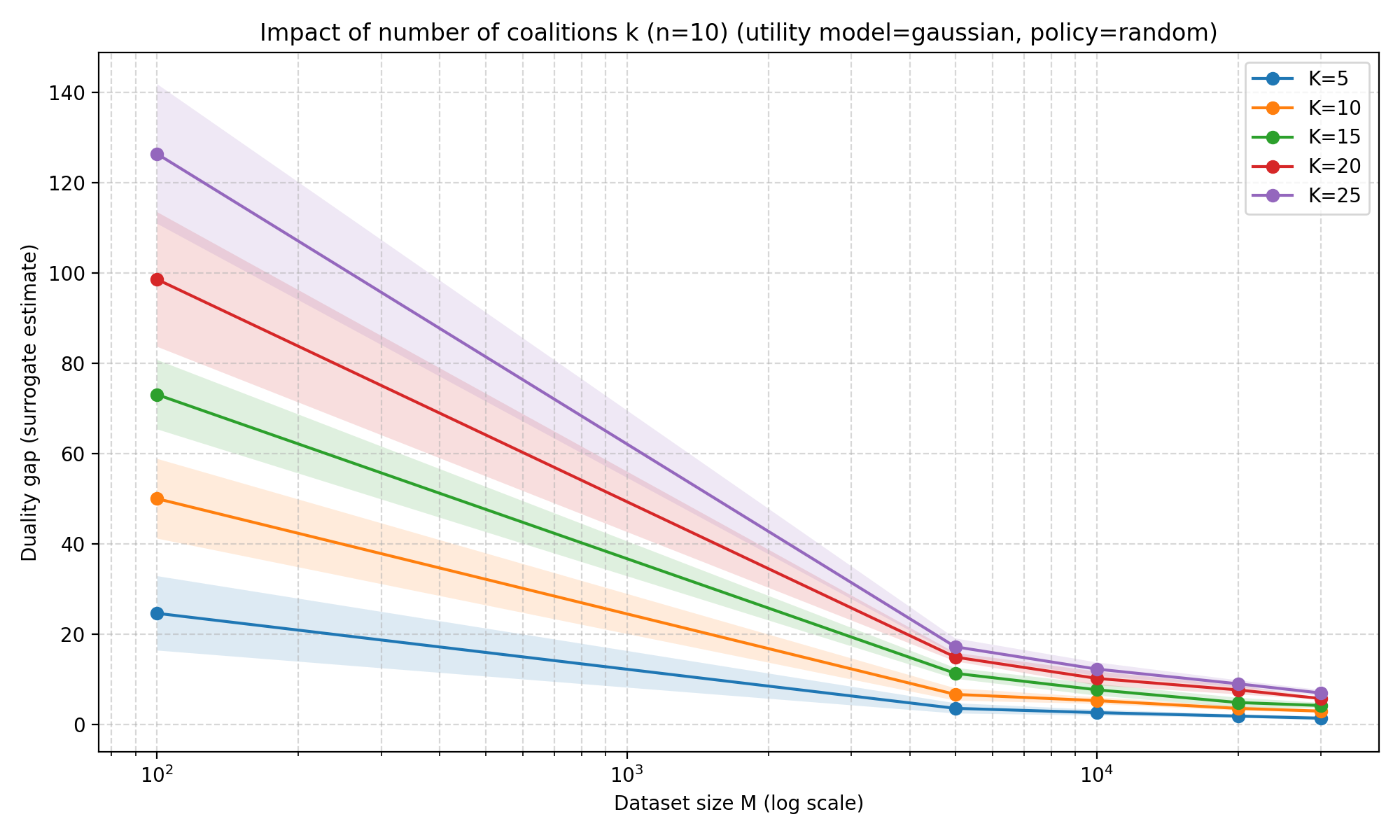}
        \end{subfigure}
        \begin{subfigure}[b]{0.5\textwidth}
            \centering
            \includegraphics[width=\linewidth]{figures/impact_K_n10_size_gaussian_random.png}
        \end{subfigure}
    \end{tabular}
    \caption{{Mean approximate duality gap versus the size of datasets generated by $\rho^{\text{rand}}$ over $5$ runs with different seeds, for varying numbers of agents (top two rows) and candidate coalitions (bottom two rows). Shaded regions indicate standard deviations.}}
    \label{supp:fig:experiments}
\end{figure} 

\begin{figure}[hp!]
    \begin{tabular}{cc}
        \begin{subfigure}[b]{0.5\textwidth}
            \centering
            \includegraphics[width=\linewidth]{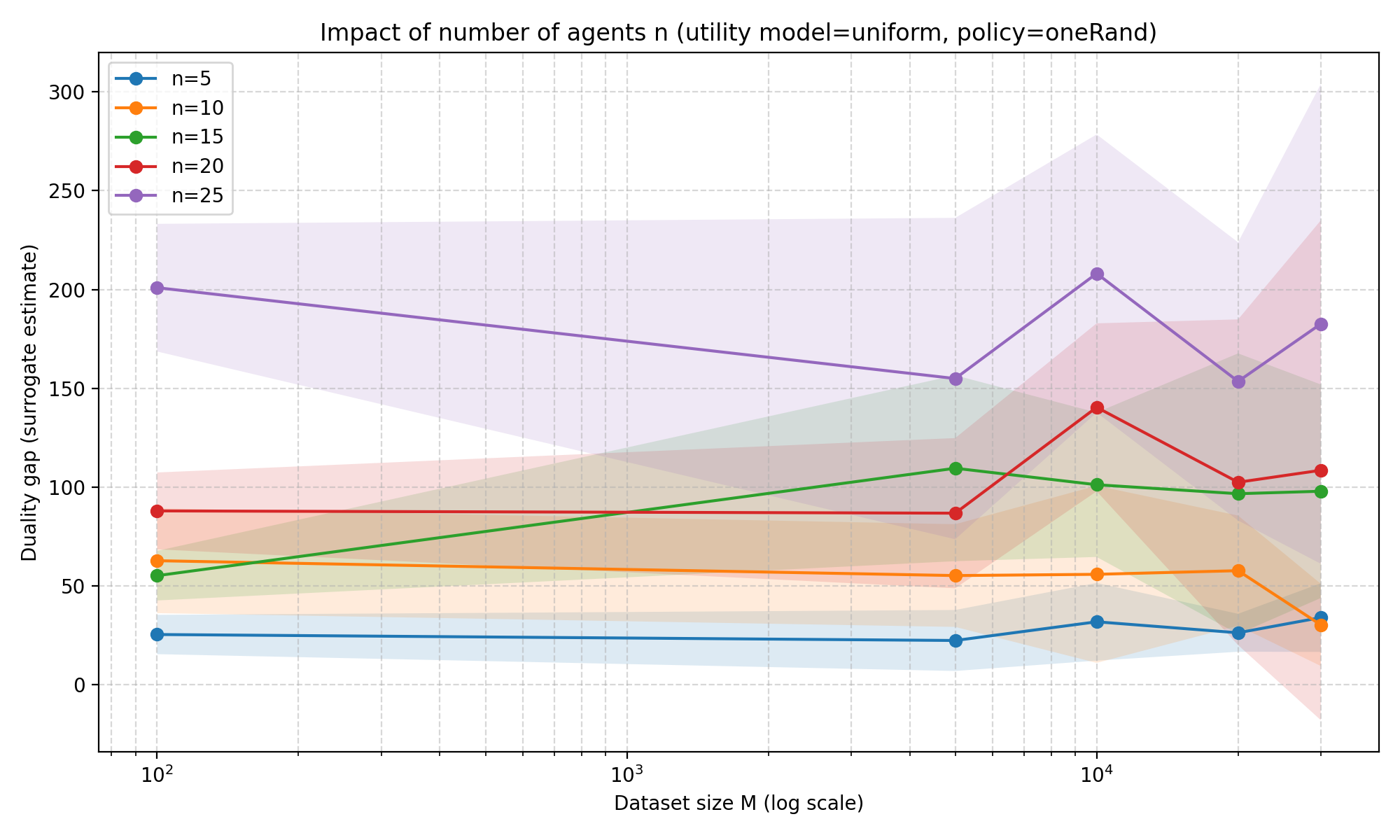}
        \end{subfigure}
        \begin{subfigure}[b]{0.5\textwidth}
            \centering
            \includegraphics[width=\linewidth]{figures/impact_n_size_uniform_oneRand.png}
        \end{subfigure}
        \\
        \begin{subfigure}[b]{0.5\textwidth}
            \centering
            \includegraphics[width=\linewidth]{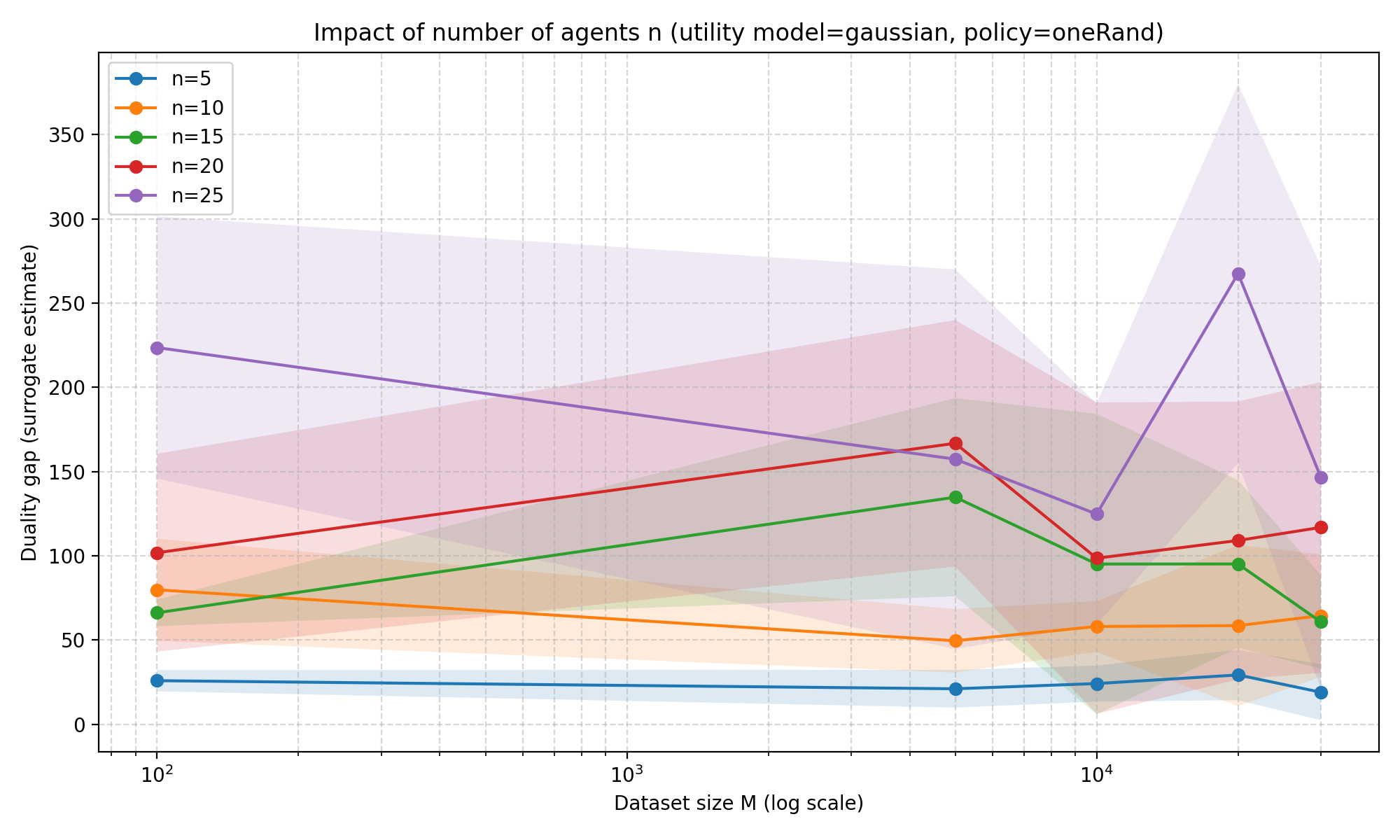}
        \end{subfigure}
        \begin{subfigure}[b]{0.5\textwidth}
            \centering
            \includegraphics[width=\linewidth]{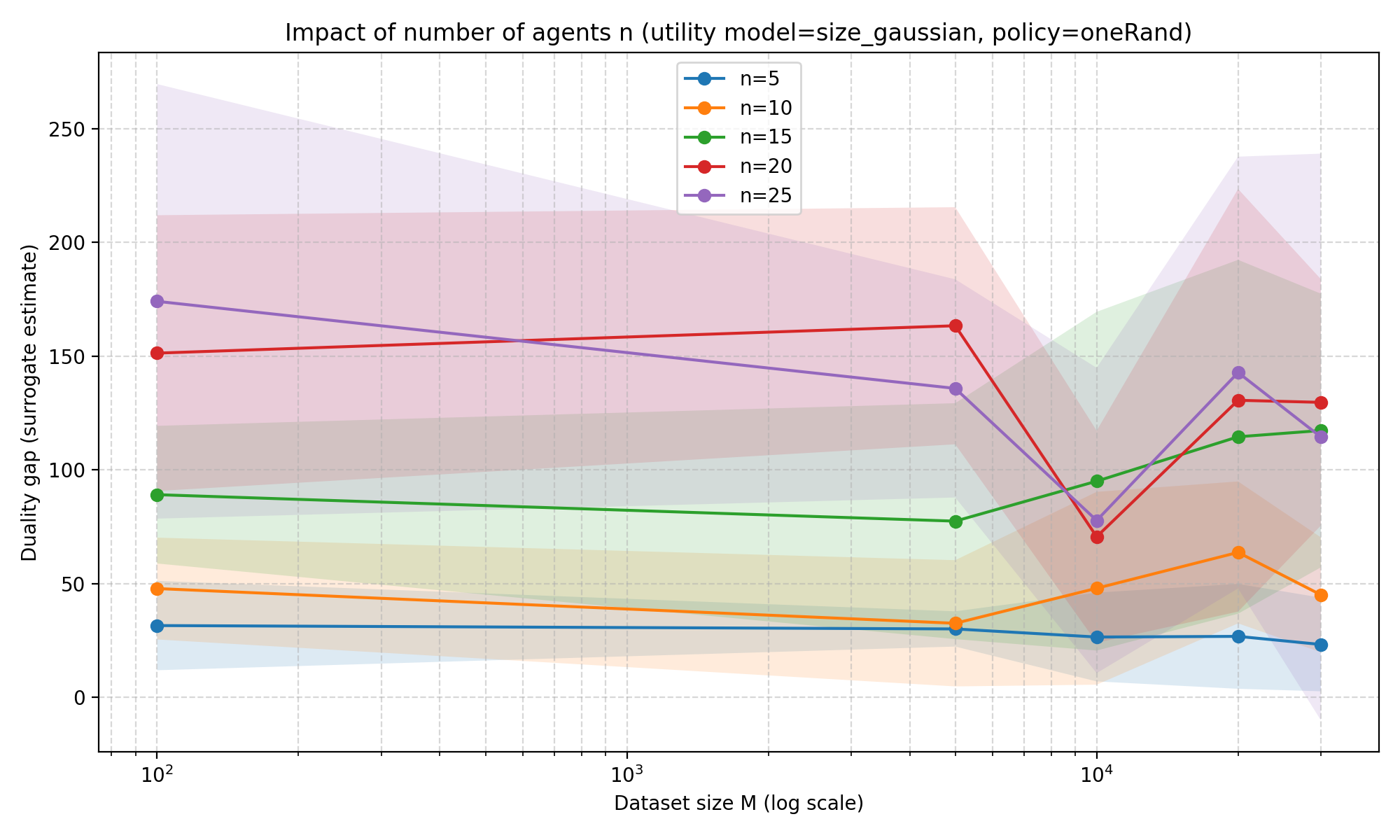}
        \end{subfigure}
        \\
        \begin{subfigure}[b]{0.5\textwidth}
            \centering
            \includegraphics[width=\linewidth]{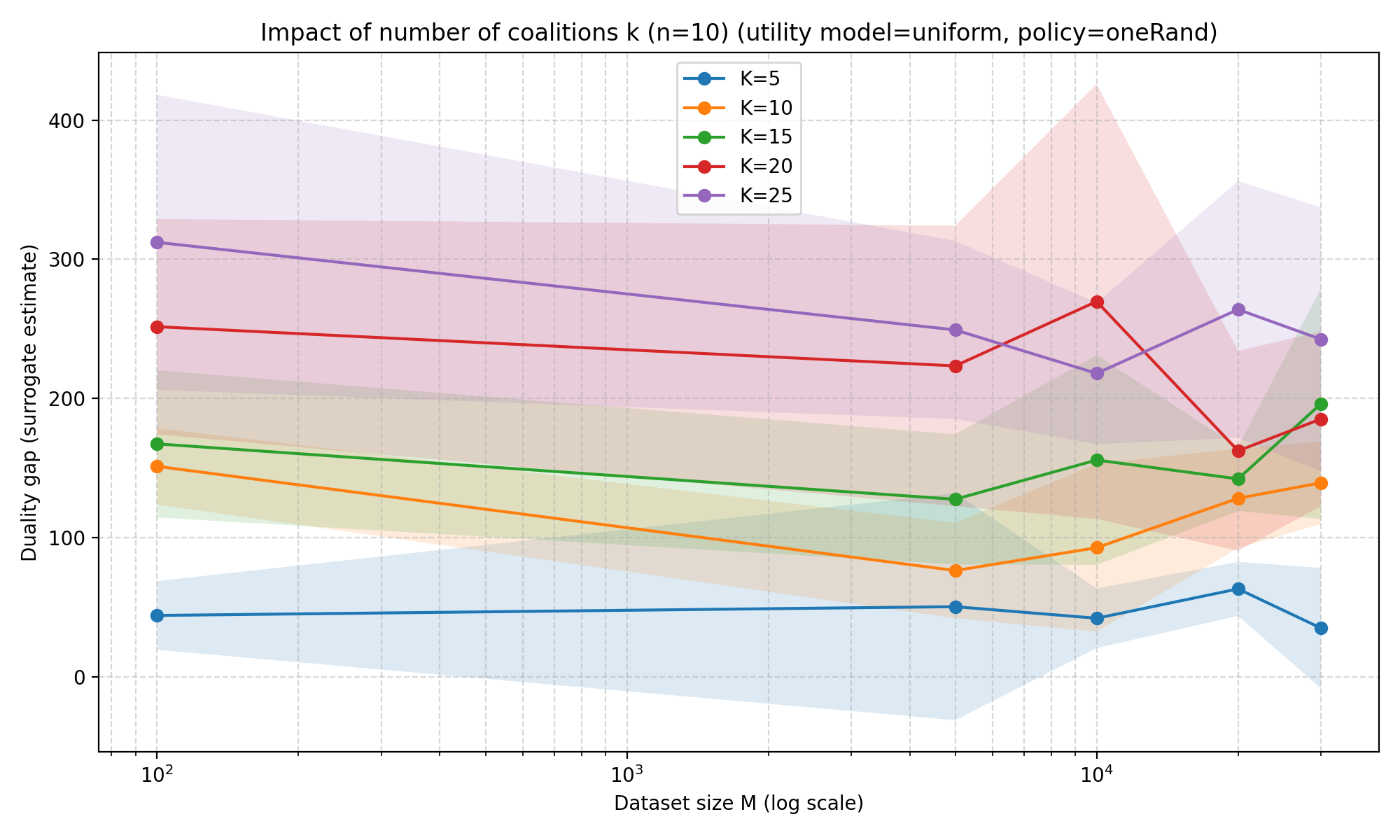}
        \end{subfigure}
        \begin{subfigure}[b]{0.5\textwidth}
            \centering
            \includegraphics[width=\linewidth]{figures/impact_K_n10_size_uniform_oneRand.png}
        \end{subfigure}
        \\
        \begin{subfigure}[b]{0.5\textwidth}
            \centering
            \includegraphics[width=\linewidth]{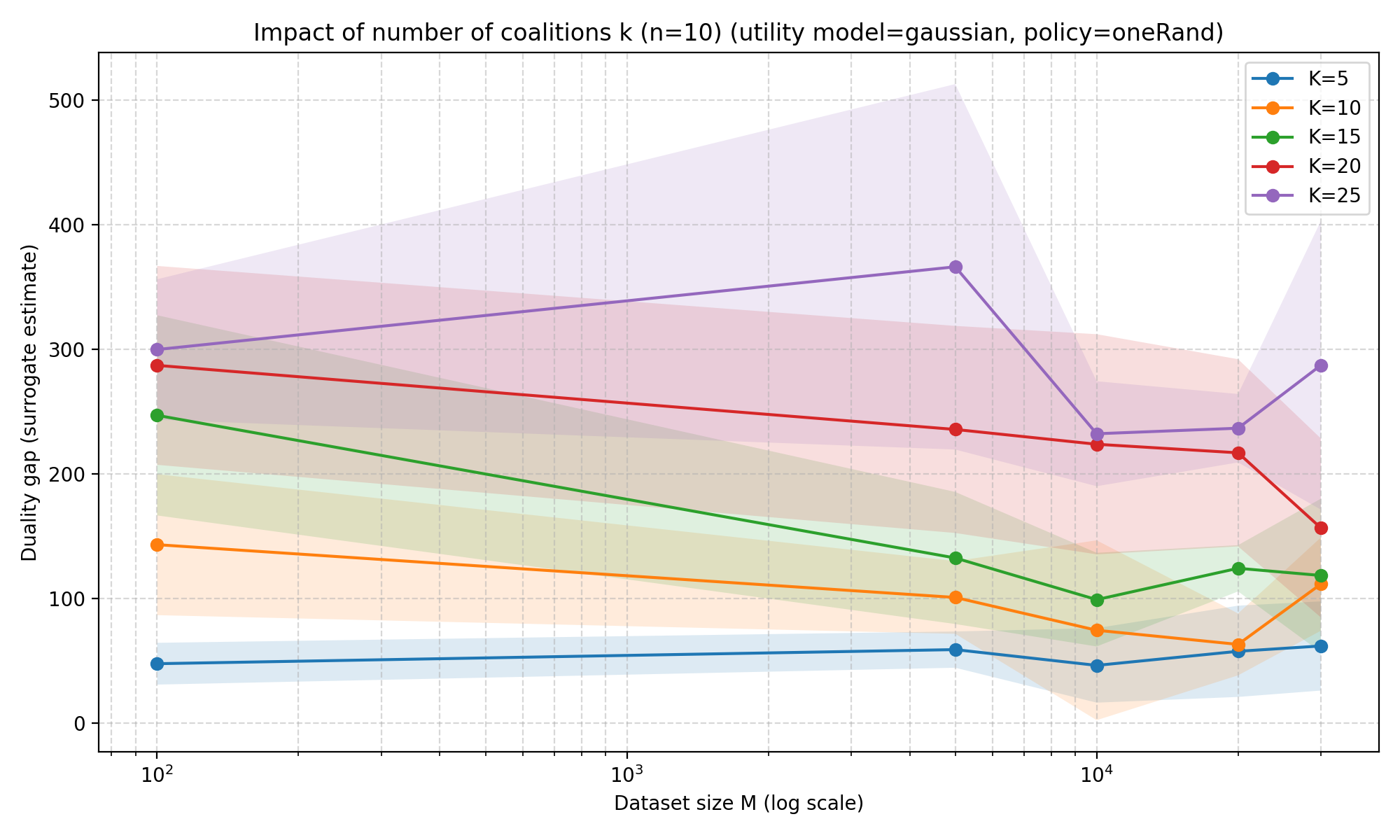}
        \end{subfigure}
        \begin{subfigure}[b]{0.5\textwidth}
            \centering
            \includegraphics[width=\linewidth]{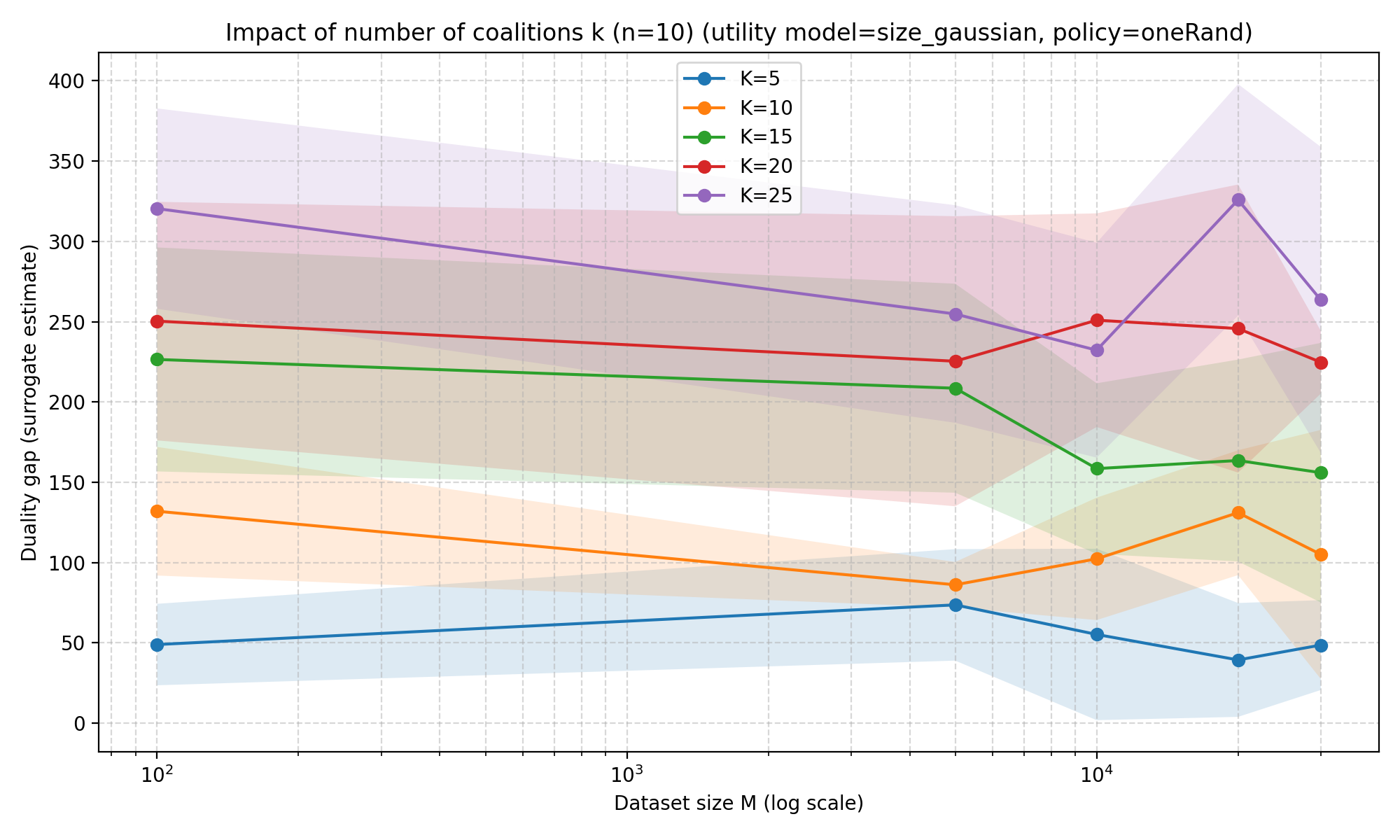}
        \end{subfigure}
    \end{tabular}
    \caption{{Mean approximate duality gap versus the size of datasets generated by$\rho^{\text{1Rand}}$ over $5$ runs with different seeds, for varying numbers of agents (top two rows) and candidate coalitions (bottom two rows). Shaded regions indicate standard deviations.}}
    \label{fig:experiments2}
\end{figure} 

\newpage
\subsubsection{Mixed Coalition Size Effects}
\label{supp:mixed effects}
We also construct stylized games where candidate coalitions differ in how their sizes affect agents’ utilities. There are $k=5$ candidate coalitions and the action set of each agent $i$ is $\mathcal{A}_i=\{\{1,2\}, \{1,3,5\}, \{4,5\}\}$. Given a joint action $\mathbf{a}$, five i.i.d. Gaussian random variables $\mu_1,\mu_2,\mu_3,\mu_4,\mu_5\sim \mathcal{N}(0,1)$ and a pair of agents $i,j$ with $a_i=a_j=\ell$, we generate the mutual utility $v_{i,j}^\ell$ based on the value of $\ell$:
\begin{enumerate}
    \item \underline{$\ell=1$:} Coalition $1$ penalizes overcrowding, i.e., $v_{i,j}^1=\mu_1-\frac{|C_1^{\mathbf{a}}|}{n+1}$

    \item \underline{$\ell=5$:} Coalition $5$ rewards coordination, i.e., $v_{i,j}^5=\mu_5+\frac{|C_1^{\mathbf{a}}|}{n+1}$.
    
    \item \underline{$\ell \in\{2,4\}$:} Coalitions $2$ and $4$ represent persistently costly environments, i.e., $v_{i,j}^\ell=\mu_\ell-1$ for $\ell \in\{2,4\}$.

    \item \underline{$\ell=3$:} Coalition 3 serves as a neutral stochastic baseline, i.e., $v_{i,j}^3=\mu_3$.
\end{enumerate}
Thus, the pure Nash-stable outcome is $\{1,3,5\}^n$.

This design reflects common trade-offs in coalition formation: agents must decide whether to join small, high-quality groups to avoid overcrowding, or larger, more cooperative ones that benefit from coordination, while avoiding consistently poor environments. It thereby provides a controlled and interpretable testbed for evaluating our algorithm. 

In fact, for this set of experiments with $n\geq 2$, both the uniformly random exploration policy $\rho^{\text{rand}}$ and the exploration policy $\rho^{\text{1Rand}}$ satisfy Assumption 1. Indeed,  the exploration policy $\rho^{\text{1Rand}}$ now satisfies Assumption 1 because it covers all coalition sizes resulting from unilateral deviations from the pure NS strategy $\{1,3,5\}^n$. Further, since each agent's action set is of size exactly $3$, the coalition size coefficient of $\rho^{\text{rand}}$ is $c_{\size}^{\bm{\varphi}^\star}=3^n$, while the coalition size coefficient of $\rho^{\text{rand}}$ is $c_{\size}^{\bm{\varphi}^\star}=3$. Therefore, to clarify the properties of $\rho^{\text{1Rand}}$, we rename it as the \textit{\textbf{coalition size exploration policy}}, and thus hereafter replace the notation $\rho^{\text{1Rand}}$ with $\rho^{\text{coalitionSize}}$.

\paragraph{Empirical Results.}

Figure \ref{fig:mixed effects} reports the mean approximate duality gap $\widehat{\gap}^\delta(\bm{\varphi}^{\out})$ of the strategy $\bm{\varphi}^{\out}$ produced by Algorithm 1 versus the size of datasets generated by $\rho^{\text{rand}}$ over $5$ runs with different seeds. Further, Figure \ref{fig:mixed effects2} reports the mean approximate duality gap $\widehat{\gap}^\delta(\bm{\varphi}^{\out})$ of the strategy $\bm{\varphi}^{\out}$ produced by Algorithm 1 versus the size of datasets generated by $\rho^{\text{coalitionSize}}$ over $5$ runs with different seeds. By Lemma \ref{supp:lemma:surrogate}, $\widehat{\gap}^\delta(\bm{\varphi}^{\out})$ upper bounds the \textit{true} duality gap with high probability, thus quantifying how close the learned strategy is to Nash stability. In both figures, we examine the effect of varying the number of agents $n \in \{3,4,5,6\}$ while fixing $k=5$, while evaluating our algorithm on datasets of sizes $M \in \{100, 2000, 4000, 6000, 8000, 10000, 20000, 30000\}$ for each fixed number of agents.

Unlike Figures \ref{supp:fig:experiments}--\ref{fig:experiments2}, Algorithm 1 consistently reaches a low approximation to Nash stability under both $\rho^{\text{rand}}$ and $\rho^{\text{coalitionSize}}$, opposed to our prior experiment where our algorithm fails to do so for $\rho^{\text{coalitionSize}}$, as noted in Figure \ref{fig:experiments2}. This supports the practical relevance of Assumption 1 in the main text: both $\rho^{\text{rand}}$ and $\rho^{\text{coalitionSize}}$ satisfy it, allowing effective learning. For both exploration policies, the scaling in $n,k,M$ also matches Theorem \ref{supp:thm:semi-bandit}, Corollary~\ref{supp:coro:optimal} and Remark 5 in the main text, similarly to our empirical results in the main text and Appendices \ref{supp:setup}--\ref{supp:results}. 

\begin{figure}[hp!]
    \begin{tabular}{cc}
        \begin{subfigure}[b]{0.5\textwidth}
            \centering
            \includegraphics[width=\linewidth]{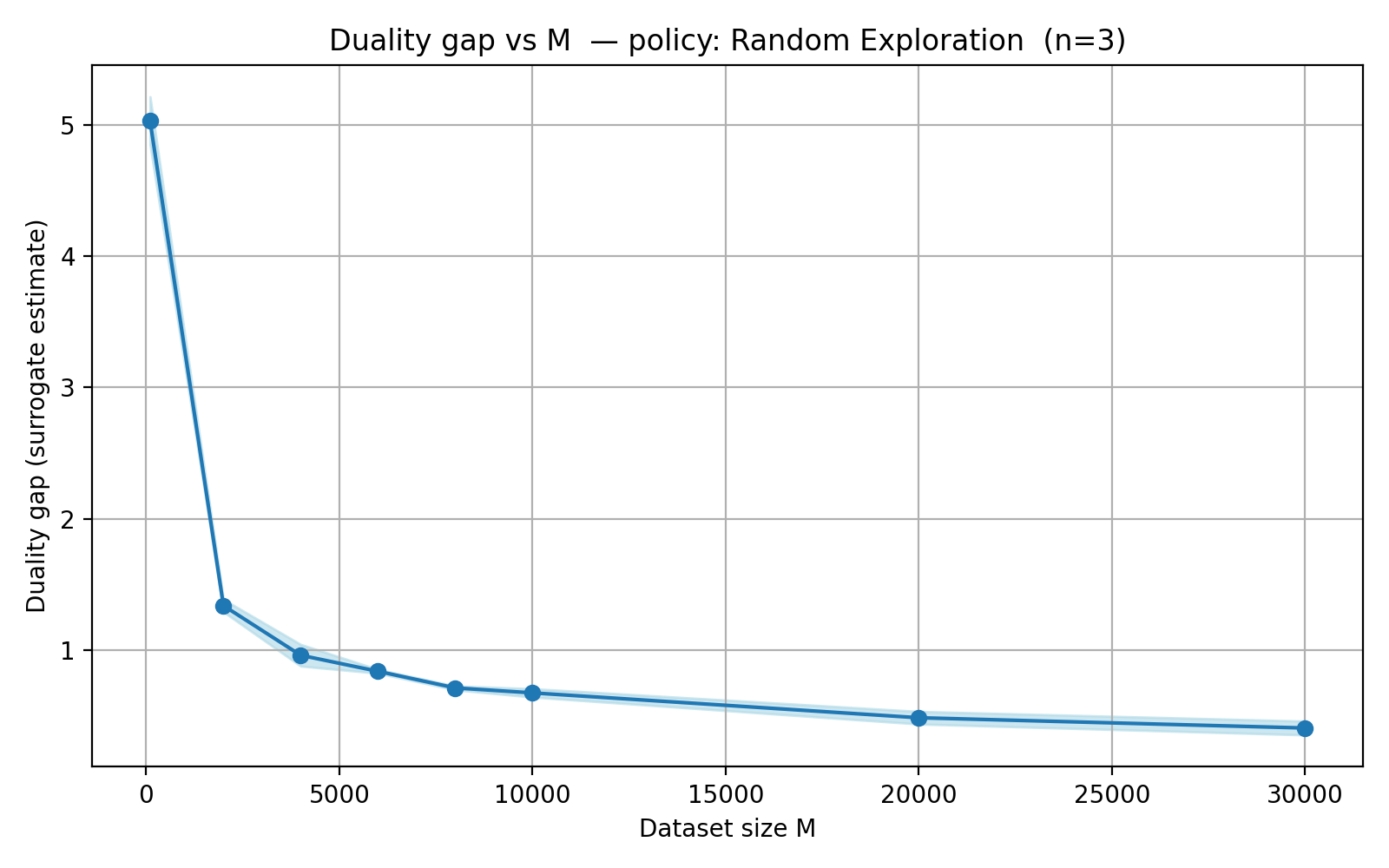}
        \end{subfigure}
        \begin{subfigure}[b]{0.5\textwidth}
            \centering
            \includegraphics[width=\linewidth]{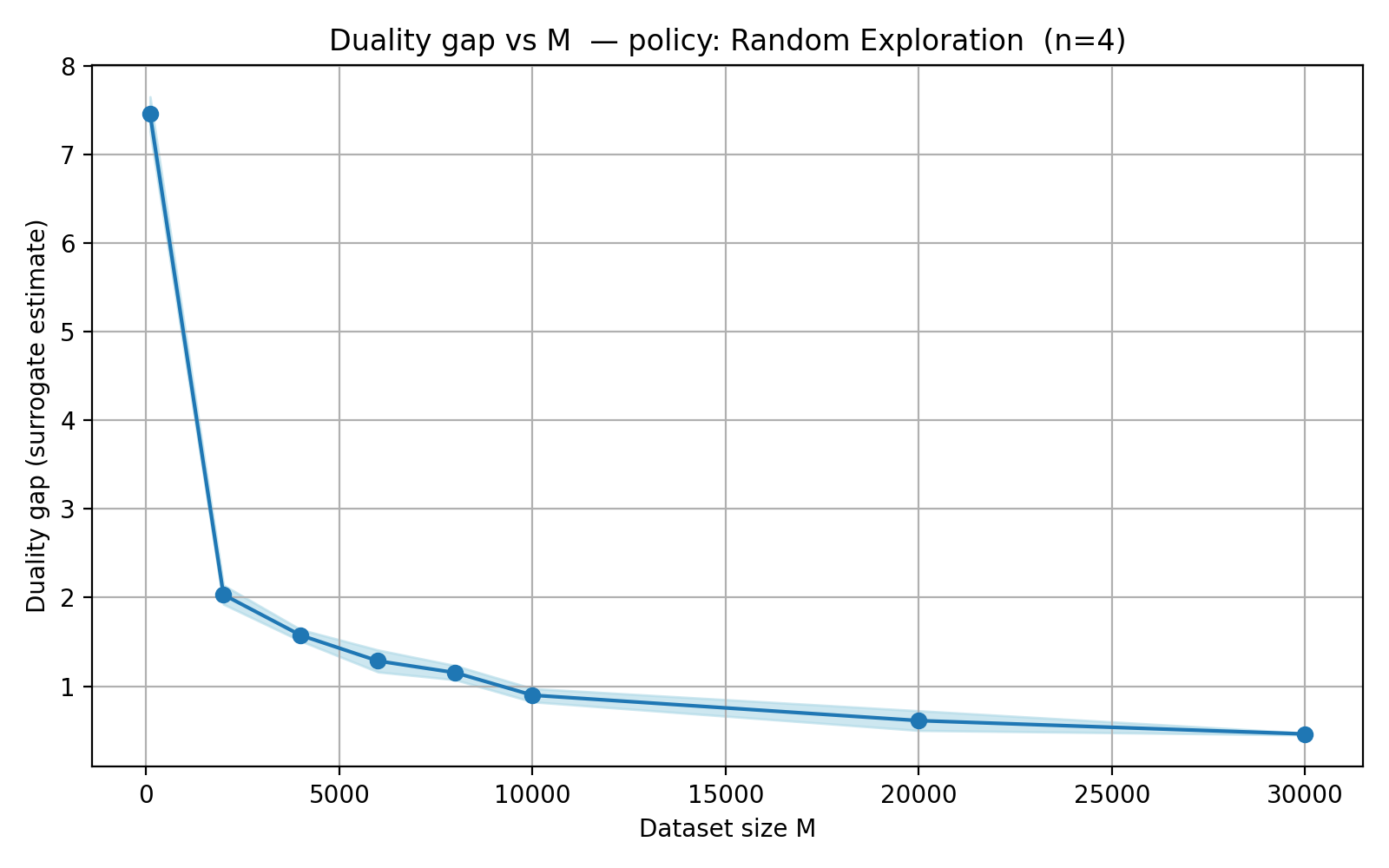}
        \end{subfigure}
        \\
        \begin{subfigure}[b]{0.5\textwidth}
            \centering
            \includegraphics[width=\linewidth]{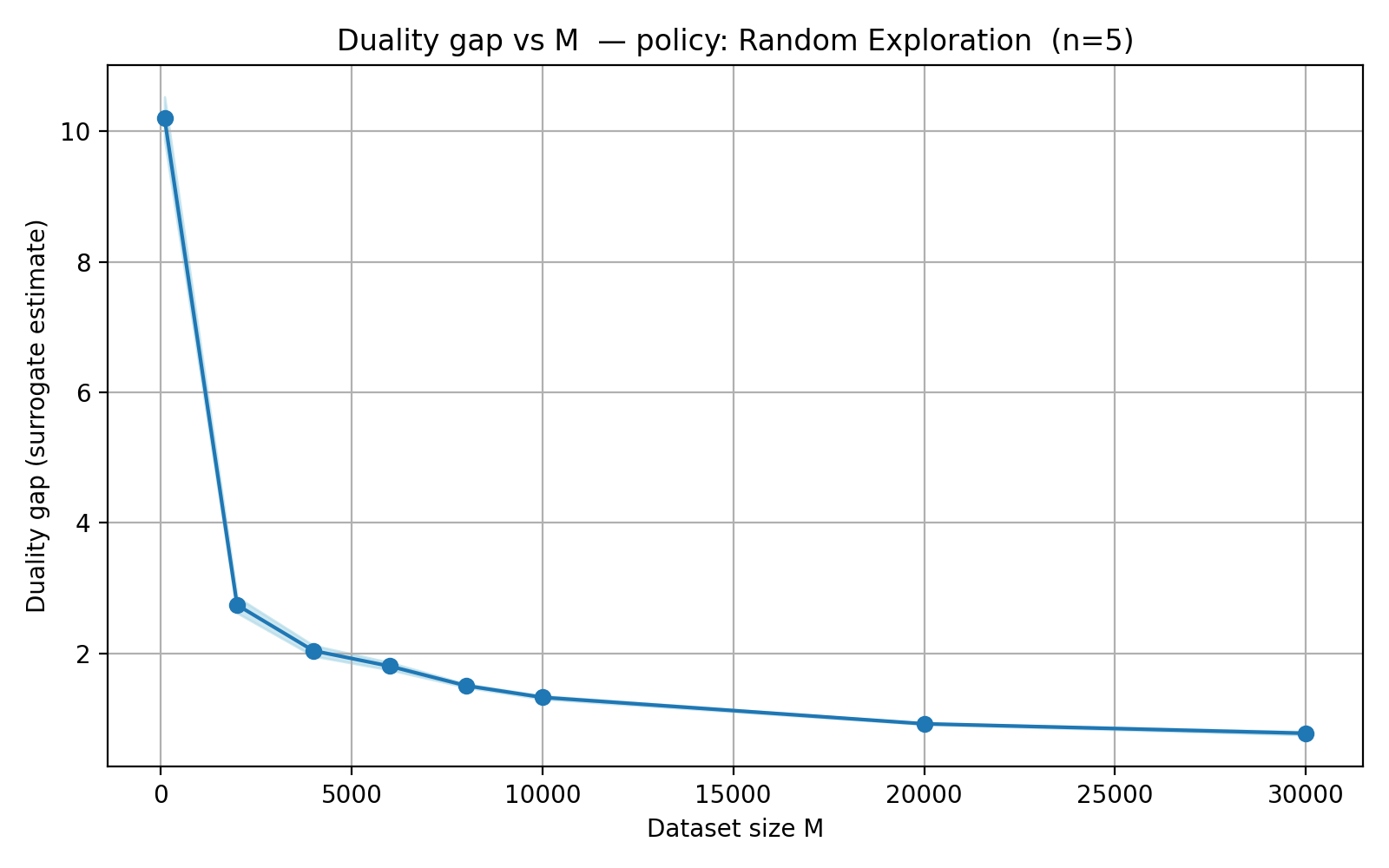}
        \end{subfigure}
        \begin{subfigure}[b]{0.5\textwidth}
            \centering
            \includegraphics[width=\linewidth]{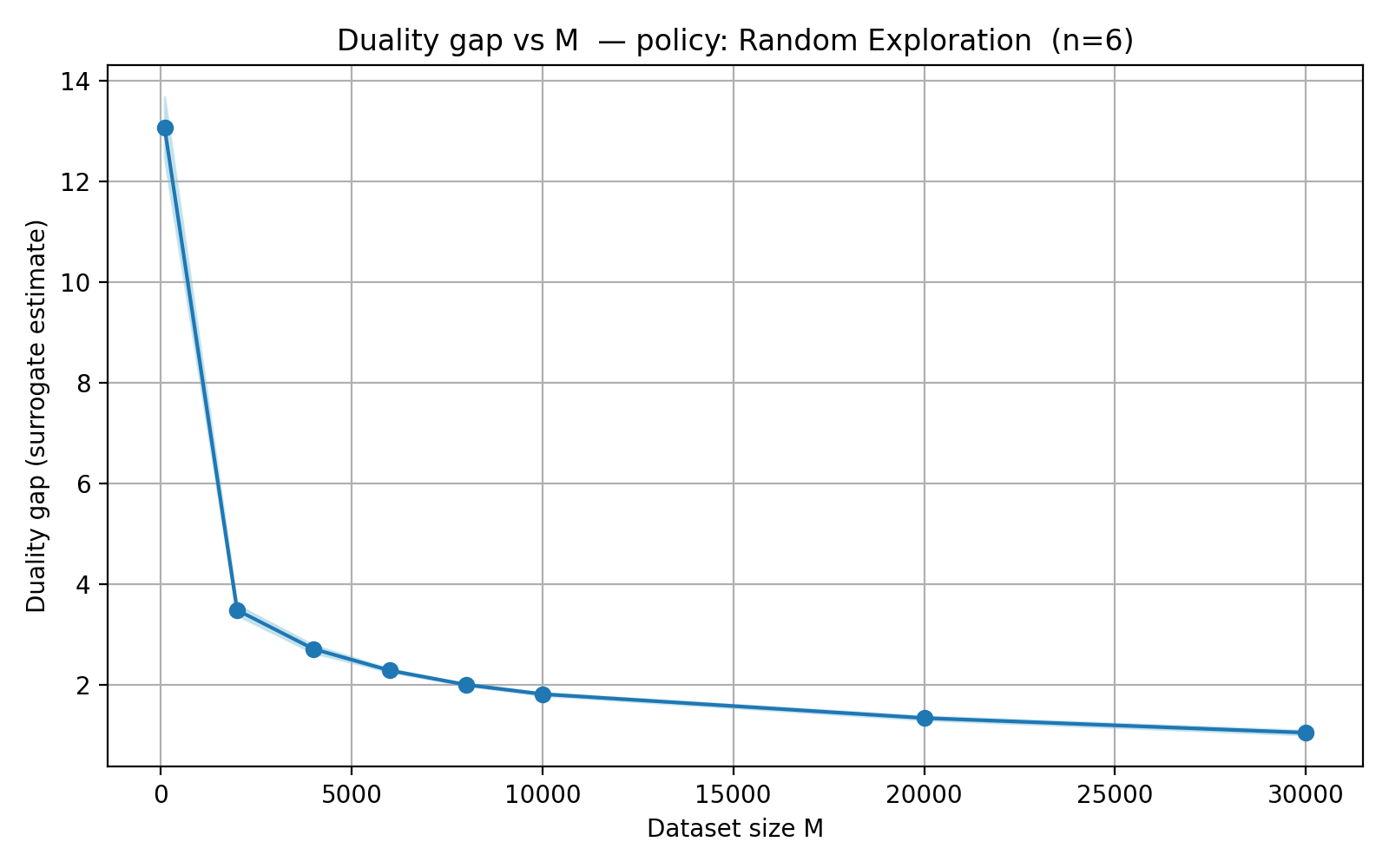}
        \end{subfigure}
    \end{tabular}
    \caption{{Mean approximate duality gap versus the size of datasets generated by $\rho^{\text{rand}}$ over $5$ runs with different seeds for the utility generation model with mixed-coalition-size effects under \textit{semi-bandit} feedback. Here, the number of agents is varied over $n \in \{3,4,5,6\}$. Shaded regions indicate standard deviations.}}
    \label{fig:mixed effects}
\end{figure} 

\begin{figure}[hp!]
    \begin{tabular}{cc}
        \begin{subfigure}[b]{0.5\textwidth}
            \centering
            \includegraphics[width=\linewidth]{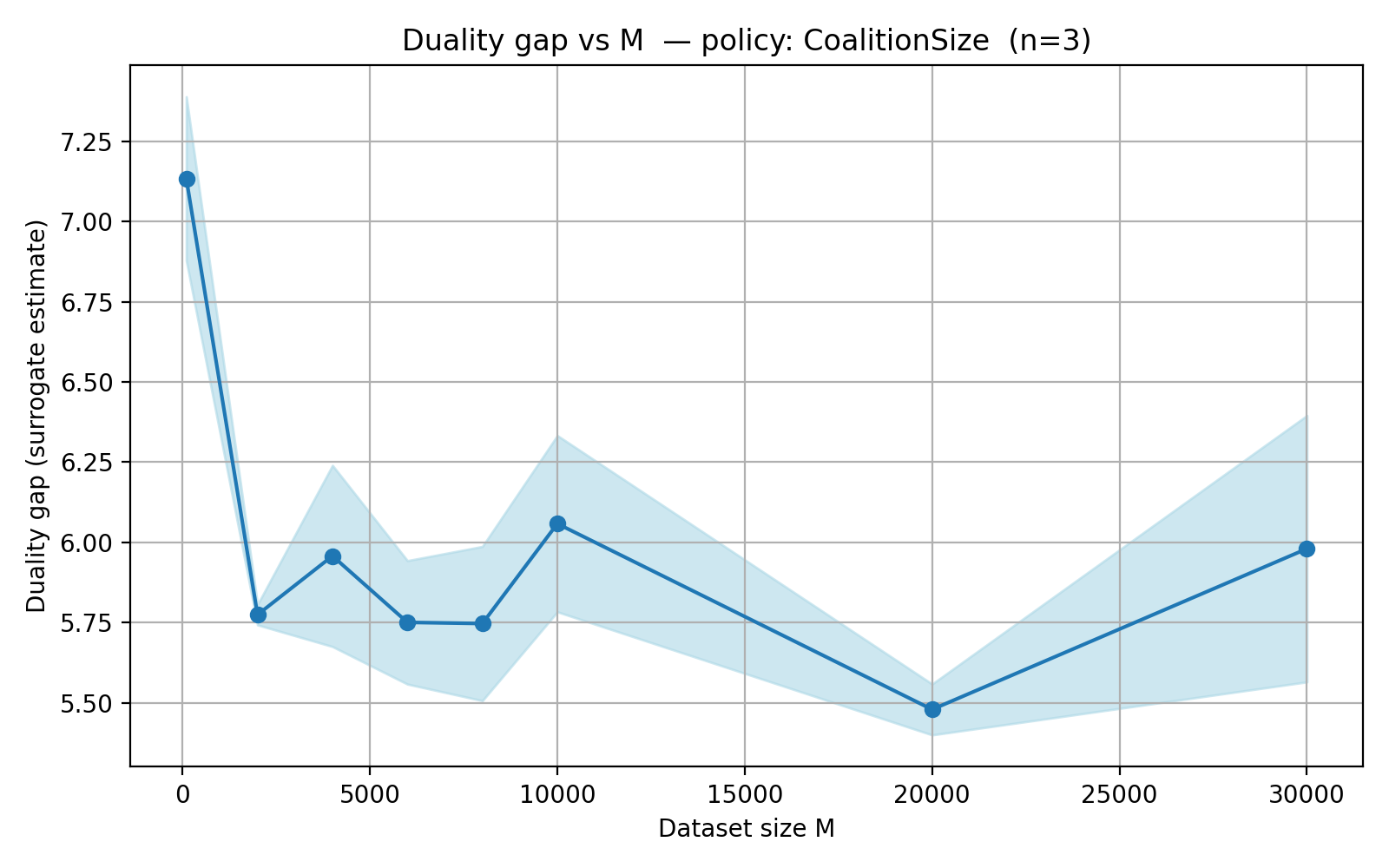}
        \end{subfigure}
        \begin{subfigure}[b]{0.5\textwidth}
            \centering
            \includegraphics[width=\linewidth]{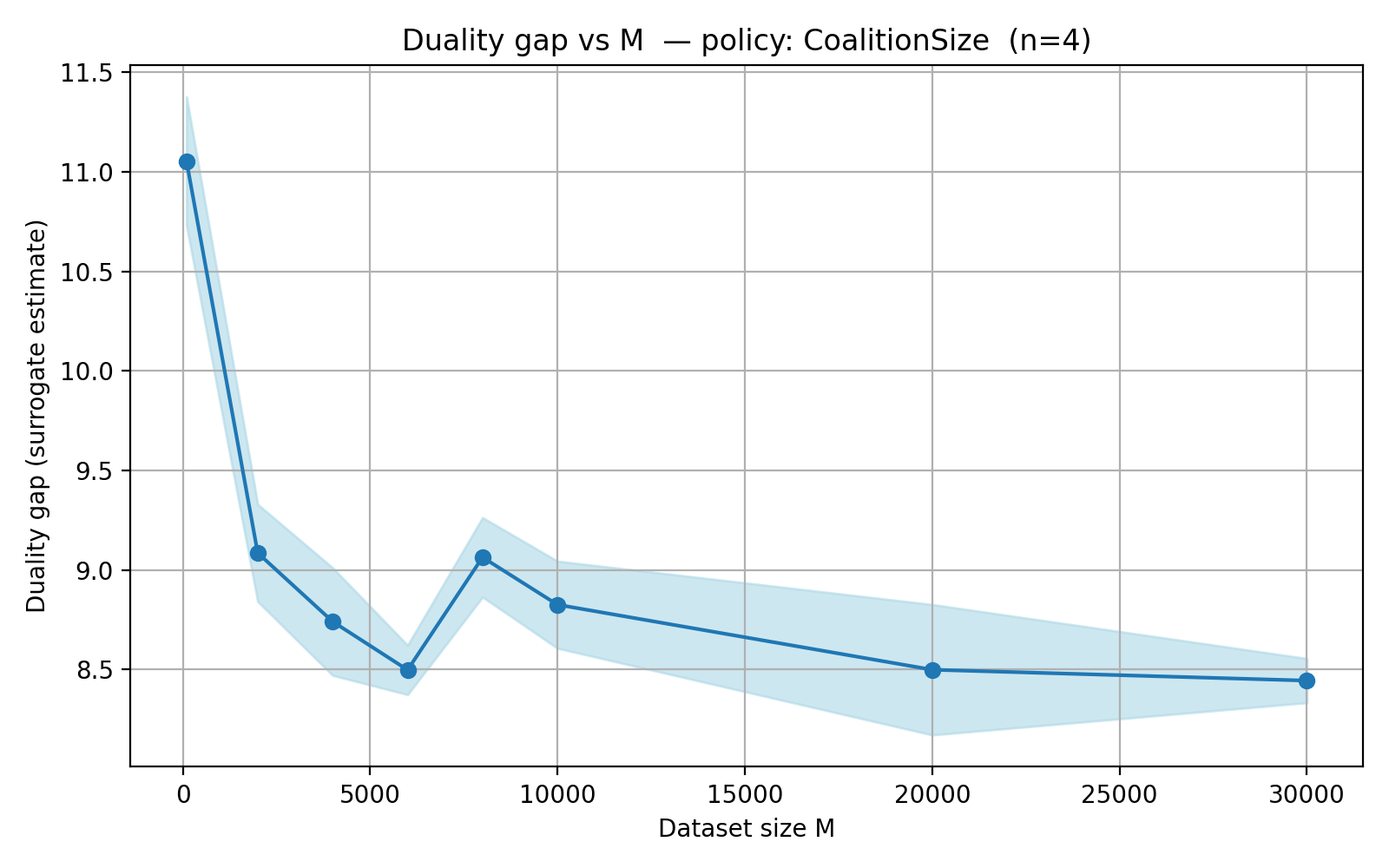}
        \end{subfigure}
        \\
        \begin{subfigure}[b]{0.5\textwidth}
            \centering
            \includegraphics[width=\linewidth]{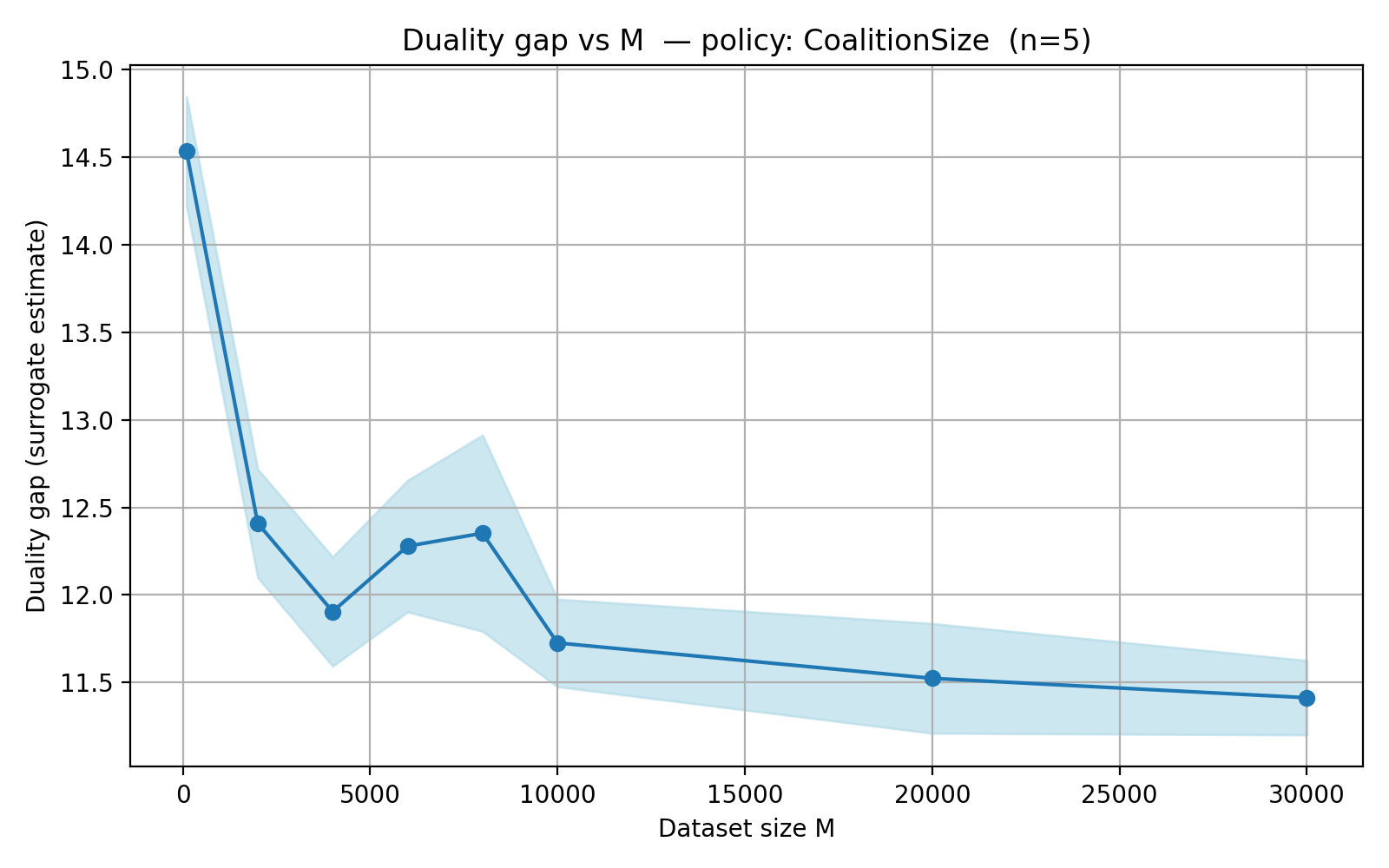}
        \end{subfigure}
        \begin{subfigure}[b]{0.5\textwidth}
            \centering
            \includegraphics[width=\linewidth]{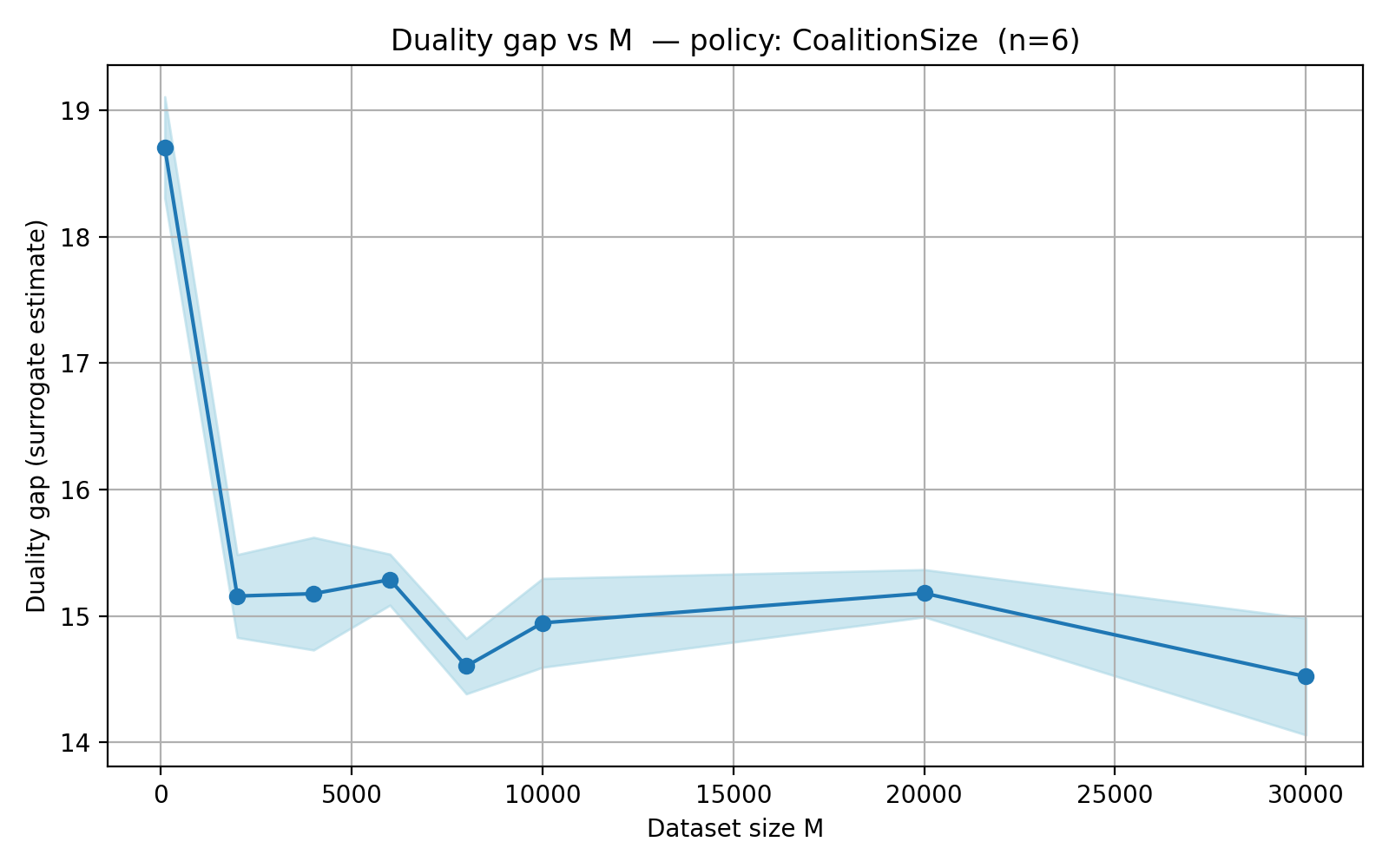}
        \end{subfigure}
    \end{tabular}
    \caption{{Mean approximate duality gap versus the size of datasets generated by $\rho^{\text{coalitionSize}}$ over $5$ runs with different seeds for the utility generation model with mixed-coalition-size effects under \textit{semi-bandit} feedback. Here, the number of agents is varied over $n \in \{3,4,5,6\}$. Shaded regions indicate standard deviations.}}
    \label{fig:mixed effects2}
\end{figure} 

\newpage

\subsection{Empirical Evaluations under Bandit Feedback}
In this section, we supply additional experimental results for bandit feedback. We herein focus on the utility generation model with mixed coalition-size-effects from Appendix \ref{supp:mixed effects}, which exhibits trends similar to our empirical results under semi-bandit feedback. 

For this set of experiments with $n\geq 2$, note that only the uniformly random exploration policy $\rho^{\text{rand}}$ satisfies Assumption 2, but the exploration policy $\rho^{\text{coalitionSize}}$ does not. 

Figure \ref{fig:mixed effects bandit} reports the mean approximate duality gap $\widehat{\gap}^\delta(\bm{\varphi}^{\out})$ of the strategy $\bm{\varphi}^{\out}$ produced by Algorithm 1 versus the size of datasets generated by $\rho^{\text{rand}}$ over $5$ runs with different seeds. Further, Figure \ref{fig:mixed effects2 bandit} reports the mean approximate duality gap $\widehat{\gap}^\delta(\bm{\varphi}^{\out})$ of the strategy $\bm{\varphi}^{\out}$ produced by Algorithm 1 versus the size of datasets generated by $\rho^{\text{coalitionSize}}$ over $5$ runs with different seeds. By Lemma \ref{supp:lemma:surrogate}, $\widehat{\gap}^\delta(\bm{\varphi}^{\out})$ upper bounds the \textit{true} duality gap with high probability, thus quantifying how close the learned strategy is to Nash stability. In both figures, we examine the effect of varying the number of agents $n \in \{3,4,5,6\}$ while fixing $k=5$, while evaluating our algorithm on datasets of sizes $M \in \{100, 2000, 4000, 6000, 8000, 10000, 20000, 30000\}$ for each fixed number of agents.

Unlike Figures \ref{fig:mixed effects}--\ref{fig:mixed effects2}, Algorithm 1 consistently reaches a low approximation to Nash stability under $\rho^{\text{rand}}$, but fails to do so under $\rho^{\text{coalitionSize}}$ as noted in Figure \ref{fig:mixed effects2 bandit}. This supports the practical relevance of Assumption 1 in the main text: $\rho^{\text{rand}}$ satisfies it, allowing effective learning, whereas $\rho^{\text{coalitionSize}}$ may violate it, yielding poorer performance. For $\rho^{\text{rand}}$, the scaling in $n,k,M$ also matches Theorem \ref{supp:thm:bandit}, Corollary~\ref{coro:optimal} and Remark 6 in the main text.

\begin{figure}[hp!]
    \begin{tabular}{cc}
        \begin{subfigure}[b]{0.5\textwidth}
            \centering
            \includegraphics[width=\linewidth]{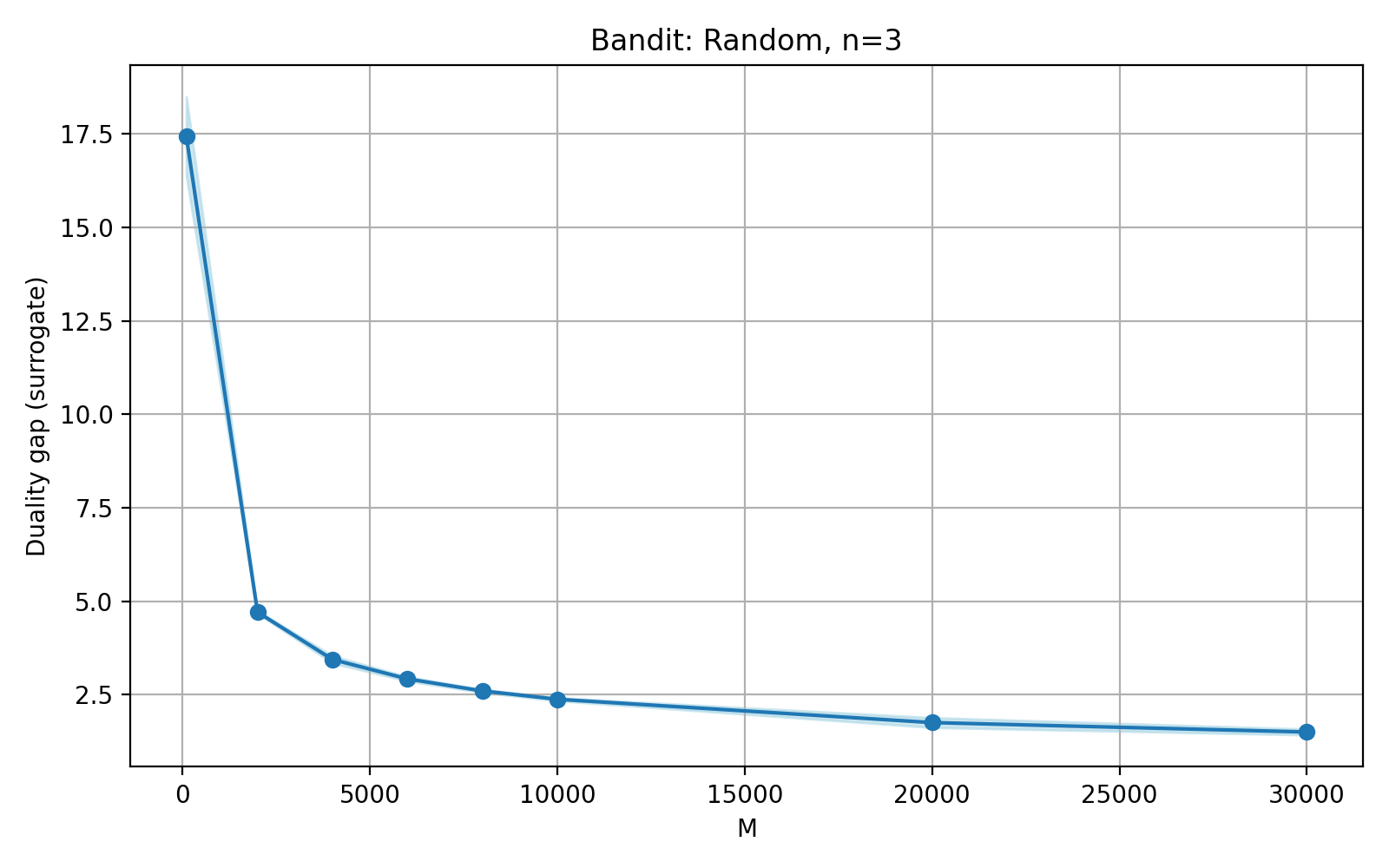}
        \end{subfigure}
        \begin{subfigure}[b]{0.5\textwidth}
            \centering
            \includegraphics[width=\linewidth]{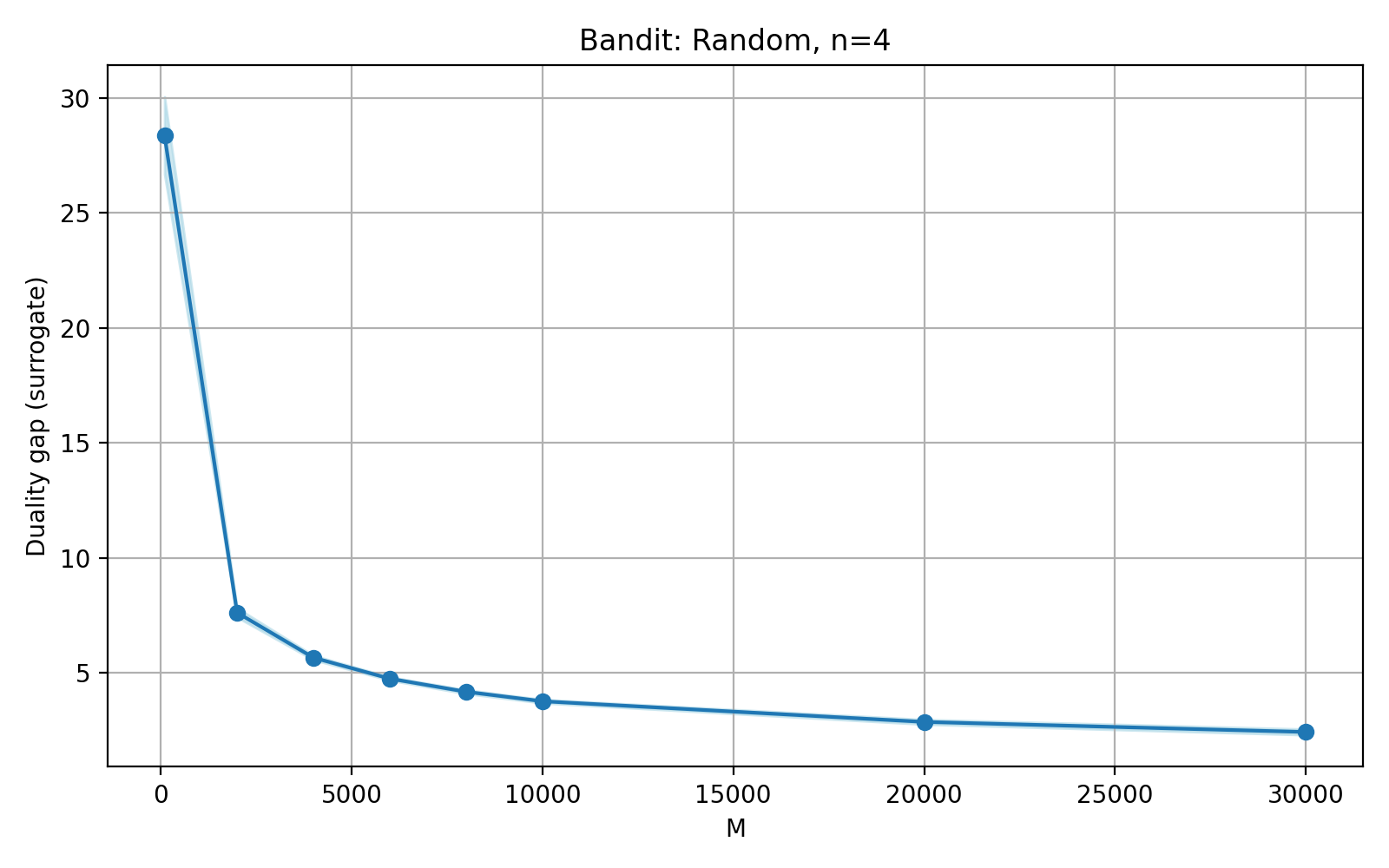}
        \end{subfigure}
        \\
        \begin{subfigure}[b]{0.5\textwidth}
            \centering
            \includegraphics[width=\linewidth]{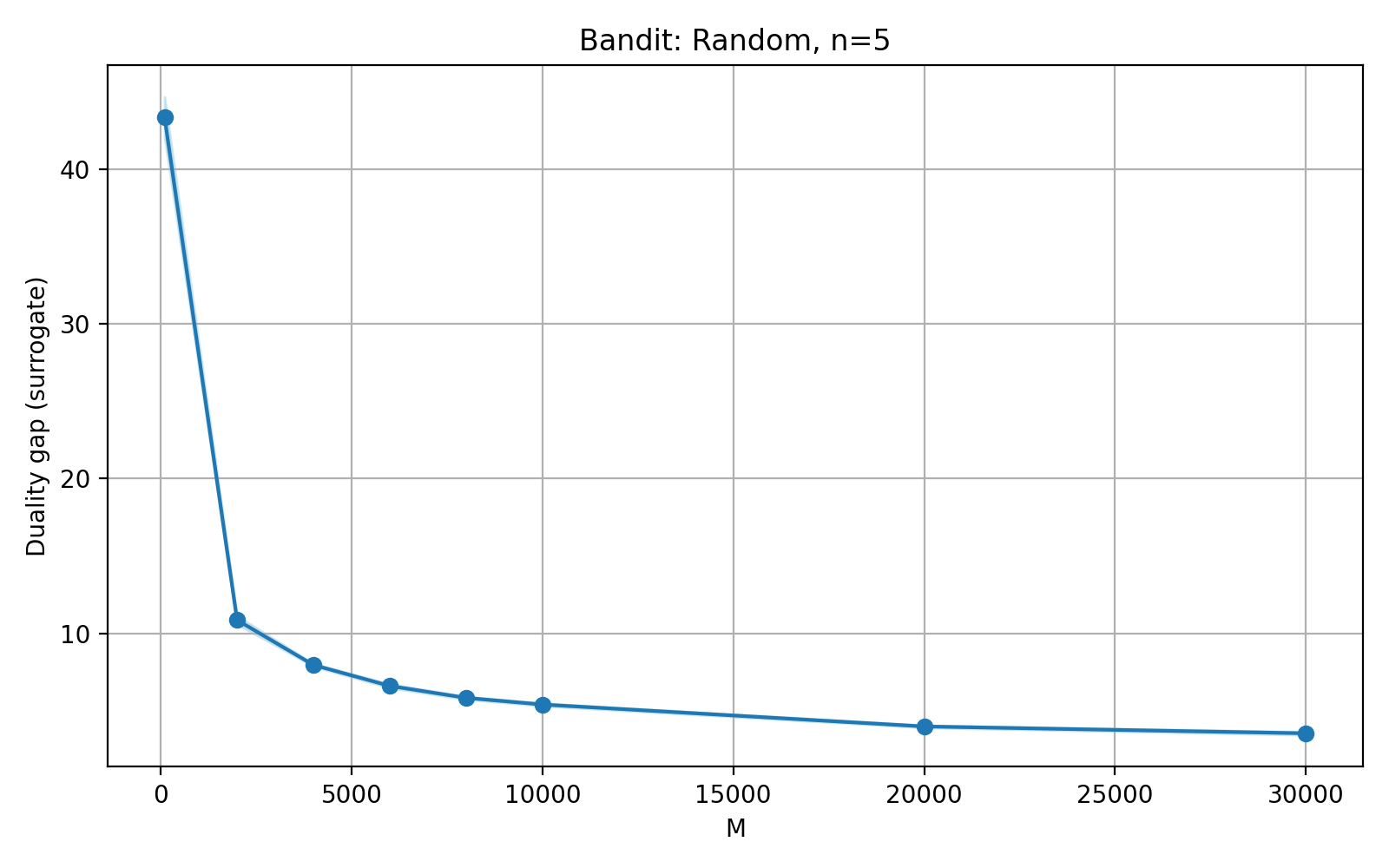}
        \end{subfigure}
        \begin{subfigure}[b]{0.5\textwidth}
            \centering
            \includegraphics[width=\linewidth]{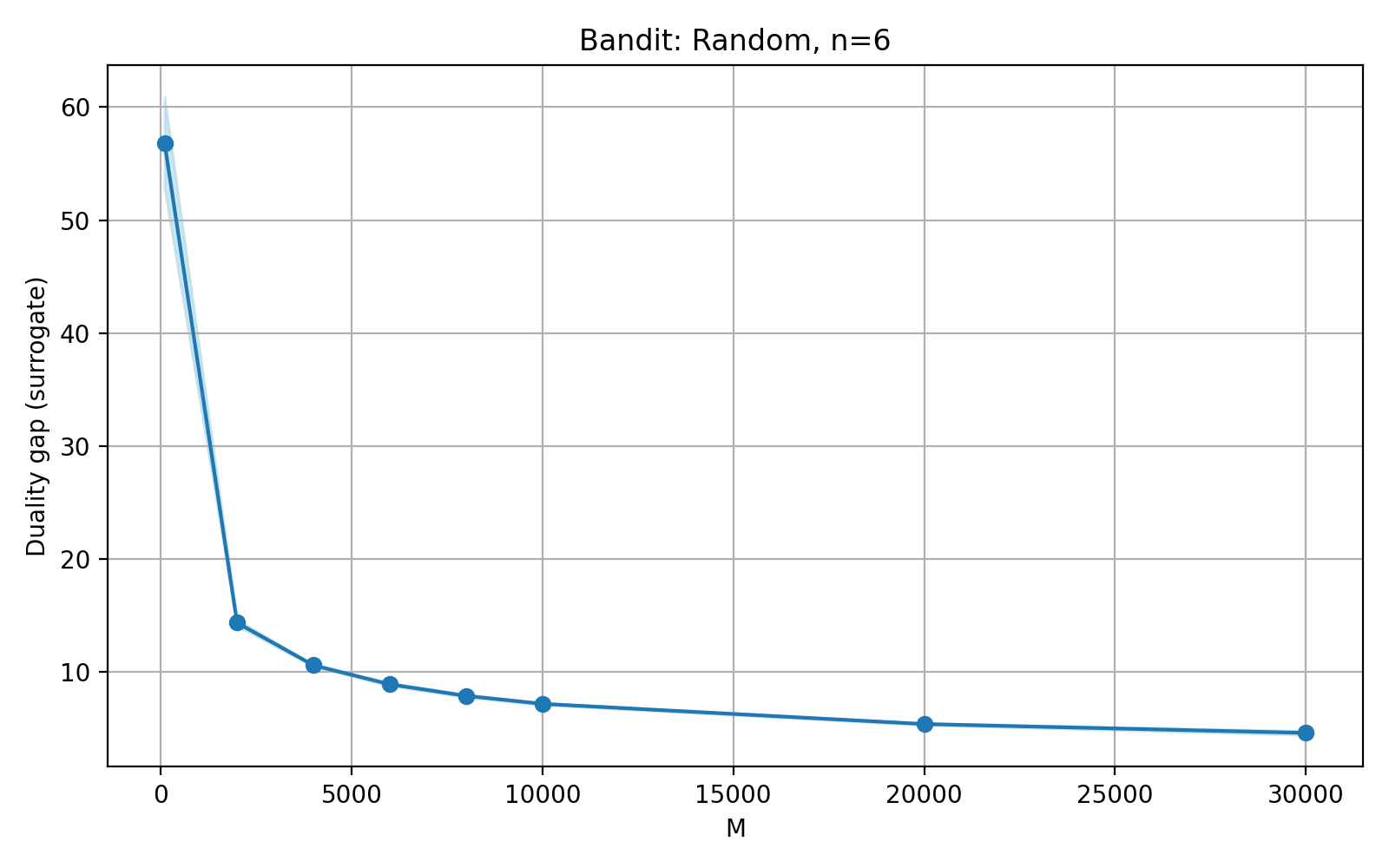}
        \end{subfigure}
    \end{tabular}
    \caption{{Mean approximate duality gap versus the size of datasets generated by $\rho^{\text{rand}}$ over $5$ runs with different seeds for the utility generation model with mixed-coalition-size effects under \textit{bandit} feedback. Here, the number of agents is varied over $n \in \{3,4,5,6\}$. Shaded regions indicate standard deviations.}}
    \label{fig:mixed effects bandit}
\end{figure} 

\begin{figure}[hp!]
    \begin{tabular}{cc}
        \begin{subfigure}[b]{0.5\textwidth}
            \centering
            \includegraphics[width=\linewidth]{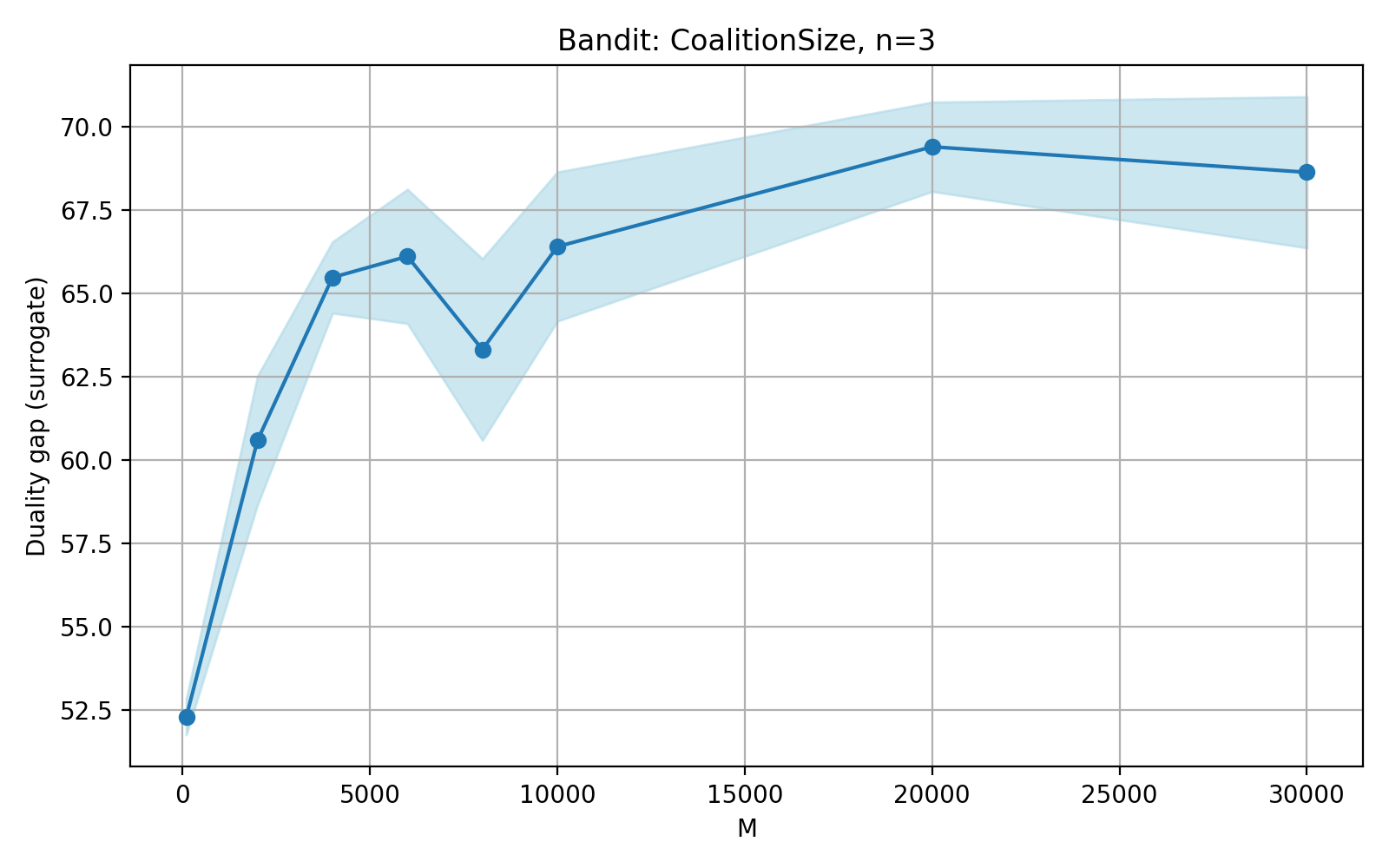}
        \end{subfigure}
        \begin{subfigure}[b]{0.5\textwidth}
            \centering
            \includegraphics[width=\linewidth]{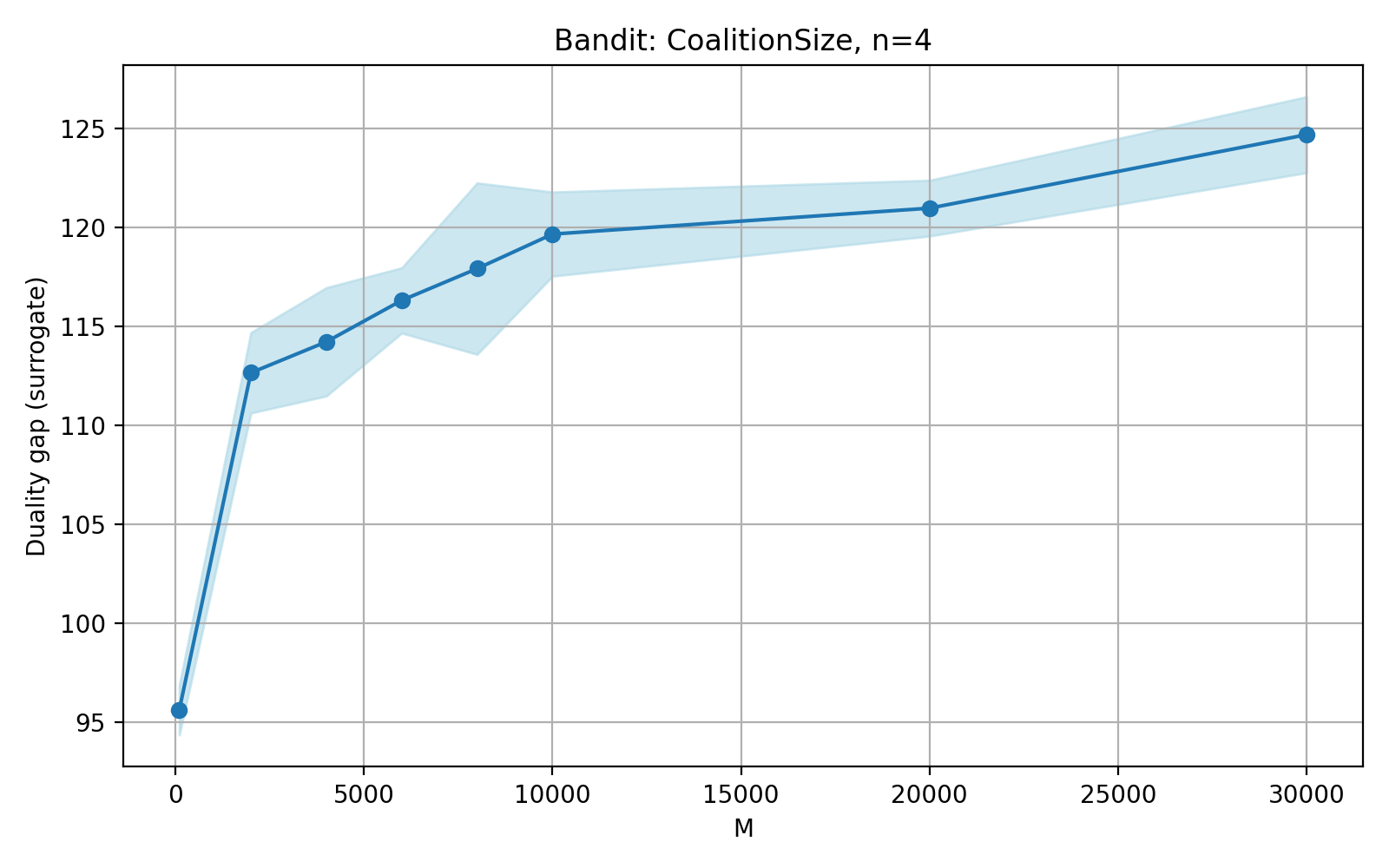}
        \end{subfigure}
        \\
        \begin{subfigure}[b]{0.5\textwidth}
            \centering
            \includegraphics[width=\linewidth]{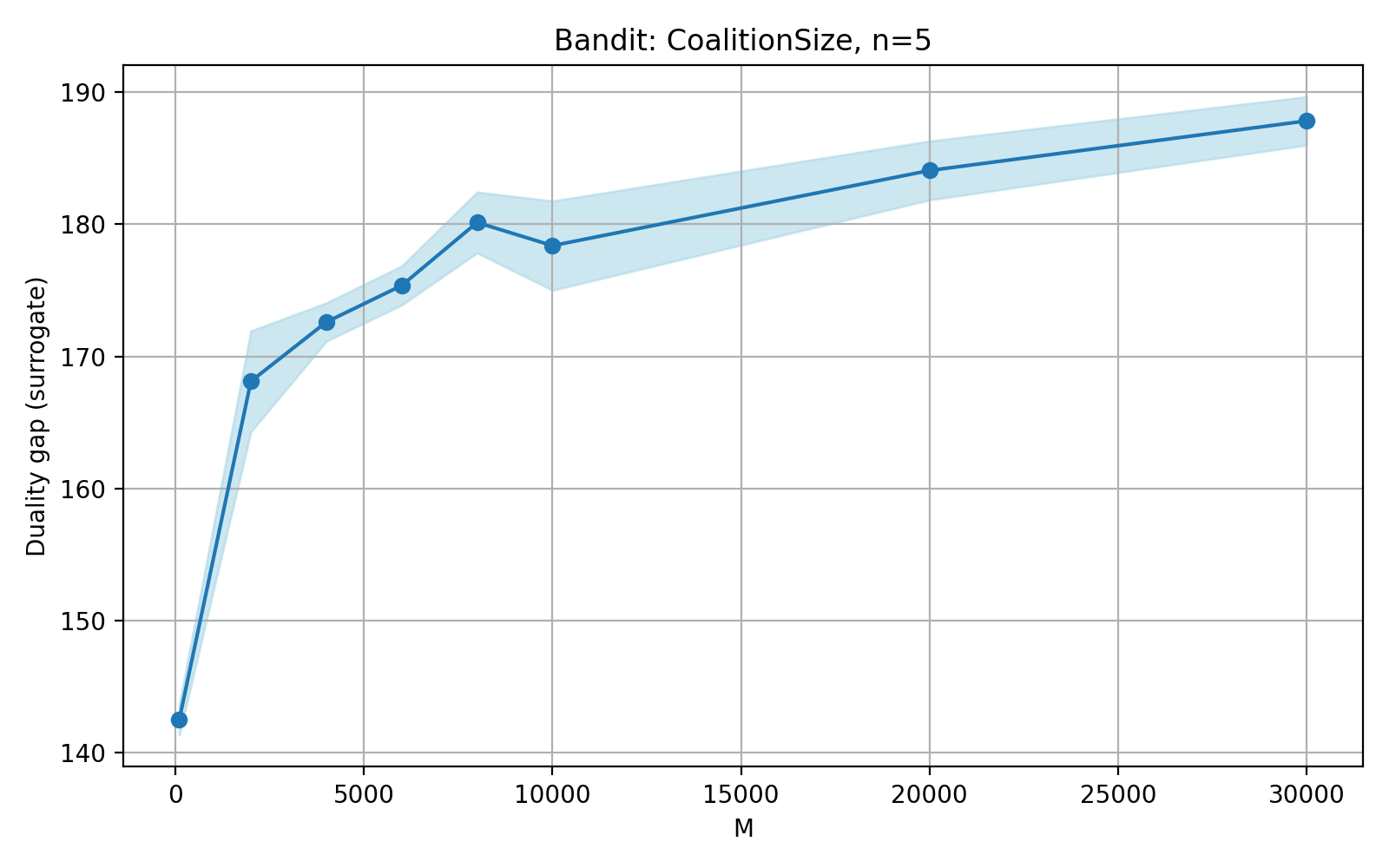}
        \end{subfigure}
        \begin{subfigure}[b]{0.5\textwidth}
            \centering
            \includegraphics[width=\linewidth]{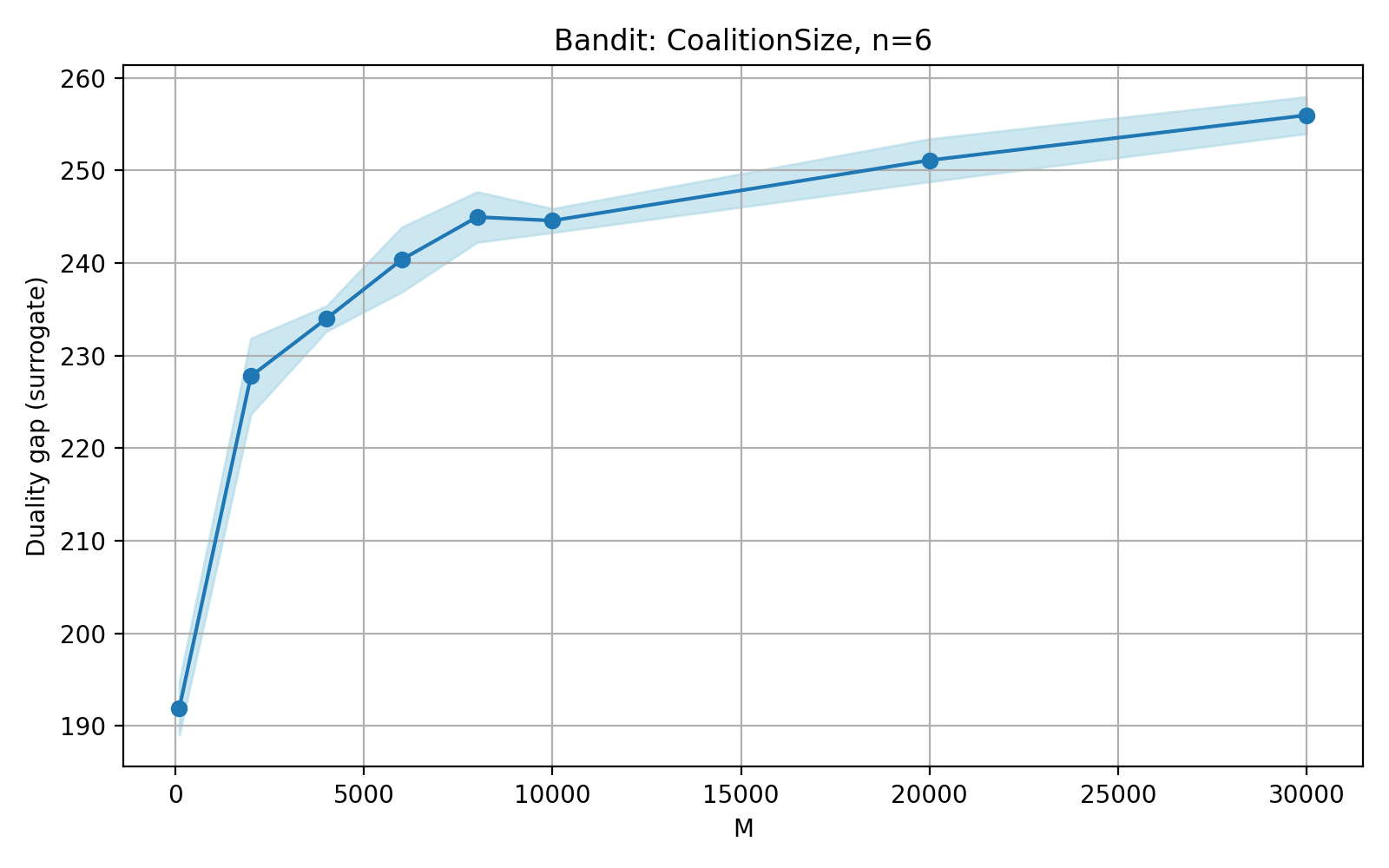}
        \end{subfigure}
    \end{tabular}
    \caption{{Mean approximate duality gap versus the size of datasets generated by $\rho^{\text{coalitionSize}}$ over $5$ runs with different seeds for the utility generation model with mixed-coalition-size effects under \textit{bandit} feedback. Here, the number of agents is varied over $n \in \{3,4,5,6\}$. Shaded regions indicate standard deviations.}}
    \label{fig:mixed effects2 bandit}
\end{figure}

\bibliographystyle{plainnat} 
\bibliography{example_paper}

\end{document}